\documentclass[conference]{IEEEtran}  
\IEEEoverridecommandlockouts % to undo lockout of the thanks command

\usepackage{amssymb,amsmath,wasysym, stmaryrd, mathrsfs, bbm}
\usepackage{mathtools}
\usepackage{xspace}
\usepackage{tikz}
\usetikzlibrary{cd,fit,calc,positioning,arrows}
\usepackage{ifthen}
\usepackage[notref,notcite,final]{showkeys} %% for Lutz
\usepackage[noadjust]{cite}
\usepackage{multicol}
\usepackage{bbm} %blackboard bold numbers
\usepackage{dsfont}
\usepackage{framed}

\usepackage{etex,etoolbox}

\usepackage{amsthm,thmtools}

\newcommand{\ownthmSpaceAbove}{5pt}
\newcommand{\ownthmSpaceBelow}{5pt}
\newcommand{\resetCurThmBraces}{%
\gdef\curThmBraceOpen{(}%
\gdef\curThmBraceClose{)}}
\resetCurThmBraces
\newcommand{\removeThmBraces}{%
\gdef\curThmBraceOpen{}%
\gdef\curThmBraceClose{}}

\declaretheoremstyle[
    spaceabove=\ownthmSpaceAbove,
    spacebelow=\ownthmSpaceBelow,
    headpunct=.,
    postheadspace=.5em,
    notebraces={\curThmBraceOpen}{\curThmBraceClose},
    postheadhook={\resetCurThmBraces},
]{definition}
\declaretheoremstyle[
    style=definition,
    bodyfont=\itshape,
    notebraces={\curThmBraceOpen}{\curThmBraceClose},
    postheadhook={\resetCurThmBraces},
]{theorem}

\usepackage[footnote,marginclue,nomargin,final]{fixme}

\usepackage{hyperref}
\hypersetup{hidelinks,final,bookmarks}

  \renewcommand{\subsectionautorefname}{Section}%
\usepackage{seqsplit}
\usepackage{xstring}
\newcommand{\defaultshowkeysformat}[1]{%
\StrSubstitute{#1}{ }{\textvisiblespace}[\TEMP]%
\parbox[t]{\marginparwidth}{\raggedright\normalfont\small\ttfamily\(\{\){\color{red!50!black}\expandafter\seqsplit\expandafter{\TEMP}}\(\}\)}%
}

\renewcommand*\showkeyslabelformat[1]{%
\noexpandarg%
\defaultshowkeysformat{#1}%
}

\tikzset{                                    
  symbol/.style={
    draw=none,
    every to/.append style={
      edge node={node [sloped, allow upside down, auto=false]{$#1$}}}
  }
}

\usepackage[inline]{enumitem}
\setlist[enumerate,1]{label=(\arabic*),font=\normalfont,align=left,leftmargin=0pt,labelindent=0pt,listparindent=\parindent,labelwidth=0pt,itemindent=!,topsep=3pt,parsep=0pt,itemsep=3pt,start=1}
\setlist[enumerate,2]{label=(\alph*),font=\normalfont,labelindent=*,leftmargin=*,start=1}
\setlist[itemize]{labelindent=*,leftmargin=*,topsep=5pt,itemsep=3pt}
\setlist[description]{labelindent=*,leftmargin=*,itemindent=-1 em}

\newcommand{\takeout}[1]{\empty}

\theoremstyle{theorem}
\newtheorem{theorem}{Theorem}[section]
\newtheorem{corollary}[theorem]{Corollary}
\newtheorem{lemma}[theorem]{Lemma}
\newtheorem{proposition}[theorem]{Proposition}

\theoremstyle{definition}
\newtheorem{defn}[theorem]{Definition}
\newtheorem{expl}[theorem]{Example}
\newtheorem{rem}[theorem]{Remark}
\newtheorem{nota}[theorem]{Notation}
\newtheorem{ass}[theorem]{Assumption}

\newcommand{\syncTh}{\JSL(\A)/\mathsf{S}}
\newcommand{\Pos}{\mathsf{Pos}}
\newcommand{\Set}{\mathsf{Set}}
\newcommand{\T}{\mathbb{T}}
\newcommand{\M}{\mathbb{M}}
\newcommand{\C}{\mathscr{C}}
\newcommand{\A}{\mathcal{A}}
\newcommand{\FA}{\mathfrak{A}}
\newcommand{\partition}{\mathsf{psums}}
\newcommand{\summands}{\mathsf{smd}}
\newcommand{\arity}{\mathsf{ar}}
\renewcommand{\Im}{\mathrm{Im}}
\newcommand{\JSL}{\mathsf{JSL}}

\newcommand{\Alg}{\mathsf{Alg}}
\newcommand{\conv}{\mathsf{conv}}
\newcommand{\gS}{\textsf{PT}}
\newcommand{\sdist}{\mathcal S}

\newcommand{\down}{\mathop{\downarrow}}
\newcommand{\E}{\mathcal{E}}

\newcommand{\id}{\mathsf{id}}

\newcommand{\ext}[1]{\widehat{#1}}
\newcommand{\ol}[1]{\overline{#1}}
\newcommand{\monoto}{\rightarrowtail}
\newcommand{\epito}{\twoheadrightarrow}
\newcommand{\pow}{\mathscr P}
\newcommand{\Posf}{\Pos_\mathsf{f}}
\newcommand{\ari}{P}
\newcommand{\sub}{\mathsf{sub}}
\newcommand{\subs}{\mathsf{subs}}
\newcommand{\terms}{\mathscr{T}}
\newcommand{\termst}[1]{\terms_{\T,#1}}
\newcommand{\rterms}[1]{\mathsf{T}_{\Sigma,#1}}
\newcommand{\xra}[1]{\xrightarrow{~#1~}}
\newcommand{\N}{\mathds{N}}
\newcommand{\cat}[1]{\mathscr{#1}}
\newcommand{\catA}{\cat A}
\newcommand{\catB}{\cat B}

\title{Behavioural Preorders via Graded Monads}

 \author{\IEEEauthorblockN{%
     Chase Ford, %
     Stefan Milius,
     Lutz Schr\"{o}der}\thanks{Work by the first and the third author supported by the DFG within the Research and Training Group 2475 "Cybercrime and 
		Forensic Computing" (grant number 393541319/GRK2475/1-2019). Work of the second author supported by the DFG under projects MI 717/7-1.}
   \IEEEauthorblockA{Friedrich-Alexander-Universit\"{a}t Erlangen-N\"{u}rnberg, Germany}
	 }

\takeout{
\IEEEoverridecommandlockouts % to undo lockout of the thanks command
\IEEEpubid{%
  \makebox[\columnwidth]{978-1-6654-4895-6/21/\$31.00~\copyright2021 IEEE \hfill}%
  \hspace{\columnsep}\makebox[\columnwidth]{ }
}%
}
\begin{document}
\FXRegisterAuthor{cf}{acf}{CF}%Chase
\FXRegisterAuthor{sm}{asm}{SM}%Stefan
\FXRegisterAuthor{ls}{als}{LS}%Lutz

\maketitle              % typeset the title of the contribution

%
% Page numbers (comment out before final submission)
%
%\pagestyle{plain}
%\thispagestyle{plain}

%%%%%%%%%%%%%%%%%%%%%%%%%%%%%%%%%%%%%%%%%%%%%%%%%%
%ABSTRACT
%%%%%%%%%%%%%%%%%%%%%%%%%%%%%%%%%%%%%%%%%%%%%%%%%%
\begin{abstract}%
Like notions of process equivalence, behavioural preorders on processes come in many flavours, ranging from fine-grained comparisons such as ready simulation to coarse-grained ones such as trace inclusion. Often, such behavioural preorders are characterized in terms of theory inclusion in dedicated characteristic logics; e.g.\ simulation is characterized by theory inclusion in the positive fragment of Hennessy-Milner logic. We introduce a unified semantic framework for behavioural preorders and their characteristic logics in which we parametrize the system type as a functor on the category $\Pos$ of partially ordered sets following the paradigm of universal coalgebra, while behavioural preorders are captured as graded monads on~$\Pos$, in generalization of a previous approach to notions of process equivalence. We show that graded monads on~$\Pos$ are induced by a form of graded inequational theories that we introduce here. Moreover, we provide a general notion of modal logic compatible with a given graded behavioural preorder, along with a criterion for expressiveness, in the indicated sense of characterization of the behavioural preorder by theory inclusion. We illustrate our main result on various behavioural preorders on labelled transition systems and probabilistic transition systems.
% We would like to encourage you to list your keywords within
% the abstract section using the \keywords{...} command.
%\keywords{}
\end{abstract}
%
% \begin{IEEEkeywords}
% behavioural preorder, characteristic modal logic, coalgebra, graded monad
% \end{IEEEkeywords}

%
%%%%%%%%%%%%%%%%%%%%%%%%%%%%%%%%%%%%%%%%%%%%%%%%%%
%INTRODUCTION
%%%%%%%%%%%%%%%%%%%%%%%%%%%%%%%%%%%%%%%%%%%%%%%%%%
\section{Introduction}
\noindent Notions of process equivalence, e.g.\ cast as notions of
state equivalence on labelled transition systems, vary on a broad
scale from bisimilarity to trace equivalence, referred to as the
linear-time branching-time spectrum~\cite{Glabbeek90}. Similar
phenomena arise in other system types, e.g.\ in probabilistic
transition systems~\cite{JouSmolka90}, where the spectrum ranges from
probabilistic bisimilarity to probabilistic trace equivalence. In the
present paper, we are concerned with spectra of behavioural
\emph{preorders} (rather than equivalences), in which a process~$A$ is
considered to be `above' a process~$B$ if~$A$ can, in some sense,
match all behaviours of~$B$. Well-known behavioural preorders on
labelled transition systems include simulation and variants thereof
such as ready or complete simulation, as well as various notions of
trace inclusion. 

Previous work~\cite{MPS15,DMS19} has shown that graded monads on the
category $\Set$ of sets provide a useful generic framework that
captures the most important equivalences on the linear-time
branching-time spectrum, as well as the mentioned equivalences on
probabilistic systems. Graded monads thus put process equivalences on
an algebraic footing. In the present paper, we extend this paradigm to
behavioural preorders on processes, using graded monads on the
category $\Pos$ of partial orders. We provide a notion of graded
ordered algebraic theory and show that such theories induce finitary
enriched graded monads on $\Pos$ in the same sense as standard
algebraic theories induce finitary monads on~$\Set$.\cfnote{Preceding
  sentence touched.} We illustrate this framework by showing that it
subsumes well-known notions of simulation and trace inclusion on
labelled transition systems and on probabilistic transition systems.

Both equivalences and preorders on processes can often be
characterized in terms of theory coincidence or theory inclusion in
dedicated modal logics, which are then called \emph{characteristic
logics}. The archetypal result of this type is the classical
Hennessy-Milner theorem, which states that states in finitely
branching labelled transition systems are bisimilar iff they satisfy
the same formulae of Hennessy-Milner logic, i.e.~if the theories of
the two states coincide~\cite{HennessyMilner85}; similarly, a
state~$y$ in a finitely branching labelled transition system simulates
a state~$x$ iff the theory of~$x$ in the positive fragment of
Hennessy-Milner logic is included in that of~$y$. We lift generic
results on characteristic logics for process
equivalences~\mbox{\cite{MPS15,DMS19}} to the setting of behavioural
preorders, obtaining a general criterion that covers numerous
characteristic logics for simulation- and trace-type behavioural
preorders on labelled transition systems and probabilistic transition
systems.

% \paragraph{Related Work} % Look super ugly!
%\smallskip
%\noindent{\bf Related Work.}
\subsection*{Related Work}
\noindent There have been a number of approaches to capturing process
equivalences coalgebraically, e.g.\ based on distributive laws
inducing liftings of functors to Kleisli~\cite{HasuoEA07} or
Eilenberg-Moore~\cite{JacobsEA15} categories of a monad split off from
the type functor, and corecursive algebras~\cite{JacobsEA18}, whose
precise relationship to graded monads has been discussed in previous
work~\cite{MPS15,DMS19}. Roughly speaking, these approaches cover
bisimulation-type equivalences and trace-type equivalences but do not
appear to subsume intermediate equivalences such as simulation
equivalence. Equivalences may be \emph{defined} by distinguishing suitable
modal logics~\cite{KlinRot15}, while we pursue the opposite approach
of designing characteristic logics for given equivalences or
preorders. Variations within trace-type equivalences, such as complete
or ready trace semantics, can be captured via so-called
decorations~\cite{BonchiEA16}. In recent work~\cite{KupkeRot20}, Kupke
and Rot present a highly general fibrational framework for coinductive
predicates and their characteristic logics. The focus on coinductive
predicates implies a certain degree of orthogonality to the
fundamentally inductive framework of graded monads; e.g.\ trace
equivalence is not a coinductive predicate. On the other hand, the
cited framework does cover simulation-type equivalences. The
fibrational setup implies a very high level of generality, while the
framework of graded monads aims at encapsulating as much proof work as
possible at the generic level. There is quite a number of coalgebraic
approaches to notions of simulation that are uniquely induced from
either a functor on $\Pos$ or a set functor equipped with a compatible
ordering, in combination with relation
liftings~\cite{Baltag00,HesselinkThijs00,Worrell00,HughesJacobs04,Cirstea06,FabregasEA09,Lev11,GorinSchroder13};
such coalgebraic notions of simulation typically do capture only
simulation-type preorders, and not, e.g., preorders of trace inclusion
type. At this level of generality, Kapulkin et al.~\cite{KKV12} give a
coalgebraic generic criterion for expressiveness of characteristic
logics for notions of simulation induced from a functor on $\Pos$
(while graded monads allow varying behavioural preorders on a given
system type). This criterion in fact turns out to be a special case of
our main result. Similarly, Baltag~\cite{Baltag00} characterizes
relator-based coalgebraic simulation for weak-pullback-preserving set
functors by preservation of a variant of Moss-style coalgebraic
logic~\cite{Moss99}; Cirstea~\cite{Cirstea06} provides a modularity
mechanism in essentially this setting.

Graded monads originally go back to Smirnov~\cite{Smi08}, and have
subsequently been generalized to indices from a monoidal
category~\cite{Kat14,Mellies17}. Our (graded) ordered algebraic
theories generalize an existing notion of theories with only discrete
arities~\cite{Bloom76}, which correspond to strongly finitary
monads~\cite{KKV12}, as well as a notion of inequational theories for
plain (i.e.~non-graded) monads on
posets~\cite{AdamekEA20}.%\lsnote{Check later mentions}

%
%%%%%%%%%%%%%%%%%%%%%%%%%%%%%%%%%%%%%%%%%%%%%%%%%%
%%%%%%%%%%%%%%%%%%%%%%%%%%%%%%%%%%%%%%%%%%%%%%%%%%
%%%%%%%%%%%%%%%%%%%%%%%%%%%%%%%%%%%%%%%%%%%%%%%%%%
%%%%%%%%%%%%%%%%%%%%%%%%%%%%%%%%%%%%%%%%%%%%%%%%%%
%%%%%%%%%%%%%%%%%%%%%%%%%%%%%%%%%%%%%%%%%%%%%%%%%%
%%%%%%%%%%%%%%%%%%%%%%%%%%%%%%%%%%%%%%%%%%%%%%%%%%
%%%%%%%%%%%%%%%%%%%%%%%%%%%%%%%%%%%%%%%%%%%%%%%%%%
%%%%%%%%%%%%%%%%%%%%%%%%%%%%%%%%%%%%%%%%%%%%%%%%%%
%PRELIMINARIES
%%%%%%%%%%%%%%%%%%%%%%%%%%%%%%%%%%%%%%%%%%%%%%%%%%
\section{Preliminaries}
\noindent We recall some basic definitions and results on partially
ordered sets and \emph{(universal) coalgebra}~\cite{Rut00}.
\subsection{Partially Ordered Sets}\label{sec:posets}
\noindent We refer to partially ordered sets as \emph{posets} as
usual, and write $\Pos$ for the category of posets and monotone
maps. We denote by $1$ the terminal object of $\Pos$, the one-element
poset, and $!\colon X\to 1$ denotes the unique map from a poset $X$
into~$1$. For every poset $X$ we denote by $|X|$ its underlying set
(and sometimes consider this as a discretely ordered poset). The
category~$\Pos$ is Cartesian closed, with the exponential $Y^X$ given
by the hom-set $\Pos(X, Y)$, i.e.\ the set of monotone maps
$f\colon X\rightarrow Y$, ordered point-wise. That is, for
$f,g\colon X\rightarrow Y$, we have $f\leq g$ iff $f(x)\leq g(x)$ for
all $x\in X$. We denote the unordered set of monotone functions
$X\to Y$ by $\Pos_0(X, Y)$, i.e., \;$\Pos_0(X, Y) = |\Pos(X, Y)|$. An
\emph{embedding} is a monotone map $f\colon X\rightarrow Y$ which is
also order-reflecting: $f(x)\leq f(y)$ implies $x\leq y$ (such~$f$ are
injective).  Thus, embeddings are essentially induced
subposets. Categorically, embeddings are precisely the regular
monomorphisms in $\Pos$.

\begin{defn}
  A functor $F\colon\Pos\rightarrow\Pos$ is \emph{locally monotone}
  (or \emph{enriched}) if for all posets $X, Y$ and all $f,g\in\Pos(X, Y)$,
  $f\leq g$ implies $Ff\leq Fg$.
\end{defn}
\noindent% 
\takeout{The category $\Pos$ is  locally finitely presentable~\cite{AR94} 
with finitely presentable objects being precisely the finite posets. 
That is, every poset is the colimit of a directed diagram of finite
posets.}%
Just as the natural numbers $[n]:=\{1,\dots, n\}$ represent
all finite sets up to isomorphism, we fix a set 
\[
\Posf
\]
of finite posets %, which we call \emph{contexts},
representing all finite posets up to isomorphism. Without loss of
generality, we assume that the posets in $\Posf$ are carried by 
initial segments of the natural numbers. However, note that their order
is not in general the usual ordering of the natural numbers. We will
have occasion for the following notions:

\begin{defn}
Let $X$ be a poset.
\begin{enumerate}
\item A subset $S\subseteq X$ is \emph{convex} if for $x, y\in S$ and
  $z\in X$, $x\leq z\leq y$ implies $z\in S$. The \emph{convex hull}
  of a subset $S'\subseteq X$ is the set
  $\conv(S')=\{z\in X\mid \exists x,y\in S'.\,x\le z\le y\}$. A convex
  set $ S\subset X$ is \emph{finitely generated} if there is a finite
  subset $S'\subseteq S$ such that $S=\conv(S')$.

\item A subset $S\subseteq X$ is \emph{downwards closed} if for
  $x\in S$ and $y\in X$, $y\leq x$ implies $y\in S$. The
  \emph{downwards closure} of a set $S'\subseteq X$ is the set
  $\down S'=\{y\in X\mid \exists x\in S'\mid y\le x\}$. A
  downwards closed subset $S\subseteq X$ is \emph{finitely generated}
  if $S=\down S'$ for some finite set~$S'$.
\end{enumerate}
\end{defn}
\subsection{Coalgebras on $\Pos$}\label{S:coalgebra}
\noindent Coalgebras provide a generic notion of state-based
transition system. We briefly recall the general notion before giving a
more specialized treatment of coalgebras on $\Pos$.

\begin{defn} 
  Let $G\colon\C\to\C$ be a functor on a category~$\C$. A
  \emph{$G$-coalgebra} is a pair $(X, \gamma)$ where $X$ is a
  $\C$-object (the \emph{state space}) and $\gamma\colon X\to GX$ is a
  $\C$-morphism (the \emph{transition map}). A \emph{$G$-coalgebra
    morphism} from $(X, \gamma)$ to $(Y, \delta)$ is a $\C$-morphism
  $h\colon X\to Y$ such that $\delta\cdot h=Gh\cdot\gamma$. We denote
  by $\mathsf{Coalg}(G)$ the category of $G$-coalgebras and
  $G$-coalgebra morphisms. We often write $(X,\gamma)$ as
  $\gamma\colon X\to GX$ or just~$\gamma$.
\end{defn}
\noindent We will work exclusively with coalgebras on $\C=\Pos$. For
$G\colon\Pos\to\Pos$, a $G$-coalgebra consists of a poset~$X$ of
\emph{states} and a monotone function $\gamma\colon X\to GX$; we may
think of $\gamma(x)$ as a structured collection of \emph{successors}
of a state~$x\in X$.

\begin{expl}\label{E:convex}
  Let $\A$ be a (discretely ordered) set of \emph{actions}.
  \begin{enumerate}
  \item\label{E:convex:2} Let $\pow^{\down}_{\omega}$ be the
    functor taking each poset to its set of finitely generated
    downwards closed subsets, ordered by set inclusion. The action of
    $\pow^{\down}_{\omega}$ on monotone maps $f\colon X \to Y$ is
    given by taking the downwards closure of the direct image:
    $\pow^{\down}_{\omega}f(S) = \down f[S]$ for
    $S \in \pow^{\down}_{\omega}X$. A coalgebra for
    $\pow^{\down}_{\omega}(\A\times(-))$ is a standard
    finitely branching~\emph{$\A$-labelled transition system}
    (\emph{$\A$-LTS}, for short) that is additionally equipped with a
    partial order which is a simulation. In the case where $|\A|=1$,
    such structures are exactly the ordered transition systems
    that appear in work on well-structured transition systems
    (e.g.~\cite{SS14}). We note that we can turn every $\A$-LTS into a
    coalgebra for $\pow^{\down}_{\omega}(\A\times(-))$ by equipping
    it with the discrete ordering.
    %, so  $\pow^{\down}_{\omega}(\A\times(-))$-coalgebras
    %subsume  standard $\A$-LTS. % SM: Already stated before.
  \item\label{item:convex} Let $C_{\omega}\colon\Pos\rightarrow\Pos$
    be the functor taking each poset to its set of finitely generated
    convex subsets, equipped with the Egli-Milner
    ordering: $A\leq B$ in $C_{\omega}X$ iff
    \[
      \forall a\in A.\exists b\in B. (a\leq b) \land\forall b\in B.\exists a\in A. (a\leq b). 
    \]
    	 The action of $C_{\omega}$ on a monotone map
         $f\colon X\rightarrow Y$ is given by
    \[
      C_\omega f(\conv(s_1,\dots, s_n))=\conv(f(s_1),\dots, f(s_n))
    \]
    (one easily checks well-definedness). A
    $C_{\omega}(\A\times(-))$-coalgebra is thus an $\A$-LTS equipped
    with a partial order that is a bisimulation in the sense of Park
    and Milner~\cite{Mil89, Par81}. Again, these structures subsume
    standard unordered $\A$-LTS.
  \item\label{E:convex3} Let $\sdist$ be the functor taking each poset
    $X$ to its set of finitely supported subdistributions, i.e.,
    functions $\mu\colon X\to [0,1]$ such that
    $\sum_{x\in X}\mu(x)\leq 1$ and the set
    $\mathsf{supp}(\mu):=\{x\in X~|~\mu(x)\neq 0\}$ is finite; we
    describe the ordering of $\sdist(X)$ in~\autoref{E:induced}. An
    $\sdist(\A\times(-))$-coalgebra is a (generative) ordered
    probabilistic transition system with possible deadlock.
   
  \end{enumerate}
\end{expl}
\noindent Just as in the $\Set$-based setting, coalgebras on the category
of posets come with a natural notion of \emph{observational behaviour}
of states. However, due to the additional structure of an ordering of states, 
observational behaviour is a preorder on states in $G$-coalgebras rather 
than an equivalence relation.

\begin{defn}[$G$-simulation]
Let $x\in X, y\in Y$ be states in $G$-coalgebras $\gamma\colon X\to GX,
\delta\colon Y\to GY$. We say that $y$ \emph{$G$-simulates} $x$ if there 
exist $G$-coalgebra morphisms $f\colon (X, \gamma)\rightarrow (Z, \zeta)$,
$g\colon (Y, \delta)\rightarrow (Z, \zeta)$ such that $f(x)\leq g(y)$;
we say that~$x$ and $y$ are \emph{$G$-behaviourally equivalent} if $f,g$
can be chosen so that $f(x)= g(y)$.
\end{defn}
\takeout{   %this has been rewritten above. Kept it here for good measure.
\begin{expl}\label{E:trans}
  Building on \autoref{E:convex}, we now introduce a notion of
  \emph{$\A$-labelled transition system} ($\A$-LTS) in $\Pos$.
  \begin{enumerate}
  \item Fix a set $\A$ of \emph{actions}. By $\A\times(-)$ we denote
    the functor which maps a poset $X$ to the poset $\A\times X$ whose
    underlying partial order is given by the usual product ordering
    with $\A$ viewed as a discretely ordered poset. Thus, for all
    $(a, x),(b,y)\in\A\times X$, we put $(a, x)\leq (b, y)$ if and
    only if $a=b$ and $x\leq y$.

  \item By an \emph{$\A$-LTS}, we understand a coalgebra for the
    functor $C_{\omega}(\A\times(-))$. Such a coalgebra
    $\gamma\colon X\rightarrow C_{\omega}(\A\times X)$ assigns to each
    state $x\in X$ a finitely generated convex set
    $\gamma(x)\in C_{\omega}(\A\times X)$ of pairs $(a, y)$
    representing that $y$ is an $a$-\emph{successor} of $x$.
	
    It is not hard to see that for an $\mathcal{A}$-LTS $(X, \gamma)$,
    states $x, y \in X$ are bisimilar if $x\leq y$.  Indeed, the
    underlying partial ordering on the poset $X$ of states is itself a
    witnessing bisimulation.  In particular, the notions of similarity
    and behavioural equivalence for
    $C_{\omega}(\mathcal{A}\times(-))$-coalgebras coincide.
  \end{enumerate}
\end{expl}
}
\noindent We will be interested in an alternative
notion of $G$-simulation which only takes finite-step behaviours
into account, and is thus better suited to the
inherently inductive nature of the graded behavioural preorders
considered in~\autoref{S:gbp}. 
\begin{defn}[Finite-depth $G$-simulation]\label{D:fdsim}
For a $G$-coalgebra
$\gamma\colon X\rightarrow GX$, we inductively define the sequence
$\gamma_n\colon X\rightarrow G^n1$ of \emph{$n$-step behaviour maps}
by
\[
\gamma_0:=\; !\colon X\to 1; \quad \gamma_{n+1}:= G\gamma_n\cdot\gamma.  
\]
Given $G$-coalgebras $\gamma\colon X\to GX$,
$\delta\colon Y\rightarrow GY$, we say that $y\in Y$
$G$-\emph{simulates} $x\in X$ \emph{at finite depth} if
$\gamma_n(x)\leq\delta_n(y)$ for all $n\in\omega$.
\end{defn}
\noindent 
Recall (e.g.~\cite{AdamekR94}) that an endofunctor is \emph{finitary}
if it preserves filtered colimits. Intuitively, for an
endofunctor on $\Pos$, being finitary means that the functor is
determined by its action on finite posets and monotone maps
between them. All functors in~\autoref{E:convex} fit the following
profile:
\begin{ass}\label{ass:functor}
  We fix from now on a finitary enriched functor $G\colon\Pos\to\Pos$
  that preserves epis and regular monos. As shown by
  Ad\'{a}mek~\cite[Thm.~4.6]{Ada03}, in extension of a well-known
  result by Worrell~\cite{Wor05}, \emph{finite depth $G$-simulation
  coincides with $G$-simulation} under these assumptions.
\end{ass}

\section{Graded Monads and Graded Algebras}\label{sec:graded}

\noindent 
We recall the notion of a graded monad~\cite{Smi08,MPS15,DMS19} with
grades in the monoid $(\N, +, 0)$. %(viewed as a discrete monoidal category). 
%That is, in this paper we will restrict attention to an instantiation of the more 
%general notion of a graded monad with grades in a monoidal category~\cite{Kat14,Mellies17}.
Ultimately, this amounts to restricting to finite-depth behavioural semantics
along the lines of finite-depth $G$-similarity (\autoref{D:fdsim}).

\begin{defn}[Graded Monads]\label{D:gm}
A \emph{graded monad} $\M$ on a category $\C$ 
consists of a family 
$(M_n\colon\mathscr{C}\rightarrow\mathscr{C})_{n<\omega}$
of functors, 
a natural transformation $\eta\colon\id_{\mathscr{C}}\to M_0$
(the \emph{unit}), and a family 
\[
\mu^{n,k}\colon M_nM_k\to M_{n+k} \quad (n,k<\omega)
\]
of natural transformations (the \emph{multiplication}). This data
is subject to \emph{unit laws} stating that the triangles in 
the diagram
\[
  \begin{tikzcd}[column sep=30]
    &
    M_nX
    \arrow[ld, "M_n\eta_X"'] \arrow[rd, "\eta_{M_nX}"]  \arrow[d, "\id"]
    % \ar[equals]{d}
    \\
    M_nM_0X
    \arrow[r, "{\mu^{n,0}}"]
    &
    M_nX
    &
    M_0M_nX \arrow[l, "{\mu^{0,n}}"']
  \end{tikzcd}
\]
\takeout{
\[
\mu^{0,n}\cdot \eta M_n=\id_{M_n}=\mu^{n,0}\cdot M_n\eta
\]
}
commute for all $n<\omega$ and the \emph{associative laws}
\begin{center}
  \begin{tikzcd}[column sep = 40]
    M_nM_kM_m
    \arrow[r, "{M_n\mu^{k,m}}"]
    \arrow[d, "{\mu^{n,k}M_m}"']
    &
    M_{n}M_{k+m} \arrow[d, "{\mu^{n, k+m}}"]
    \\
    M_{n+k}M_m \arrow[r, "{\mu^{n+k,m}}"]
    &
    M_{n+k+m}                               
\end{tikzcd}
\end{center}
for all $n,k,m<\omega$. The graded monad $\M$ is \emph{finitary} if
every~$M_n$ is finitary. When $\C=\Pos$,  $\M$
is \emph{enriched} if every~$M_n$ is locally monotone.
\end{defn}

\begin{expl}\label{E:gradedmonad}
  We recall some basic examples of graded monads for later use in our
  examples of graded behavioural preorders. Further examples will be
  given once the notion of graded theory is in place
  (\autoref{E:induced}).
\begin{enumerate} 
\item\label{E:inducedmonad} The $n$-fold composition $M_n=G^n$ of the
  functor~$G$ defines a graded monad with unit $\eta=\id$ and
  multiplication $\mu^{n,k}= \id_{G^{n+k}}$~\cite{MPS15}.

\item Every enriched monad $(T, \eta, \mu)$ on $\Pos$ defines a graded
  monad with $M_nX:= T(\A^n\times X)$ \cite{MPS15}. Graded monads of
  this type capture trace-type equivalences and preorders; we will
  later see an example capturing probabilistic trace inclusion,
  with~$T$ being a subdistribution monad on~$\Pos$.
  
\item\label{item:Kleisli} Given a functor~$F$ and a monad~$T$
  on~$\Pos$ and a monad-over-functor distributive law
  $\lambda\colon FT\to TF$ -- a so-called \emph{Kleisli law} -- we
  obtain a graded monad~$\M$ with $M_nX=TF^nX$. In fact, the previous
  example is an instance of this construction, being induced by a
  distributive law $A\times T(-)\to T(A\times(-))$~\cite{MPS15}. This
  relates to Kleisli-style coalgebraic trace semantics on set
  coalgebras~\cite{HasuoEA07} but works beyond sets.
\item\label{item:EM} Similarly, given~$F$,~$T$ as above % in \autoref{item:Kleisli}
  and a functor-over-monad distributive law $\lambda\colon TF\to FT$
  -- a so-called \emph{Eilenberg-Moore law} -- we obtain a graded
  monad~$\M$ with $M_nX=F^nTX$~\cite{MPS15}. This relates to
  Eilenberg-Moore-style coalgebraic trace semantics~\cite{JacobsEA15}
  (for which a typical example on the category of sets would have
  $T=\pow$ and $F=X^A$, so $FT$-coalgebras are labelled transition
  systems and $F^nT1\cong(\pow 1)^{A^n}$ consists of
  sets of length-$n$
  traces).  % ; as a typical example, take~$T$ to be
  % one of the several variants of the powerset functor on~$\Pos$
  % (\autoref{sec:posets}), and $FX=2\times X^\A$ where~$2$ is the
  % $2$-chain and~$\A$ a (discretely ordered) set of actions; then
  % $M_nX=$
\end{enumerate}
\end{expl}
\subsection*{Graded Algebras}

\noindent 
Just as in the case of plain monads, graded monads enjoy
both Eilenberg-Moore and Kleisli-style constructions~\cite{MPS15,FKM16}. In
particular, they come with a notion of \emph{graded algebras}~\cite{MPS15}, which
will be of central importance in the semantics of graded logics.

\begin{defn}\label{D:galg}
  Let $n\in\omega$ and let $\mathbb{M}$ be a graded monad on a
  category $\C$. An \emph{$M_n$-algebra}
  $A= ((A_k)_{k\leq n}, (a^{m,k})_{m+k\leq n})$ consists of a family
  $(A_k)_{k\leq n}$  of $\C$-objects  (the \emph{carriers} of~$A$) and a
  family
  \[
    a^{m,k}\colon M_mA_k\rightarrow A_{m+k} \quad (m+k\leq n)
  \]
  of morphisms in $\C$ (the \emph{structure maps} of~$A$) such that
  $a^{0,m}\cdot \eta_{A_m}=\id_{A_m}$ for all $m\leq n$ and, whenever
  $m+k+r\leq n$, the diagram below commutes:
  \begin{equation}\label{D:gradalg}
    \begin{tikzcd}[column sep = 40]
      M_mM_rA_k
      \arrow[r, "{M_ma^{rk}}"]
      \arrow[d, "{\mu^{mr}_{A_k}}"']
      &
      M_mA_{r+k} \arrow[d, "{a^{m, r+k}}"]
      \\
      M_{m+r}A_k \arrow[r, "{a^{m+r, k}}"]
      &
      A_{m+r+k}                           
    \end{tikzcd}
  \end{equation}
  (from now on, we omit commas in superscripts of
  $a,\mu$ when confusion is unlikely).  A \emph{morphism} of
$M_n$-algebras from $A$ to $((B_k)_{k\leq n}, (b^{mk})_{m+k\leq n})$
is a family $f_k\colon A_k\rightarrow B_k$ of $\C$-morphisms  such
that $b^{mk}\cdot M_mf_k=f_{m+k}\cdot a^{mk}$ whenever $m+k\leq n$.
  
$M_{\omega}$-algebras and their morphisms are defined similarly by
allowing indices to range over all $m,r,k\in\omega$.
\end{defn}
\removeThmBraces \noindent It is immediate from the definitions that for all
$l,n,X$, the family $(M_{l+k}X)_{k\le n}$, with the multiplication
maps as structure maps, is an $M_n$-algebra. Indeed,
$((M_kX)_{k\leq n}, (\mu^{mk})_{m+k\leq n})$ is the free $M_n$-algebra
on $X$, w.r.t.\ the forgetful functor which maps an $M_n$-algebra
$((A_m), (a^{mk}))$ to $A_0$~\cite[Prop.~6.2]{MPS15}. The monad
induced by this adjunction is~$M_0$.
\begin{rem}
  Graded monads in the above sense go back to
  Smirnov~\cite{Smi08}. They are a special instance of the more
  general notion of graded monads used by Katsumata~\cite{Kat14},
  Fujii et al.~\cite{FKM16}, and Mellies~\cite{Mellies17}. Their
  notion is that of a \emph{lax monoidal action}
  $\cat M \times \cat C \to \cat C$, where $\cat C$ is a category and
  $\cat M$ a monoidal category (cf.~\cite{JanelidzeK01}). Taking
  $\cat M$ to be the monoid $(\N, + ,0)$, considered as a discrete
  monoidal category, we obtain our notion of graded monad. The notion
  of $M_\omega$-algebras, originally defined for our present notion of
  graded monad~\cite{MPS15}, was later extended to graded
  Eilenberg-Moore algebras in the above-mentioned more general
  setting~\cite[Def.~1]{FKM16}.
\end{rem}

%
%%%%%%%%%%%%%%%%%%%%%%%%%%%%%%%%%%%%%%%%%%%%%%%%%%
%%%%%%%%%%%%%%%%%%%%%%%%%%%%%%%%%%%%%%%%%%%%%%%%%%
%%%%%%%%%%%%%%%%%%%%%%%%%%%%%%%%%%%%%%%%%%%%%%%%%%
%%%%%%%%%%%%%%%%%%%%%%%%%%%%%%%%%%%%%%%%%%%%%%%%%%
%%%%%%%%%%%%%%%%%%%%%%%%%%%%%%%%%%%%%%%%%%%%%%%%%%
%%%%%%%%%%%%%%%%%%%%%%%%%%%%%%%%%%%%%%%%%%%%%%%%%%
%%%%%%%%%%%%%%%%%%%%%%%%%%%%%%%%%%%%%%%%%%%%%%%%%%
%%%%%%%%%%%%%%%%%%%%%%%%%%%%%%%%%%%%%%%%%%%%%%%%%%
%GRADED MONADS
%%%%%%%%%%%%%%%%%%%%%%%%%%%%%%%%%%%%%%%%%%%%%%%%%%
\section{Presentations of Graded Monads on $\Pos$}\label{S:presentations}
\noindent From the perspective of classical (universal) algebra, a
simple notion of ordered algebraic theory arises by just replacing 
equations by inequalities. An ordered algebra is then
understood as a poset $A$ equipped with an $n$-ary monotone function
$A^n\to A$ for each $n$-ary operation symbol. Varieties of ordered
algebras in this sense were introduced by Bloom~\mbox{\cite{Bloom76, 
BW83}}; they correspond precisely to \emph{strongly finitary}
monads on $\Pos$~\cite{KV17} in the sense of Kelly and
Lack~\cite{KL93}.

From a categorical perspective, there is interest in an alternative
notion of ordered theory which captures \emph{all} finitary monads on
$\Pos$, in analogy with the classic correspondence between varieties
of algebras and finitary monads on $\Set$. In joint work with
Ad\'{a}mek~\cite{AdamekEA20}, we have introduced a notion of ordered
algebraic theory specified by operations with finite posets as arities
and inequations that depend on ordered contexts -- an idea tracing
back to work by Kelly and Power~\cite{KP93}. These theories are indeed
in correspondence with \emph{all} finitary monads on $\Pos$.

We introduce a graded extension of the latter notion of ordered
algebraic theories; we show that every graded theory~$\T$ induces a
graded monad $M^{\T}$ on $\Pos$. As a key step in this construction,
we introduce an \emph{inequational logic} for a given graded theory,
which we then use in a free model construction. We illustrate these
devices through examples motivated from the perspective of process
semantics. It happens that our current examples do not mention
properly ordered arities; they do however involve inequalities with
properly ordered contexts.

\subsection{Graded Theories}

\noindent We begin by introducing a notion of graded signature along
with a corresponding notion of graded ordered algebra. The salient
feature of a graded signature is that its operations are equipped with
a \emph{depth}, which should be understood as a measure of how many
steps its operations look ahead in the transition structure. We give
the general definition first but note already now that we will later
restrict attention to the case where all operations have depth
either~$0$ or~$1$.

\begin{defn}
  \begin{enumerate}
  \item A \emph{graded signature} $\Sigma$ consists of a family of
    sets $\Sigma(\ari, n)$ of \emph{depth-$n$ operations of arity
      $\ari$}, where $\ari\in\Posf$ and $n\in\omega$.  Given an
    operation $\sigma\in\Sigma$, we typically denote its arity and
    depth by $\arity(\sigma)$ and $d(\sigma)$, respectively.
	
  \item Let $n\in\omega$. A \emph{$(\Sigma, n)$-algebra} is a
    structure $A=((A_k)_{k\leq n}, (-)^A)$ consisting of a family
    $(A_k)_{k\leq n}$ of posets and a family
    \[
      \sigma^A_k\colon \Pos(\arity(\sigma), A_k)\to A_{d(\sigma)+k},
      \qquad
      \text{$d(\sigma)+k\leq n$}
    \]
    of monotone functions for all $\sigma\in\Sigma$.
    A \emph{homomorphism of $(\Sigma, n)$-algebras} is a family
    $(h_k\colon A_k\to B_k)_{k\leq n}$ of monotone functions such that
    $h_{k+d(\sigma)}\cdot \sigma^A_k = \sigma^B_k\cdot \Pos(\arity(\sigma), h_k)$ for 
    all operations $\sigma\in\Sigma$ and $k + d(\sigma) \leq n$. We denote by
    $\Alg(\Sigma, n)$ the category of $(\Sigma, n)$-algebras and their
    homomorphisms. 
    
    The notions of $(\Sigma, \omega)$-\emph{algebras} and 
    $(\Sigma, \omega)$-\emph{homomorphisms} are defined similarly, only 
    now~$k$ ranges over $\omega$.
  \end{enumerate}
\end{defn}

\noindent Intuitively, the carrier $A_k$ consists of the algebra
elements which have already absorbed operations up to depth~$k$. In
particular, $(\Sigma, n)$-algebras are meant to interpret operations
of accumulated depth at most~$n$. We fix a graded signature~$\Sigma$ for the
remainder of this section. We now introduce a corresponding notion of
\emph{uniform-depth $\Sigma$-terms}, which we inductively generate according
to a similar intuition as described above:

\begin{defn}
  Given a poset $\Gamma$ of \emph{variables} (such posets are called
  \emph{context} in the sequel)\smnote{So \emph{context} includes
    variables, and we do not need to speak of ``context of variables''
    (there are \emph{no} other contexts) in our text.}, the sets
  $\mathsf{T}_{\Sigma, k}(\Gamma)$ of \emph{$\Sigma$-terms of
    (uniform) depth~$k$ (in context $\Gamma$)} are defined recursively
  as follows:
    \begin{itemize}
    \item every variable $x\in \Gamma$ is a term of depth~$0$.
      
    \item Given an operation $\sigma\in\Sigma(\ari, k)$ and a function
      $f\colon |\ari|\to\mathsf{T}_{\Sigma, m}(\Gamma)$, $\sigma(f)$ is
      a term of depth~($k+m)$.
    \end{itemize}
\end{defn}
\begin{rem}\label{R:constants}%\smnote{I inserted this remark.}
  Note that constants, being operations whose arity is the empty
  poset, are terms of every uniform depth $k \geq d(\sigma) $ via the
  unique empty functions
  $f\colon \emptyset \to \mathsf{T}_{\Sigma, m}(\Gamma)$
  $(m \in \omega)$.
\end{rem}

\begin{nota}
  It will sometimes be more convenient to view terms of the shape
  $\sigma(f)$ as expressions $\sigma(t_1,\dots, t_n)$ where
  $f\colon|\arity(\sigma)|\to\mathsf{T}_{\Sigma, m}(\Gamma)$ is the
  assignment $i\mapsto t_i$ for all $i\in\arity(\sigma)$. In this case,
  we simply write $\sigma(t_i)$. We will pass freely between these 
  presentations without further mention.
\end{nota}

\begin{expl}\label{E:terms}
  Consider the graded signature $\Sigma$ consisting of a depth-0
  operation $+$ whose arity is the discrete poset $\{0, 1\}$ and a
  depth-$1$ operation $a$ of arity $\{0\}$. The $\Sigma$-terms
  of depth~$0$ in context~$\Gamma$ are finite sums of variables
  from~$\Gamma$, such as $(x+y)+x$. Terms of depth~$n+1$ are finite
  sums of terms $a(s)$ where~$s$ is a term of depth~$n$. E.g.,
  $x+a(y)$ is \emph{not} a term of any uniform depth since $x$
  and $a(y)$ do not share a common depth.
\end{expl}

\noindent Due to the nature of partially ordered arities, the evaluation 
of $\Sigma$-terms in a $(\Sigma, n)$-algebras is given by a partial map.

\begin{defn}
  Let $A$ be a $(\Sigma, n)$-algebra. Given a context~$\Gamma$ and a
  monotone map $\iota\colon \Gamma\to A_m$ where $m\leq n$, the
  \emph{evaluation map} is the family
\[
\iota^\#_k\colon\mathsf{T}_{\Sigma, k}(\Gamma)\to |A_{m+k}| \quad (m+k\leq n)
\]
of partial maps defined recursively by 
\begin{enumerate}
\item $\iota^\#_0(x)=\iota(x)$ for  $x\in|\Gamma|$, and

\item $\iota^\#_k(\sigma(f))$ is defined for $\sigma\in\Sigma$ and
  $f\colon|\arity(\sigma)|\to\mathsf{T}_{\Sigma, q}(\Gamma)$ such that
  $q+d(\sigma)= k$ (so $\sigma(f)\in\mathsf{T}_{\Sigma, k}$) iff
	\begin{itemize}
	\item all $\iota^\#_q(f(i))$ are defined, and
	\item $i\leq j$ in $\arity(\sigma)$ implies $\iota_q^\#(f(i))\leq\iota_q^\#(f(j))$
		in $A_{q}$;
        \end{itemize}
        in this case, $\iota_k^\#(\sigma(f))=\sigma^A_q(\iota_q^\#\cdot f)$.
\end{enumerate}
\end{defn}

\begin{defn}\label{D:theory}
\begin{enumerate}
\item Let $k\in\omega$. A \emph{depth-$k$ $\Sigma$-inequation (in
    context $\Gamma$)} is a pair $(s, t)$ of depth-$k$ $\Sigma$-terms
  in context $\Gamma$, denoted by $\Gamma\vdash_k s\leq t$.  When the
  depth of $(s, t)$ is not important, we also speak of
  ($\Sigma$-)\emph{inequations}. When~$\Gamma$ is discrete, we often
  omit its mention and just write $s\leq t$. We write
  $\Gamma\vdash_k s=t$ for the conjunction of the inequations
  $\Gamma\vdash_k s\leq t$ and $\Gamma\vdash_k t\leq s$.

\item 
   A $(\Sigma, n)$-algebra~$A$ \emph{satisfies} a
  $\Sigma$-inequation  $\Gamma\vdash_k s\leq t$ if for every
  monotone map $\iota\colon\Gamma\to A_m$ such that
  $m+k\leq n$, both $\iota_k^\#(s)$ and $\iota_k^\#(t)$ are defined
  and $\iota_k^\#(s)\leq\iota_k^\#(t)$.

\item 
  A \emph{graded theory} $\T=(\Sigma,\E)$ consists of a graded
  signature~$\Sigma$ and a set $\E$ of $\Sigma$-inequations, 
  the \emph{axioms} of $\T$.
  
\item A $(\T, n)$-\emph{model}, for $n \leq \omega$, is a
  $(\Sigma, n)$-algebra that satisfies every axiom in~$\E$.
  The \emph{category of $(\T,n)$-models} is the full subcategory of
  $\Alg(\Sigma,n)$ given by all $(T,n)$-models.\smnote{And this is 
  \emph{all} we need to say; no special notation is needed!}
\end{enumerate}
\end{defn}

\noindent 
In particular, a $(\Sigma, n)$-algebra $A$ satisfies
$\Gamma\vdash_k t\leq t$ iff $\iota^\#(t)$ is defined for every
monotone map $\iota\colon\Gamma\to A_m$ such that
$m+k\leq n$. This motivates the following notation:

\begin{nota}
We abbreviate $\Gamma\vdash_k t\leq t$ by 
$\Gamma\vdash_k \down  t$. 
\end{nota}

\begin{rem}
  Let $\Sigma$ be a signature such that every operation has depth
  $0$. Then a $(\Sigma, 0)$-algebra is precisely an ordered algebra in
  the sense of Ad\'{a}mek et al.~\cite{AdamekEA20}; in their
  terminology, a variety of coherent $\Sigma$-algebras is the category
  of $(\T,0)$-models for a graded theory $\T(\Sigma, \E)$, where $\E$
  is a set of (necessarily depth-0) $\Sigma$-inequations. These
  categories of algebras are in bijective correspondence with finitary
  enriched monads $T\colon\Pos\to\Pos$~\cite[Prop.~4.6]{AdamekEA20}.
\end{rem}

\noindent In the examples, we restrict our attention to a class of
graded theories in which every operation and axiom has depth at most
1; such theories, called \emph{\mbox{depth-1} graded theories}, were
previously identified~\cite{MPS15} as a key ingredient in the
framework of graded logics for behavioural equivalences of coalgebra
for set functors. We return to this point in~\autoref{S:depth1}.

\begin{expl}\label{E:theory}
  Let $\A$ be a set of \emph{actions}. We give selected examples of
  graded theories for nondeterministic and probabilistic $\A$-labelled
  transition systems; we describe the arising graded monads
  in~\autoref{E:induced}, and the ensuing behavioural preorders (some
  of which in fact have the character of equivalences) in
  \autoref{S:gbp}.
  \begin{enumerate}
\item\label{item:theory-jsl} The graded theory $\JSL(\A)$ (for
  \emph{join semilattice}) has \mbox{depth-$1$} \emph{choice
    operations} $a_1(-)+\cdots+a_n(-)$ of (discrete) arity $n\ge 0$
  for $a_1,\dots, a_n\in\A$; in case $n=0$, we write the arising
  depth-$1$ constant as~$0$.
  %Where $w=a_1\cdots a_n$, we write $\Sigma_w(-):= a_1(-)+\cdots+a_n(-)$. % ,
  % and we define $0:= \Sigma_{\epsilon}$ where $\epsilon$ is the empty 
  % word. In particular, given a decomposition $w=w_1w_2$ of  
  % $w\in\A^*$ into words $w_1,w_2\in\A^*$, we will write
  % \[
  %   \Sigma_w(\bar{x}, \bar{y})=:\Sigma_{w_1}(\bar{x})+\Sigma_{w_2}(\bar{y}).
  % \]
  The axioms of $\JSL(\A)$ are all depth-1 equations
  \begin{equation*}
    a_1(x_1)+\dots+ a_n(x_n)=b_1(y_1)+\dots+b_k(y_k)
  \end{equation*}
  (coded as pairs of inequations) where
  $\{(a_1,$ $x_1),\dots,$ $(a_m,$ $x_m)\}$ $=\{(b_1,y_1),\dots,(b_k,y_k)\}$.
  % \begin{align*}
  %   \{x, y\}
  %   &
  %   \vdash_1 \Sigma_{a_1a_2}(x, y) = \Sigma_{a_2a_1}(y, x)
  %   \tag{Commutativity}
  %   \\
  %   x
  %   &
  %   \vdash_1 \Sigma_{aa}(x) = \Sigma_a(x)
  %   \tag{Idempotency}
  %   \\
  %   x
  %   &
  %   \vdash_1 0+ \Sigma_a(x) = \Sigma_a(x) \tag{Identity}
  % \end{align*}
  Thus, $\JSL(\A)$ combines nondeterminism and $\A$-labelled actions
  without interaction. We will see that this theory captures
  bisimilarity on $\A$-LTS.

\item\label{item:theory-down}
The graded theory $\JSL^{\down} (\A)$ refines the theory
$\JSL(\A)$ through the addition of the depth-1 axiom
schema
\begin{equation*}
  x\leq y\vdash_1 a(x)+a(y)+t = a(y) +t %\tag{Monotonicity}
\end{equation*}
(with non-discrete context) where $a\in\A$ and~$t$ represents a
remaining formal sum $b_1(z_1)+\dots+b_k(z_k)$, $k\ge 0$. It
effectively specifies that the actions are monotone w.r.t.\ the
ordering induced by the join semilattice structure~$+$, equivalently
that this ordering coincides with the one that comes from living over
posets; we will see that this property relates to similarity.

\item\label{item:theory-sync} The graded theory $\syncTh$ (with
  $\mathsf{S}$ standing for \emph{synchrony}) extends $\JSL(\A)$ by
  giving the constant~$0$ depth~$0$ (instead of depth~$1$) and adding
  the depth-1 axiom scheme
  \begin{equation*}
    a(0)+t= t %\tag{$\mathsf S$}
  \end{equation*}
  where~$t$ represents a remaining formal sum
  $b_1(z_1)+\dots+b_k(z_k)$, $k\ge 0$. That is, deadlock is preserved
  by the actions. We will see that this theory relates to a notion of
  process equivalence between bisimilarity and trace equivalence, in
  which processes are identified if they have the same uniform-depth
  finite computation trees.

\item\label{item:theory-gsubconvex}
\label{item:theory-subconvex} Additionally letting actions
  distribute over joins relates to trace-type equivalences and
  preorders~\cite{DMS19}. We skip the nondeterministic version and
  instead look at probabilistic traces. We recall the (plain)
  algebraic theory of subconvex algebras~\cite[Definition
  2.7]{PumplunRohrl84} (also known as positive convex
  modules~\cite{Pumpluen03}); we will use this theory to capture
  finite probability subdistributions (which are defined like finite
  distributions except the global mass is only required to be at
  most~$1$ rather than exactly~$1$). Its operations are formal
  \emph{subconvex combinations} $\sum_{i=1}^np_i\cdot(-)$ for
  $\sum p_i\le 1$, and its equational axioms reflect the laws of
  (plain) monad algebras: The equation
  $\sum \delta_{ik}\cdot x_k=x_i$, with $\delta_{ik}$ being Kronecker
  delta ($\delta_{ik}=1$ if $i=k$, and $\delta_{ik}=0$ otherwise),
  reflects compatibility with the unit, and the equation scheme
\[\textstyle
  \sum_{i=1}^n p_i\cdot\sum_{k=1}^m q_{ik}\cdot x
  _k=\sum_{k=1}^m\big(\sum_{i=1}^n p_i q_{ik})\cdot x_k
\]
reflects compatibility with the monad multiplication. We form the
ordered algebraic theory $\mathbb{S}$ of ordered subconvex algebras 
by further imposing inequations
\[\textstyle
  % \{x_1,\dots, x_n\}\vdash %% SM: discrete context omitted
  \sum_{i=1}^n p_i\cdot x_i\leq \sum_{i=1}^n
  q_i\cdot x_i\quad (\text{$p_i\leq q_i$ for all $i\leq n$}).
\]
% \lsnote{Removed the monotonicity axiom as we're in the version where
% everything is monotone anyway.}
% and inequations of the form
% \[\textstyle
%   \{x_i\leq y_i\mid i\leq n\}
%   \vdash
%   \sum_{i=1}^n p_i\cdot x_i \leq \sum_{i=1}^n p_i\cdot y_i,
% \]
% specifying that formal subconvex combinations are monotone
% operations. \smnote{The name `monotonicity 
% (condition)' is used several times later on (in the appendix), 
% and so it should be explained what this refers to.}  
The graded theory $\gS(\A)$ of probabilistic trace inclusion has
subconvex combinations as depth-$0$ operations and actions $a(-)$ as
unary depth-$1$ operations; besides the mentioned inequations of
ordered subconvex algebras (which are now depth-$0$ axioms), we
include depth-1 equations stating that actions distribute over
subconvex combinations, i.e.~for $\sum_{i=1}^k p_i\le 1$ and $a\in\A$,
we postulate depth-$1$ equations
\begin{equation*}\textstyle
  a(\sum_{i=1}^k p_i \cdot x_i) = \sum_{i=1}^kp_i\cdot a(x_i).
\end{equation*}
\end{enumerate}
\end{expl}

\subsection{Logic of Inequations in Context}\label{S:LoIC}

\noindent Let $\T$ be a graded $\Sigma$-theory with axiom set $\E$. We
introduce a \emph{logic of inequations in context} for $\T$: we give a
system of rules for deriving inequations in context which is both
sound (\autoref{T:sound}) and complete (\autoref{T:complete}) for the
semantics of graded $\Sigma$-terms in $(\T, n)$-models. We first fix 
some syntactic notions:

\begin{defn}
\begin{enumerate}
\item A \emph{uniform substitution} is a function
  $\gamma\colon|\Delta|\rightarrow\mathsf{T}_{\Sigma, k}(\Gamma)$
  where $k\in\omega$ and $\Gamma, \Delta$ are contexts.
  %; when $\Sigma$
  %is fixed we denote the set of uniform substitutions
  %$\Delta\rightarrow\mathsf{T}_{\Sigma, k}(\Gamma)$ by
  %$\subs_k(\Delta, \Gamma)$.
	
  \item The set $\sub(t)$ of \emph{subterms of
      $t\in\mathsf{T}_{\Sigma, k}(\Gamma)$} is defined recursively in the expected way:
    \begin{align*}
      \sub(x)
      &
      =\{x\}\text{ for all }x\in\Gamma;
      \\
      \sub(\sigma(f))
      &\textstyle
      =\{\sigma(f)\}\cup\bigcup_{i\in\arity(\sigma)}\sub(f(i)).
    \end{align*}
    For terms $s, t\in \mathsf{T}_{\Sigma, k}(\Gamma)$, we
    write $\sub(s,t)$ for $\sub(s) \cup \sub(t)$.
  \end{enumerate}
\end{defn}

\noindent The requirement that uniform substitutions map all variables
to terms of the same uniform depth ensures that they can be
meaningfully applied to uniform-depth terms: Applying
$\gamma\colon|\Delta|\to\mathsf{T}_{\Sigma, k}(\Gamma)$ to a
term~$t\in\mathsf{T}_{\Sigma, p}(\Delta)$ results in a term
$\ol\gamma(t)\in\mathsf{T}_{\Sigma, k+p}(\Gamma)$, recursively defined
by
\begin{enumerate}
\item $\ol\gamma(x) = \gamma(x) \in \mathsf{T}_{\Sigma,k}(\Gamma)$
  for $x \in \Delta$,
  
\item  $\ol\gamma(\sigma(f)) = \sigma(\overline{\gamma} \cdot f) \in
  \mathsf{T}_{\Sigma,k+m+n}(\Gamma)$ for $\sigma \in \Sigma(\ari, m)$
  and  $f\colon |P| \to \mathsf{T}_{\Sigma,n}(\Gamma)$.
\end{enumerate}
\takeout{%% SM: reformulated
\begin{align*}
  \ol\gamma(x) 
  &:= \gamma(x) \in \mathsf{T}_{\Sigma,k}(\Gamma) &
  \text{for $x \in \Delta$},
  \\
  \overline{\gamma}(\sigma(t_1,\dots, t_n))
  & :=
  \sigma(\overline{\gamma}(t_1),\dots, \overline{\gamma}_p(t_n)),
\end{align*}
where $d(\sigma) + p = m$.}% end takeout
%
% \smnote{I think this is not helpful at all; on the contrary it obscures proof obligations
% in the proof \autoref{P:free}(1).
%
%Abusing notation, we will simply denote
%the extension $\overline{\gamma}$ by $\gamma$.

\begin{defn}\label{D:Tterms}
  The \emph{logic of inequations in context} is given by the set of
  rules shown in \autoref{fig:rules}. In the notation,
  $\Gamma, \Delta$ are contexts. The side conditions $(+)$ for
  ($\mathsf{Ar}$) and ($\mathsf{Mon}$) state that
  $f,g\colon |\arity(\sigma)|\to\mathsf{T}_{\Sigma, k}(\Gamma)$.
  Condition~$(*)$ states that $\Delta\vdash_n s\leq t\in\E$ is an
  axiom and $\gamma\colon |\Delta| \to\mathsf{T}_{\Sigma,k}(\Gamma)$ a
  uniform substitution. Condition $(**)$ is as follows:
  $\gamma\colon |\Delta|
  \to\mathsf{T}_{\Sigma,k}(\Gamma)$ %\in\subs_k(\Delta, \Gamma)$
  is a uniform substitution and for some axiom $\Delta\vdash_n s\leq
  t\in\E$ there is a subterm of the form
  $\sigma(f)\in\sub(s,t)$, where
  $f\colon\arity(\sigma)\to\mathsf{T}_{\Sigma,m}(\Delta)$, such that
  $u=f(i)$ and $v=f(j)$ for some $i\leq j$ in $\arity(\sigma)$.
\end{defn}

\begin{figure}
\begin{framed} %% SM: There is not enought space for putting a frame.
\begin{gather*}
  (\mathsf{Var})
  \;
  \frac{}{\Gamma\vdash_0 x\leq y}
  \;
  \text{$x\leq y$ in $\Gamma$} % the logical structure here is a text
                               % with formulas with a big formula and
                               % that is how it should be typeset
  \\[1.6mm] %the optional argument \\[length] controls vertical spacing
  (\mathsf{Ar})
  \;
  \frac{\{\Gamma\vdash_k f(i)\leq f(j)\mid i\leq j\in
    {\arity(\sigma)}\}}{\Gamma\vdash_{k+d(\sigma)}\down \sigma(f)}
  \;
  (+)
  %f\colon |\arity(\sigma)|\to\mathsf{T}_{\Sigma, k}(\Gamma)
  \\[1.6mm]
  (\mathsf{Trans})
  \;
  \frac{\Gamma\vdash_k s\leq t \qquad \Gamma\vdash_k t\leq u}{\Gamma\vdash_k s\leq u} 
  \\[1.6mm]
  (\mathsf{Mon})
  \; \frac{%
    \begin{array}{@{}c@{}}
      \{\Gamma\vdash_k f(i)\leq g(i)\mid i\in\arity(\sigma)\}\\
      \Gamma\vdash_{k+d(\sigma)} \down \sigma(f) \qquad \Gamma\vdash_{k+d(\sigma)} \down \sigma(g) 
    \end{array}}
  {\Gamma\vdash_{k+d(\sigma)} \sigma(f)\leq\sigma(g)} \; (+)
  \\[1.6mm] 
  (\mathsf{Ax1})
  \;
  \frac{\{\Gamma\vdash_k\gamma(x)\leq\gamma(y)\mid \Delta\vdash_0
    x\leq y\}}{\Gamma\vdash_{n+k} \ol\gamma(s)\leq \ol\gamma(t)}
  \;(*)
%  \begin{array}{@{\;}c}
%    \Delta\vdash_n s\leq t\in\E, \\
%    \gamma\colon |\Delta| \to\mathsf{T}_{\Sigma,k}(\Gamma) %in\subs_k(\Delta, \Gamma)
%  \end{array}
  \\[1.6mm]
  (\mathsf{Ax2})
  \;
  \frac{\{\Gamma\vdash_k\gamma(x)\leq\gamma(y)\mid \Delta\vdash_0 x\leq y\}}
  {\Gamma\vdash_{m+k} \ol\gamma(u)\leq \ol\gamma(v)} \;(**)
  \end{gather*}
\end{framed}
\caption{System of rules for graded inequational reasoning, see 
	      \autoref{D:Tterms} for details on the side conditions $(+)$, $(*)$, and $(**)$.}
\label{fig:rules}
\end{figure}
\noindent As one might expect, the logic of inequations is more subtle
than equational logic. We briefly discuss the salient rules.  The rule
$(\mathsf{Ar})$ declares that the arities of operation symbols define
their domain of definition (cf.~\autoref{C:arities}), while
$(\mathsf{Mon})$ reflects that operations are monotone. The rule
$(\mathsf{Ax1})$ allows introducing axioms, instantiated using uniform
substitions that derivably respect the context inequations. It is
supplemented by the rule
$(\mathsf{Ax2})$ which says that every axiom implicitly includes the
statement that in all occurrences of
operations~$\sigma$ in the axiom, the arguments satisfy the
inequalities in the arity
of~$\sigma$; again, these implicit constraints may be instantiated
using uniform substitions that derivably respect the context.

The combination of the rules $(\mathsf{Ax1})$ and $(\mathsf{Ax2})$ 
features prominently in ensuring that each term occurring in
a given axiom is derivably defined (we make this precise
in~\autoref{P:subterms}). The key is the following lemma:

\begin{lemma}\label{L:subterms}
  Assume that $\Gamma\vdash_k s\leq t$ is derivable and let
  $\sigma(f)\in\sub(s, t)$ where
  $f\colon|\arity(\sigma)|\to \mathsf{T}_{\Sigma, n}(\Gamma)$. Then
  $\Gamma\vdash_n f(i)\leq f(j)$ is derivable for all $i\leq j$ in
  $\arity(\sigma)$.
\end{lemma}

\begin{corollary}\label{C:arities}
  Let $\sigma\in\Sigma$, and let
  $f\colon|\arity(\sigma)|\to\mathsf{T}_{\Sigma, k}(\Gamma).$ Then
  $\Gamma\vdash_{k+d(\sigma)}\down \sigma(f)$ is derivable iff
  $\Gamma\vdash_{k} f(i)\leq f(j)$ is derivable for all $i\leq j$ in
  $\arity(\sigma)$.
\end{corollary}

\begin{proposition}[Subterm Rule]\label{P:subterms}
Assume that $\Gamma\vdash_k s\leq t$ is derivable and let 
$u\in\sub(s, t)$ have depth~$m$. Then $\Gamma\vdash_m\down u$ is 
derivable. In particular, $\Gamma\vdash_k\down s$ and 
$\Gamma\vdash_k\down t$ are derivable.
That is, the following rule is admissible:
\[
  (\mathsf{Sub})
  \frac{\Gamma\vdash_k s \leq t}{\Gamma \vdash_m\down u}
  \; u \in \sub(s,t)\cap\mathsf{T}_{\Sigma, m}(\Gamma)
\]
\end{proposition}

% \begin{expl}\label{E:Tterms}
%   Let $\T$ be the graded $\Sigma$-theory with $\E=\varnothing$ so that
%   the rules $(\mathsf{Ax1})$ and $(\mathsf{Ax2})$ are vacuous. In this
%   case, we have that $\Gamma\vdash_k\down t$ is derivable iff
%   one of the following holds:
%   \begin{enumerate}
%   \item $t$ is a variable in context $\Gamma$ (and $k=0$), or 
%   \item $t=\sigma(t_i)$ (and $k=m+\ell$) for some operation
%     $\sigma\in\Sigma(\Delta, m)$ and some
%     $t_1,\dots, t_n\in\terms_{\T, \ell}(\Gamma)$ such that
%     $\Gamma\vdash_{\ell} t_i\leq t_j$ is derivable for all $i\leq j$ in
%     $\Delta$.
%   \end{enumerate}
% \end{expl}

\subsection{Free Models of Graded Theories}\label{S:freemodel}

\noindent Given a graded theory $\T$, a poset $X$, and 
$n\leq \omega$, we describe a construction of a free $(\T,n)$-model
$FX$ on $X$.

\begin{defn}
  Let $\T$ be a graded theory and let $X$ be a poset. The
  \emph{set $\terms_{\T, k}(X)$ of depth-$k$ $\T$-defined terms in $X$}
  is the set of terms
  $t\in\mathsf{T}_{\Sigma, k}(X)$ such that $X\vdash_k\down t$ 
  is derivable. \emph{Derivable inequality} is the relation $\leq$ on 
  $\terms_{\T, k}(X)$ given by putting $s\leq t$ iff $X\vdash_k s\leq t$ 
  is derivable.
\end{defn}
%
% \begin{expl}%\label{E:Tterms-ctd}
%   %We continue \autoref{E:Tterms}.  reference no longer exists
%   Unwinding the rules for order
%   formation, it is easy to verify that for all
%   $t_1, t_2\in\terms_{\T, k}(X)$ we have $t_1\leq t_1$ if and
%   only if
%   \begin{enumerate}
%   \item $t_1,t_2$ are variables in $X$ such that $t_1\leq t_2$, or
%   \item for some $\ari \in \Posf$ and some $m\in\omega$ there exists 
%     $\sigma\in\Sigma(\ari, m)$ and 
%     $f_1, f_2\in\Pos(\ari, \terms_{\T, k-m}(\Gamma))$ with $f_1 \leq f_2$
%     such that $t_i= \sigma(f_i)$ ($i=1, 2$).
%   \end{enumerate}
% \end{expl}

\noindent % It follows immediately from the subterm rule (\autoref{P:subterms})
% that the relation $\leq$ is defined: if $X\vdash_k s\leq t$ is 
% derivable, then $s,t\in\terms_{\T, k}(X)$. Furthermore,
Derivable inequality is clearly reflexive and transitive, with the latter
following immediately from the rule $(\mathsf{Trans})$. In short, we
have the following:

\begin{proposition}
  The set $(\terms_{\T, k}(X), \leq)$ is a preorder.
\end{proposition}
%
%It is now easy to verify that the logic of inequations in context is sound:
%
\smnote{Do not move the proof needs the previous proposition.}
\begin{theorem}[Soundness]\label{T:sound}
  Let $n\leq\omega$. If an inequation in context is derivable, then it
  is satisfied by every $(\T,n)$-model.
\end{theorem}
%
%\smnote[inline]{TODO: whoever added this should put a proof in the appendix.}
%
\begin{defn}\label{D:freealgebra}
  Let $X$ be a poset. The $(\Sigma, n)$-algebra $FX$ has carriers
  $(FX)_k:=\terms_{\T, k}(X)/\mathord{\sim}$ for $k\leq n$ where
  $\sim$ is the equivalence relation defined by \emph{derivable
    equality}: $s \sim t$ iff $X\vdash_k s\leq t$ and
  $X\vdash_k t\leq s$ are derivable. Now we define, for every~$\sigma$
  in~$\Sigma$, the family of maps $\sigma^{FX}_k$ by
  \begin{align*}
    \sigma^{FX}_k\colon\Pos(\arity(\sigma), (FX)_k)
    &
    \to (FX)_{k+d(\sigma)}  \quad (k+d(\sigma)\leq n)
    \\
    f
    &
    \mapsto [\sigma(u_k\cdot f)]_{k+d(\sigma)},
  \end{align*}
  where $[t]_k$ denotes the equivalence class of $t\in (FX)_k$ and
  $u_k\colon (FX)_k \monoto \terms_{\T,k}(X)$ is some splitting of the
  canonical quotient map $q_k\colon \terms_{\T,k}(X) \epito (FX)_k$;
  that is, $q_k\cdot u_k = \id_{(FX)_k}$. We define
  $\eta_X\colon X\to (FX)_0$ by $\eta_X(x) = [x]_0$.

  For $n= \omega$, we define $FX$ similarly, only now $k$ and
  $k+d(\sigma)$ range over $\omega$.
\end{defn}
\begin{lemma}\label{L:welldefined}
  The maps $\sigma^{FX}_k$ are independent of the choice of $u_k$ and
  monotone.
\end{lemma}
\begin{proposition}\label{P:free}
  For every $n \leq \omega$, the $(\Sigma, n)$-algebra $FX$ is the
  free $(\T, n)$-model on $X$.
\end{proposition}
\noindent
More precisely, $FX$ is a free $(\T,n)$-model with the universal morphism
$\eta_X$; that is, $\eta_X$ is monotone, and for every $(\T, n)$-model
$A$ and every monotone map $h\colon X\to A_0$, there exists a
unique $(\Sigma, n)$-homomorphism $h^\sharp\colon FX\to A$ such that
$h^\sharp_0\cdot\eta_X = h$.
\begin{corollary}\label{C:adj}
  The forgetful functor from the category of $(\T, n)$-models to
  $\Pos$ mapping $(A_k)_{k \leq n}$ to  $A_0$ has a left adjoint,
  which assigns to a poset $X$ the free $(\T, n)$-model. 
  %$FX$. %% Saves one line.
\end{corollary}

\noindent By considering the interpretation of $\Sigma$-terms in the
free $(\T,n)$-model $FX$, we have the following:
\begin{theorem}[Soundness and Completeness]\label{T:complete}
  \mbox{ }\\%\par\noindent
  Let $n\leq\omega$. An inequation in context is derivable iff it is
  satisfied by every $(\T, n)$-model.
\end{theorem}

\subsection{From Graded Theories to Graded Monads}\label{S:theorymonad}

\noindent We will now describe how every graded theory $\T$ induces a graded
monad $M^\T$ on $\Pos$. Fujii et
al.~\cite[Section~3]{FKM16} explain%
\smnote{Note that this is not their new result but something
  apparently well-known about lax actions of monoidal categories on a
  category.}  that given an adjunction
$L \dashv R\colon \catB \to \catA$ and a (strict) action
$*\colon \N \times \catB \to \catB$ of the monoid $(\N, +, 0)$
(considered as a discrete monoidal category) on $\catB$ one has a
graded monad $((M_n), \eta, \mu^{n,m})$ given by
\[
  M_n X = R(n * LX), \ \eta_X\colon X\to RLX = R(0 * LX) = M_0 X
\]
and
$\hspace{2.7cm}%\phantom{\mu^{n,m}_X\colon M_nM_m X =\ }
\begin{tikzcd}[column sep = 7, baseline=(B.base)]
  |[alias=B]|\llap{$\mu^{n,m}_X\colon M_nM_m X =\ $}%\ar[equals]{r}
  R(n * LR (m* LX))
  \ar{d}{R(n * \varepsilon_{m * LX})}
  \\
  \llap{$M_{n+m} X = R((n+m)$}* LX) = R(n *(m* LX))
\end{tikzcd}
$

\noindent
where $\eta$ and $\varepsilon$ are the unit and counit of the given adjunction.

We apply this to the adjunction obtained from \autoref{C:adj} for
$n = \omega$ and the obvious action on the category of
$(\T, \omega)$-models given by
\[
  n * ((A_k), (\sigma^A_k)) = ((A_{n+k}), (\sigma^A_{n+k})).
\]
This yields a graded monad $M^\T$ on $\Pos$ with $M^\T_n X$ being the $n$th
component of the free $(\T,\omega)$-model $FX$.

\takeout{
\smnote{We have no proof details of the correspondence anywhere; and
  unlike the CALCO 2015 paper cannot cite the result. So I propose
  that we do not claim any correspondence before having a written
  proof. Maybe we manage to have the construction of a monad from a
  theory written up. If not, then the whole paper has a huge gap!}

We describe a construction of a finitary enriched graded
monad on $\Pos$ from a given graded theory $\T$. 

Each graded $\Sigma$ theory $\T$ with axioms $\E$ induces a 
finitary enriched graded monad $\mathbb{M}_{\T}$ on $\Pos$ by 
taking $M_n X$ to be the poset of depth-$n$ $\T$-defined terms with
variables in $X$, modulo derivable equality. We then take its unit to have 
components given by the inclusion of variables as terms and its multiplication 
$\mu^{nk}\colon M_nM_k\to M_{n+k}$ to have components given by the 
collapsing of layered terms, as usual.

Conversely, every finitary enriched graded monad $(M_n)_{n<\omega}$ on
$\Pos$ induces a graded theory $\mathbb{M}_{\T}$ whose signature is given
by putting $\Sigma(\Gamma, n):= M_n \Gamma$ for every context $\Gamma$
and every $n\in\omega$. Then, for a poset $X$, each depth-$n$ $\Sigma$-term 
$t$ has a canonical interpretation $\llbracket t\rrbracket\in M_n X$ defined 
recursively in the usual way. We let $\E$ consist of an inequation in context
$\Gamma\vdash_n s\leq t$ just in case
$\llbracket s\rrbracket\leq\llbracket t\rrbracket$ in $M_n\Gamma$.
}% end takeout

\takeout{ %% failed attempt via finitary Kleili triples
We will now describe how we obtain a finitery graded monad on $\Pos$
from graded theory. To this end we recall the notion of a \emph{finitary Kleisli
  triple} introduced by Ad\'amek et al.~\cite{AdamekEA03} for locally
finitely presentable categories. A finitary Kleili triple on $\Pos$ consists of an object
assignment $\Gamma \mapsto  M\Gamma$ assigning to every finite poset
$\Gamma$ a poset $\M\Gamma$, a family of maps $\eta_\Gamma\colon
\Gamma \to M\Gamma$ and an operation $(-)^*$ assigning to each
monotone map $\Gamma \to M\Delta$ a monotone map $s^*\colon M\Gamma
\to \M\Delta$ subject to the following axioms
%such that for all $s$ and $t\colon \Delta \to M\Epsilon$ we have
\[
  s^* \cdot \eta_\Gamma = s,\qquad (\eta_\Gamma)^* = \id_{M_\Gamma},
  \qquad (t^* \cdot s)^* = t^* \cdot s^*
\]
for all $s$ and $t\colon \Delta \to M\Upsilon$. There is an
equivalence of categories between finitary Kleisli
triples and finitary monads
on $\Pos$~\cite[Cor.~3.3]{AdamekEA03}.

Now observe that we clearly obtain a finitary Kleisli triple from
\autoref{P:free}:
\smnote[inline]{Write some more details}
Thus, we have proved
\begin{corollary}
Every graded theory induces a finitary monad on $\Pos$.
\end{corollary}
\smnote[inline]{This is of course not what we need; we need a
  \emph{graded} monad.}
%%% I will continue here.
}% end  takeout
\begin{expl}\label{E:induced}
  \lsnote{@Lutz: Add proofs to appendix}
  We describe the functor part $(M_n)_{n<\omega}$ of the graded monads
  induced by the graded theories of~\autoref{E:theory}.
  \takeout{
  \renewcommand{\subsectionautorefname}{Appendix}% SM: I'd like 'Appendix' clickable.
  for detailed proofs, see \autoref{S:app-monadtheory}%
  \renewcommand{\subsectionautorefname}{Section}:
  }
\begin{enumerate}
    
\item\label{item:induced-convexpower}
The graded theories $\JSL(\A)$ and $\JSL^{\down}(\A)$ induce the 
graded monad with $M_n=G^n$ (cf.~\autoref{E:gradedmonad}) for the functors
$G=C_{\omega}(\A\times(-))$ and $G=\pow^{\down}_{\omega}(\A\times(-)),$
respectively. 

\item\label{item:induced-synch} The graded monad induced by $\syncTh$
  is inductively described as follows. Every $M_nX$ consists of
  \emph{live} elements and a single additional element called
  \emph{deadlock}, and denoted~$0$.  The live elements of $M_0X$ are
  the elements of~$X$, ordered as in~$X$ and incomparable to deadlock;
  i.e.\ $M_0X=X+1$, with $1=\{0\}$. The live elements of
  $M_{n+1}X$ are non-empty (order-theoretically) convex sets, ordered
  by the Egli-Milner ordering, of pairs $(a,t)$ where $a\in\A$ and~$t$
  is a live element of~$M_nX$; such pairs are ordered component-wise.
  Up to identifications caused by convexity, the live elements
  of~$M_nX$ may thus be seen as unordered trees, with edges labelled
  in~$\A$ and leaves labelled in~$X$, that have uniform depth~$n$ in
  the sense that every leaf is~$n$ steps from the root. This explains
  our use of the term \emph{synchronization}: At depth~$n$, $\syncTh$
  only takes into account computation trees that proceed in synchrony
  for~$n$ steps. Note that this description implies that~$M_1X$ is
  isomorphic to $C_\omega(\A\times X)$, with the empty set as
  deadlock.

\item\label{item-gsubconvex} The ordered algebraic theory~$\mathbb{S}$
  of~\autoref{E:theory}\ref{item:theory-subconvex} induces
  the monad~$\sdist$ on $\Pos$ which assigns to a poset~$X$ the set of
  all finite subdistributions on $X$, equipped with an ordering that
  we describe below. Then, the graded theory $\gS$
  of~\autoref{E:theory}\ref{item:theory-gsubconvex} induces the graded
  monad~$\M$ with $M_n=\sdist(\A^n\times(-))$.
\end{enumerate}
   
%\item
\label{item:induced-subconvex}
\noindent We now describe the ordering of the subdistribution monad~$\sdist$.
A similar description has been provided by Jones and Plotkin \cite[Lemma 9.2]{JP89} 
in a more restricted domain-theoretic setting (directed-complete posets). Their
argument is based on the max-flow-min-cut theorem;
 \renewcommand{\subsectionautorefname}{Appendix}%
in \autoref{S:app-sdist},%
\renewcommand{\subsectionautorefname}{Section}
we give an independent, syntactic argument. We represent finite subdistributions 
$\mu\colon X \to [0,1]$ as a formal sums
$\sum_{i=1}^n p_i\cdot x_i$ where $n\in\omega$, $p_i\in(0,1]$, and
$\sum_{i=1}^np_i\leq 1$. 
\begin{defn}\label{d:sdistorder}
\begin{enumerate}
\item A \emph{subdivision} of a formal subconvex combination
  $\sum_{i=1}^n p_i\cdot x_i$ is a formal subconvex combination
  $\sum_{i=1}^n(\sum_{j=1}^{m_i}\overline{p}_j)\cdot x_i$ such that
  $\sum_{j=1}^{m_i}\overline{p}_j= p_i$ for all $i\leq n$. We then
  also say that
  $\sum_{i=1}^n(\sum_{j=1}^{m_i}\overline{p}_j)\cdot x_i$
  \emph{refines} $\sum_{i=1}^n p_i\cdot x_i$.

\item We say that a formal subconvex combination
  $\sum_{i = 1}^n p_i\cdot x_i$ is \emph{obviously below}
  $\sum_{i=1}^m q_i\cdot y_i$ if there is an injective map
  $f\colon [n]\to [m]$ such that $p_i\leq q_{f(i)}$ and
  $x_i\leq y_{f(i)}$ for all $i\leq n$. In this case, we write
  $\sum_{i=1}^n p_i\cdot x_i\sqsubseteq \sum_{i=1}^m q_i\cdot y_i$,
  and say that~$f$ \emph{witnesses} this relation. % If in fact
  % $p_i=q_{f(i)}$ and $x_i=y_{f(i)}$ for all~$i$, then we say that~$f$
  % is \emph{non-increasing}.

\item Finally, we define the relation $\preccurlyeq$ on $\sdist X$ by 
putting $\sum_{i=1}^n p_i\cdot x_i\preccurlyeq\sum_{j=1}^m q_j\cdot y_j$ 
if and only if there exist subdivisions $d_1$ and $d_2$ of $\sum_{i=1}^n p_i\cdot x_i$ 
and $\sum_{j=1}^k q_j\cdot y_j$, respectively, such that $d_1\sqsubseteq d_2$. 
\end{enumerate}
\end{defn}
\begin{theorem}\label{T:sdist}
The monad $\sdist$ given by the above theory 
assigns to a poset $X$ the set of all finitely generated
subdistributions on $X$ equipped with the order $\preccurlyeq$.
\end{theorem}
\noindent%
\renewcommand{\subsectionautorefname}{Appendix}% SM I'd like 'Appendix' clickable.
For a detailed proof, see \autoref{S:app-sdist}.%
\renewcommand{\subsectionautorefname}{Section}
\end{expl}
\subsection{Depth-1 Graded Monads}\label{S:depth1}
\noindent An important class of graded theories are those in which
operations and axioms are restricted to have depth at most~$1$.
Graded monads~$\M$ induced by such a graded theory are particularly
well-behaved in that their $M_n$-algebras may be constructed
compositionally from $M_1$-algebras~\cite{MPS15}. We briefly recall
the essentials, specializing to posets.

\begin{defn}
  A graded theory $\T$ is \emph{depth-1-generated} if all of its
  operations have depth at most 1. We call $\T$ \emph{depth-1} if it
  is depth-1-generated and all of its axioms are depth-1. A graded
  monad on $\Pos$ is \emph{depth-1-generated} (depth-1) if it is
  (isomorphic to) the graded monad induced by some depth-1-generated
  (depth-1) graded theory.
\end{defn}

\begin{expl}
The theories of \autoref{E:theory} are depth-1.
\end{expl}
\noindent 
We fix the following terminology: a natural transformation is an
\emph{epi-transformation} if all its components are epimorphisms. In
particular, a natural transformation between endofunctors on $\Pos$ is
an epi-transformation if its components are surjections.

Now, we have the following characterization of depth-1 graded monads
on $\Pos$, the proof of which is completely analogous to the
corresponding characterization on
$\Set$~\cite[Prop. 7.3]{MPS15} (the only difference in our setting
rests in our additional assumption that $M_1\mu^{1,n}$ is an
epi-transformation; this is automatic in the $\Set$-based case since epis split in $\Set$ and are hence preserved by all functors).

\removeThmBraces
\begin{proposition}\label{P:depth1}
  Let $\M=((M_n)_{n<\omega}, \eta, (\mu^{n,k})_{n,k<\omega})$ be a
  graded monad on $\Pos$. Then $\M$ is depth-1-generated if and only
  if $\mu^{n,k}$ and $M_1\mu^{1,n}$ are epi-transformations for all
  $n,k\in\omega$. Moreover, $\M$ is depth-1 if and only if $\M$ is
  depth-1-generated and, for all $n\in\omega$, the diagram
  \begin{equation} 
    \begin{tikzcd}[column sep = 40]
      M_1M_0M_n
      \arrow[r, "{M_1\mu^{0,n}}", shift left]
      \arrow[r, "{\mu^{1,0}M_n}"', shift right]
      &
      M_1M_n
      \arrow[r, "{\mu^{1,n}}"]
      &
      M_{n+1}
    \end{tikzcd}
  \end{equation}
  is object-wise a coequalizer diagram in the category of $M_0$-algebras.
\end{proposition}
\resetCurThmBraces

\begin{rem}\label{rem:mone}
As indicated above, in depth-$1$ graded monads there is a
focus on $M_0$- and $M_1$-algebras. We list some basic observations on
such algebras.
\begin{enumerate}
\item $M_0$-algebras are Eilenberg-Moore algebras for the monad $(M_0, \eta, \mu^{0,0})$.
\item By observations in~\autoref{sec:graded}, for all $n$, $X$,
  $M_nX$ is an $M_0$-algebra with structure map
  $\mu^{0n}\colon M_0M_n X\to M_n X$.
\item\label{item:mone} Similarly, for all~$n$,~$X$, $(M_nX,M_{n+1}X)$
  form an $M_1$-algebra, with structure maps
  $\mu^{0,n},\mu^{1n},\mu^{0,n+1}$.

\item The associative law $\mu^{nk}\cdot\mu^{0n}M_k=\mu^{0,n+k}\cdot M_0\mu^{nk}$
% \begin{center}
%   \begin{tikzcd}[column sep = 40]
%     M_0M_nM_k
%     \arrow[d, "M_0\mu^{nk}"']
%     \arrow[r, "\mu^{0n}M_k"]
%     &
%     M_nM_k \arrow[d, "\mu^{nk}"]
%     \\
%     M_0M_{n+k}
%     \arrow[r, "{\mu^{0,n+k}}"]
%     &
%     M_{n+k}                              
%   \end{tikzcd}
% \end{center}
  of~$\M$ then states precisely that
  $\mu^{n,k}_X\colon M_nM_k X\rightarrow M_{n+k}X$ is a morphism of
  $M_0$-algebras.
\item An $M_1$-algebra consists of $M_0$-algebras $(A_0, a^{00})$ and
  $(A_1, a^{01})$ and a \emph{main structure map}
  $a^{10}\colon M_1A_0\to A_1$ satisfying two instances
  of Diagram~\eqref{D:gradalg}; one of which is the square 
\begin{equation*}
  \begin{tikzcd}[column sep = 30]
    M_0M_1A_0 \ar{r}{M_0a^{10}}
    \ar{d}[swap]{\mu^{01}_{A_0}}
    &
    M_0A_1
    \ar{d}{a^{01}}
    \\
    M_1A_0
    \ar{r}{a^{01}}
    &
    A_1                        
  \end{tikzcd}
\end{equation*} 
stating that $a^{10}$ is a morphism of $M_0$-algebras from 
$(M_1A_0, \mu^{10}_{A_0})$ to $(A_1, a^{01})$, while the other 
expresses that the diagram below commutes:
  \begin{equation}\label{D:coeq}
    \begin{tikzcd}[column sep = 35]
      M_1M_0A_0
      \arrow[r, "\mu^{10}_{A_0}", shift left]
      \arrow[r, "M_1a^{00}"', shift right]
      &
      M_1A_0 \arrow[r, "a^{10}"]
      &
      A_1.
    \end{tikzcd}
  \end{equation}
\end{enumerate}
\end{rem}
% 
% We can take products of algebras in the expected way: 

% \removeThmBraces
% \begin{proposition}[{\cite[Prop.~6.3]{MPS15}}]\label{P:products}
%   If $\C$ has products, then the category of $M_n$-algebras and their
%   morphisms has products.
% \end{proposition} 
% \resetCurThmBraces

% \noindent
% Explicitly, the product of an $I$-indexed family
% $A^i=((A^i_k)_{k\leq n}, (a^{mk}_i)_{m+k\leq n})$ of $M_n$-algebras
% has carriers $\prod_{i\in I}A^i_k$ for $k\leq n$ and structure maps
% given by the following composition:
% \[%\textstyle
%   M_m\big(\prod_{i\in I}A^i_k\big)
%   \xrightarrow{\langle M_m\pi_i\rangle_{i \in I}}
%   \prod_{i\in I}M_mA^i_k
%   \xrightarrow{\prod_{i\in I}a^{mk}_i}
%   \prod_{i\in I}A^i_{m+k}.
% \]

%
\takeout{%% SM: not needed; the full proof is in the cited paper.
\begin{proof}[Proofsketch]
  Given an $M_n$-algebra $((A_k), (a^{mk}))$ and a monotone map
  $f\colon X\rightarrow A_0$, define its free extension
  $f^{\sharp}=(f^{\sharp}_k)$ by putting
  $f^{\sharp}_k= a^{k0}\cdot M_kf$.
\end{proof}
}% end takeout

\subsection{Canonical $M_1$-algebras}

\noindent $M_1$-algebras such that Diagram~\eqref{D:coeq} is a
coequalizer are of particular interest for use in the semantics of
graded logics; these are precisely the $M_1$-algebras that are free
over their $0$-part~\cite{DMS19}. We recall the details
presently.

\begin{defn}\label{D:canonical}
The \emph{$0$-part} of an $M_1$-algebra $A$ is the $M_0$-algebra 
$(A_0, a^{00})$. Taking $0$-parts defines a functor $Z$ from the category of 
$M_1$-algebras to the category of $M_0$-algebras which maps a morphism 
$(h_0, h_1)$ of $M_1$-algebras to $h_0$. An $M_1$-algebra $A$ is \emph{canonical} 
if it is free over its $0$-part with respect to the functor~$Z$, with
universal morphism $\id_{A_0}$.%$\id_{(A_0,a^{00})}$.
\end{defn}

\noindent In other words, an $M_1$-algebra
$A=((A_0, A_1), (a^{00}, a^{01}, a^{10}))$ is canonical if for every
$M_1$-algebra $B$ with $0$-part $(B_0, b^{00})$ and every morphism
$h\colon A_0\rightarrow B_0$ of $M_0$-algebras, there exists a unique
$\C$-morphism $h^{\sharp}\colon A_1\rightarrow B_1$ such that
$(h, h^{\sharp})$ is a morphism of $M_1$-algebras, a property that
will later allow us to interpret modalities on~$A$.

\removeThmBraces
\begin{proposition}[{\cite[Lem.~5.3]{DMS19}}]\label{L:canonical}\label{P:coeq}
  An $M_1$-algebra $A$ is canonical iff \eqref{D:coeq} is a
  (reflexive) coequalizer diagram in the category of $M_0$-algebras.
\end{proposition}
\resetCurThmBraces
\noindent%
\renewcommand{\propositionautorefname}{Propositions}%
Combining \autoref{L:canonical} and~\ref{P:depth1}, we have:
\renewcommand{\propositionautorefname}{Proposition}

\begin{corollary}\label{C:canonical}
  Let $\M$ be a depth-1 graded monad on $\Pos$. Then, for every
  $n\in\omega$ and and every poset $X$, the $M_1$-algebra
  $(M_nX, M_{n+1}X)$ as per
  \autoref{rem:mone}\ref{item:mone} % and structure
% maps $(\mu^{0n}_X, \mu^{0, n+1}_X, \mu^{1, n+1}_X)$
is canonical.
\end{corollary}
%

% \removeThmBraces
% \begin{theorem}{\cite[Thm. 3.7]{DMS19}}\label{T:depth1}
% Depth-1 graded monads are in bijective correspondence with 6-tuples
% $(M_0, M_1, \eta, \mu^{0,0}, \mu^{1,0}, \mu^{0,1})$ such that the given
% data satisfy all applicable instances of the graded monad laws.
% \end{theorem}
% \resetCurThmBraces
%
%%%%%%%%%%%%%%%%%%%%%%%%%%%%%%%%%%%%%%%%%%%%%%%%%%
%%%%%%%%%%%%%%%%%%%%%%%%%%%%%%%%%%%%%%%%%%%%%%%%%%
%%%%%%%%%%%%%%%%%%%%%%%%%%%%%%%%%%%%%%%%%%%%%%%%%%
%%%%%%%%%%%%%%%%%%%%%%%%%%%%%%%%%%%%%%%%%%%%%%%%%%
%%%%%%%%%%%%%%%%%%%%%%%%%%%%%%%%%%%%%%%%%%%%%%%%%%
%%%%%%%%%%%%%%%%%%%%%%%%%%%%%%%%%%%%%%%%%%%%%%%%%%
%%%%%%%%%%%%%%%%%%%%%%%%%%%%%%%%%%%%%%%%%%%%%%%%%%
%%%%%%%%%%%%%%%%%%%%%%%%%%%%%%%%%%%%%%%%%%%%%%%%%%
%BEHAVIOURAL PREORDERS
%%%%%%%%%%%%%%%%%%%%%%%%%%%%%%%%%%%%%%%%%%%%%%%%%%
\section{Behavioural Preorders via Graded Monads}\label{S:gbp}
\noindent We next discuss how graded monads induce behavioural
preorders on systems.
\begin{defn}\label{D:gbp}
  A \emph{graded behavioural preorder} on $G$-coalgebras is a pair
  $(\alpha, \M)$ consisting of a graded monad $\M$ on $\Pos$ and a 
  natural transformation $\alpha\colon G\to M_1$. The sequence
  $(\gamma^{(n)}\colon X\rightarrow M_nX)_{n<\omega}$ of
  \emph{$n$-step $(\alpha, \M)$-behaviours} for a $G$-coalgebra
  $\gamma\colon X\rightarrow GX$ is defined by
\begin{align*}
  \gamma^{(0)} &:= \eta_X,\\
  \gamma^{(n+1)}&:= (X\xrightarrow{\alpha_X\cdot\gamma} M_1X\xrightarrow{M_1\gamma^{(n)}} M_1M_nX\xrightarrow{\mu^{1, n}} M_{n+1}X).
\end{align*}
For states $x\in X$ and $y\in Y$ in $G$-coalgebras $\gamma\colon
X\rightarrow GX$ and $\delta\colon Y\rightarrow GY$, we say that $y$
\emph{$(\M,\alpha)$-refines} $x$ if
\[
  M_n{!}\cdot\gamma^{(n)}(x)\leq
  M_n{!}\cdot\delta^{(n)}(y)\qquad\text{for all $n\in\omega$}.
\]
\end{defn}
\noindent We take a look at the behavioural preorders induced by our
running examples:

\subsubsection{$G$-Similarity}\label{E:gbp}
Recall from \autoref{E:gradedmonad} that taking $M_n= G^n$ defines a
graded monad $\mathbb{G}$. It is easy to
verify that $(\id_G,\M)$-refinement is precisely finite-depth
$G$-similarity in the sense of~\autoref{D:fdsim}. Under our running
assumptions (\autoref{ass:functor}), $(\id_G,\M)$-refinement
thus captures (coalgebraic) $G$-simulation.

For example, taking $G=C_{\omega}(\A\times (-))$, we capture bisimilarity 
on $\A$-LTS. (It may seem surprising that we obtain bisimilarity as a behavioural 
preorder in this case; however, recall from \autoref{E:convex}\ref{item:convex} that the 
ordering on an $C_{\omega}(\A\times-)$-coalgebra is a bisimulation, and also that on 
discrete orders, $C_\omega(\A\times(-))$ is the standard powerset equipped with the discrete 
ordering.)

\subsubsection{Similarity}
As a further case of this type, recall the graded theory
$\JSL^{\down}(\A)$ from \autoref{E:theory}\ref{item:theory-down}.
Since %$\JSL^{\down}(\A)$
it has no depth-$0$ operations, its induced graded monad also has the
form $M_n= G^n$, for the functor
$G=(\pow^{\down}_{\omega}(\A\times(-)))^n$
(see~\autoref{E:induced}\ref{item:induced-convexpower}). %
By taking $\alpha=\id$, we thus obtain a graded behavioural preorder
that corresponds to coalgebraic similarity on coalgebras for
$\pow^{\down}_{\omega}(\A\times(-))$, which is similarity in the
classical sense~\cite{Lev11}.

Further simulation-type behavioural preorders such as ready simulation 
and complete simulation are covered in a straightforward manner by adding components to 
the action set. For instance, ready similarity, which additionally requires states to have the same 
\emph{ready sets}, i.e.~sets of enabled actions, is covered by taking 
$G=\pow^{\down}_\omega(\pow_\omega(\A)\times\A\times (-)+1)$, with 
$\pow_\omega(\A)$ discretely ordered and~$1$ denoting deadlock. The original 
functor $\pow^{\down}_\omega(\A\times(-))$ is then transformed into~$G$ by 
mapping $S$ to $\{(\pi_1[S],a,x)\mid (a,x)\in S\}$ for $S\neq\emptyset$, where~$\pi_1$ 
denotes the first projection, and to the unique element of~$1$ otherwise. This is similar to 
decoration-oriented coalgebraic approaches to trace semantics~\cite{BonchiEA16}. 
\subsubsection{Kleisli and Eilenberg-Moore laws} Recall that Kleisli
laws and Eilenberg Moore laws for a functor~$F$ and a monad~$T$ induce
graded monads of the shape $M_xX=TF^nX$ and $M_nX=F^nTX$, respectively
(\autoref{E:gradedmonad}\ref{item:Kleisli},\ref{item:EM}). On
sets, the arising graded semantics~\cite{MPS15} are known to relate to
existing approaches to coalgebraic trace
semantics~\cite{HasuoEA07,JacobsEA15}, being slightly finer (by
considering all traces rather than only accepting ones). For the
Kleisli case, the extension to $\Pos$ is, in principle, new. It
subsumes typical notions of trace inclusion; one example of this type
is probabilistic trace inclusion, discussed further below.

\subsubsection{Synchronous bisimilarity}\label{sec:synchrony}
Recall the theory $\syncTh$
(\autoref{E:theory}\ref{item:theory-sync}) and the description of the
induced graded monad~$\M$
(\autoref{E:induced}\ref{item:induced-synch}). For the functor
$G=C_\omega(\A\times(-))$, whose coalgebras are $\A$-LTS equipped with
an ordering that is a bisimulation (\autoref{E:convex}), we define a
graded behavioural preorder by the isomorphism
$C_\omega(\A\times(-))\to M_1$ mentioned in
\autoref{E:induced}\ref{item:induced-synch}.  The induced notion of
behavioural preorder is in fact, like in the case of bisimilarity, a
notion of process equivalence that we term \emph{synchronous
  bisimilarity}: Recall that two states are (finite-depth) bisimilar
iff they have the same tree unfoldings at every finite depth, where
the relevant trees have set-based branching in the sense that every
node is identified by its set of children (more precisely the convex
closure of this set, in our present order-theoretic setting). In these
terms, two states are synchronously bisimilar if they have the same
\emph{pruned} tree unfolding at every finite depth, where the pruning
removes all branches of the tree that end in deadlocks before the
current depth is reached. One easily constructs separating examples
showing that synchronous bisimilarity lies strictly between trace
equivalence and bisimilarity.

\subsubsection{Probabilistic Trace Inclusion}\label{S:probtraceinc}

\takeout{%cf: I moved this to label{E:theory}, notes included
We capture finite probability subdistributions (which are defined like
finite distributions except the global mass is only required to be at
most~$1$ rather than exactly~$1$) by 
the algebraic theory of subconvex
algebras (also known as positive convex modules~\cite{Pumpluen03}),
whose operations are formal \emph{subconvex combinations}
$\sum_{i=1}^np_i\cdot(-)$ for $\sum p_i\le 1$, and whose equations
reflect the laws of (plain) monad algebras: The equation
$\sum \delta_{ik}\cdot x_k=x_i$, with $\delta_{ik}$ being Kronecker
delta ($\delta_{ik}=1$ if $i=k$, and $\delta_{ik}=0$ otherwise),
reflects compatibility with the unit, and the equation
scheme
\[\textstyle
  \sum_{i=1}^n p_i\cdot\sum_{k=1}^m q_{ik}\cdot x
  _k=\sum_{k=1}^m\big(\sum_{i=1}^n p_i q_{ik})\cdot x_k
\]
reflects compatibility with the monad multiplication. We form the
graded algebraic theory $\mathbb{S}$ by further 
imposing inequations of the form
\[\textstyle
  % \{x_1,\dots, x_n\}\vdash %% SM: discrete context omitted
  \sum_{i=1}^n p_i\cdot x_i\leq \sum_{i=1}^n
  q_i\cdot x_i\quad (\text{$p_i\leq q_i$ for all $i\leq n$})
\]
and inequations of the form
\[\textstyle
  \{x_i\leq y_i\mid i\leq n\}
  \vdash
  \sum_{i=1}^n p_i\cdot x_i \leq \sum_{i=1}^n p_i\cdot y_i,
\]
specifying monotonicity of formal
subconvex combinations. 
\smnote{The name `monotonicity (condition)' is
  used several times later on (in the appendix), and so it should be
  explained what this refers to.}
 }  
 
\takeout{%cf: I moved this to "from graded theories to graded monads"
Interpreting this algebraic theory
over $\Pos$ yields the monad~$\sdist$ on $\Pos$ which assigns to a
poset $X$ the set of all finite subdistributions on $X$ equipped with
the order that we now describe. A similar description has been
provided by Jones and Plotkin \cite[Lemma 9.2]{JP89} in a more
restricted domain-theoretic setting (directed-complete posets). Their
argument is based on the max-flow-min-cut theorem; we give an
independent, more syntactic argument. We represent finite
subdistributions $\mu\colon X \to [0,1]$ as a formal sums
$\sum_{i=1}^n p_i\cdot x_i$ where $n\in\omega$, $p_i\in(0,1]$, and
$\sum_{i=1}^np_i\leq 1$. 
\begin{defn}\label{d:sdistorder}
\begin{enumerate}
\item A \emph{subdivision} of a formal subconvex combination
  $\sum_{i=1}^n p_i\cdot x_i$ is a formal subconvex combination
  $\sum_{i=1}^n(\sum_{j=1}^{m_i}\overline{p}_j)\cdot x_i$ such that
  $\sum_{j=1}^{m_i}\overline{p}_j= p_i$ for all $i\leq n$. We then
  also say that
  $\sum_{i=1}^n(\sum_{j=1}^{m_i}\overline{p}_j)\cdot x_i$
  \emph{refines} $\sum_{i=1}^n p_i\cdot x_i$.

\item We say that a formal subconvex combination
  $\sum_{i = 1}^n p_i\cdot x_i$ is \emph{obviously below}
  $\sum_{i=1}^m q_i\cdot y_i$ if there is an injective map
  $f\colon [n]\to [m]$ such that $p_i\leq q_{f(i)}$ and
  $x_i\leq y_{f(i)}$ for all $i\leq n$. In this case, we write
  $\sum_{i=1}^n p_i\cdot x_i\sqsubseteq \sum_{i=1}^m q_j\cdot y_i$,
  and say that~$f$ \emph{witnesses} this relation. If in fact
  $p_i=q_{f(i)}$ and $x_i=y_{f(i)}$ for all~$i$, then we say that~$f$
  is \emph{non-increasing}.

\item Finally, we define the relation $\preccurlyeq$ on $\sdist X$ by 
putting $\sum_{i=1}^n p_i\cdot x_i\preccurlyeq\sum_{j=1}^m q_j\cdot y_j$ 
if and only if there exist subdivisions $d_1$ and $d_2$ of $\sum_{i=1}^n p_i\cdot x_i$ 
and $\sum_{j=1}^k q_j\cdot y_j$, respectively, such that $d_1\sqsubseteq d_2$. 
\end{enumerate}
\end{defn}
\begin{theorem}\label{T:sdist}
The monad $\sdist$ given by the above theory 
assigns to a poset $X$ the set of all finitely generated
subdistributions on $X$ equipped with the order $\preccurlyeq$.
\end{theorem}
\noindent
\renewcommand{\subsectionautorefname}{Appendix} % SM want 'Appendi' clickable.
For a detailed proof, see \autoref{S:app-sdist}.
\renewcommand{\subsectionautorefname}{Section}
}
\takeout{%moved this to \ref{E:convex}
Coalgebras for the functor 
$\sdist(\A\times(-))$ on $\Pos$ are then ordered probabilistic transition 
systems with possible deadlock.
}
\takeout{%cf: I moved this to \label{E:theory}
We define a graded monad on $\Pos$ by taking subconvex combinations as
operations of depth~$0$, and the actions as unary operations of
depth~$1$; besides the mentioned inequations of subconvex algebras, we
include equations stating that the actions distribute over subconvex
combinations, i.e.~for $\sum_{i=1}^k p_i\le 1$ and $a\in\A$, we
postulate depth-$1$ equations
\begin{equation*}\textstyle
  a(\sum_{i=1}^k p_i \cdot x_i) = \sum_{i=1}^kp_i\cdot a(x_i).
\end{equation*}
} \takeout{%I moved this to \label{E:induced}
  The induced graded monad has $M_n = \sdist(\A^n\times (-))$.  }
Recall that coalgebras for the functor $G=\sdist(\A\times(-))$ on
$\Pos$ are ordered probabilistic transition systems with possible
deadlock (\autoref{E:convex}), and moreover recall the graded
monad~$\M$ with \[M_nX=\sdist(\A^n\times(-))\]
(see \autoref{E:induced}\ref{item-gsubconvex}). The graded behavioural
preorder induced by~$\M$ (and the identity $G\to M_1$) is
\emph{probabilistic trace inclusion}: For a state~$x$ in a coalgebra
$\gamma\colon X\to GX$,
$M_n!\cdot\gamma^{(n)}(x)\in\sdist(\A^n\times 1)$ is the distribution
on \mbox{length-$n$} traces induced by~$n$ independently distributed
successive random transition steps.

%
%%%%%%%%%%%%%%%%%%%%%%%%%%%%%%%%%%%%%%%%%%%%%%%%%%
%%%%%%%%%%%%%%%%%%%%%%%%%%%%%%%%%%%%%%%%%%%%%%%%%%
%%%%%%%%%%%%%%%%%%%%%%%%%%%%%%%%%%%%%%%%%%%%%%%%%%
%%%%%%%%%%%%%%%%%%%%%%%%%%%%%%%%%%%%%%%%%%%%%%%%%%
%%%%%%%%%%%%%%%%%%%%%%%%%%%%%%%%%%%%%%%%%%%%%%%%%%
%%%%%%%%%%%%%%%%%%%%%%%%%%%%%%%%%%%%%%%%%%%%%%%%%%
%%%%%%%%%%%%%%%%%%%%%%%%%%%%%%%%%%%%%%%%%%%%%%%%%%
%%%%%%%%%%%%%%%%%%%%%%%%%%%%%%%%%%%%%%%%%%%%%%%%%%
%LOGICS
%%%%%%%%%%%%%%%%%%%%%%%%%%%%%%%%%%%%%%%%%%%%%%%%%%
\section{Characteristic Graded Logics for Graded Behavioural Preorders}\label{sec:logics}

\noindent We recall the general setup of graded
logics~\mbox{\cite{MPS15,DMS19}}, which are modal logics with formulae
graded according to the number of transition steps they look ahead.
Their \textbf{syntax} is parameterized over a set $\Theta$ of \emph{truth
  constants}, a set~$\mathcal{O}$ of \emph{propositional operators}
with assigned finite arities, and a set~$\Lambda$ of \emph{modalities}
with assigned finite arities, which we fix for the remainder of the
paper. We will restrict our technical discourse to unary modalities
for readability; the treatment of $n$-ary modalities for $n\geq 1$ is
achieved through additional indexing in a similar manner as we do for
$n$-ary propositional operators.

Then, sets $\mathcal{L}_n$ of \emph{formulae of uniform depth~$n$} are
defined recursively. The set $\mathcal{L}_0$ of graded formulae of
uniform depth~$0$ is generated by the grammar
\[
\varphi::= p(\varphi_1,\dots, \varphi_k)\mid c \quad (k\text{-ary } p\in\mathcal{O}, c\in\Theta)
\]
and the set $\mathcal{L}_{n+1}$ of graded formulae of uniform depth $n+1$ is generated by the grammar 
\[
\varphi::= p(\varphi_1,\dots, \varphi_k)\mid L(\psi)	\quad (k\text{-ary } p\in\mathcal{O}, L\in\Lambda)
\]
where $\psi$ is a formula of uniform depth $n$. Note that
propositional operators may be constant, i.e.\ have arity~$0$; they
then appear in every~$\mathcal{L}_n$, while truth constants appear
only in~$\mathcal{L}_0$. \emph{We assume that all $0$-ary
  propositional operators are also truth constants.}

Besides the functor~$G$, the \textbf{semantics} of a graded logic is
further parameterized by a graded behavioural preorder
$(\alpha, \M=((M_n), \eta, (\mu^{nk})))$ with $\M$ depth-1, along with
the following.
\begin{itemize}%[wide]
\item An $M_0$-algebra $(\Omega, o\colon M_0\Omega\rightarrow\Omega)$
  of \emph{truth values}. By the standard theory of monad algebras,
  powers $\Omega^k$ of $\Omega$ are again $M_0$-algebras.
\item For each truth constant $c\in\Theta$, a map
  $\hat{c}\colon 1\to\Omega$.
\item For each propositional operator $p\in\mathcal{O}$, with
  arity~$k$, an $M_0$-morphism
  $\llbracket p\rrbracket\colon \Omega^k\rightarrow \Omega$.
\item For each $L\in\Lambda$, an $M_1$-algebra $A_L$ with carriers
  $(\Omega, \Omega)$ and structure maps
  $(o, o, \llbracket L\rrbracket\colon M_1\Omega\rightarrow\Omega)$.
\end{itemize}
\noindent We fix these data for the remainder of the paper. 

\begin{rem}\label{E:tobject}
  Which propositional operators can be used is thus determined
  by~$M_0$; that is, the more depth-$0$ operations there are, the
  fewer propositional operators are allowed. E.g.\ for graded monads
  $M_nX=G^nX$ (which capture $G$-similarity), $M_0$ is identity, so
  every monotone map can be used as a propositional operator. As a
  further very simple example, in the graded monad induced by the
  graded theory $\syncTh$ (\autoref{E:theory}\ref{item:theory-sync}),
  $M_0$-algebras are just partially ordered sets with a distinguished
  element~$0$, which propositional operators need to preserve
  (allowing, e.g., disjunction and conjunction but ruling out $\top$
  as a propositional operator). Further examples will be seen later.
\end{rem}
\noindent Given $L\in\Lambda$, a canonical $M_1$-algebra $A$, and an
$M_0$-morphism $f\colon A_0\rightarrow\Omega$, we write
$\llbracket L\rrbracket(f)$ to denote the unique $M_0$-morphism
extending $f$ to an $M_1$-morphism $A\rightarrow A_L$. That is,
$\llbracket L\rrbracket(f)$ is the unique $M_0$-morphism such that the
square below commutes:
\begin{equation}\label{Di:modality}
\begin{tikzcd}
M_1A_0 \arrow[r, "M_1f"] \arrow[d, "a^{10}"'] & M_1\Omega \arrow[d, "\llbracket L\rrbracket"] \\
A_1 \arrow[r, "\llbracket L\rrbracket(f)"]    & \Omega                                         
\end{tikzcd}
\end{equation}

\begin{defn}\label{D:semantics}
  The \emph{evaluation} $\llbracket\varphi\rrbracket$ of a formula
  $\varphi$ of uniform depth $n$ is defined recursively as an
  $M_0$-morphism
  $\llbracket\varphi\rrbracket\colon (M_n1, \mu^{0n}_1)\rightarrow
  (\Omega, o)$ as follows:
\begin{align*}
 \llbracket c\rrbracket & := (M_01\xrightarrow{M_0\hat{c}}
 M_0\Omega\xrightarrow{o}\Omega); \\
 \llbracket p(\varphi_1,\dots, \varphi_k)\rrbracket& := \llbracket p\rrbracket\langle \llbracket\varphi_1\rrbracket,\dots, \llbracket\varphi_k\rrbracket\rangle; \\
 \llbracket L(\varphi)\rrbracket &:= \llbracket L\rrbracket(\llbracket\varphi\rrbracket)
\end{align*}
(where $\langle\dots\rangle$ denotes tupling of maps).  The
\emph{meaning} of a graded formula $\varphi$ of uniform depth $n$ in a
$G$-coalgebra $\gamma\colon X\rightarrow GX$ is then given by the map
\[
\llbracket\varphi\rrbracket_{\gamma}:=X\xrightarrow{M_n!\cdot\gamma^{(n)}}M_n1\xrightarrow{\llbracket\varphi\rrbracket}\Omega.
\]
\end{defn}
\noindent Here, note that in the clause for
$\llbracket L(\varphi)\rrbracket$,
$\llbracket L\rrbracket(\llbracket\varphi\rrbracket)$ is actually
defined because by \autoref{C:canonical}, $(M_n1, M_{n+1}1)$ is a
canonical $M_1$-algebra. By construction, every graded logic is
preserved by the underlying behavioural preorder:
\begin{lemma}\label{lem:preservation}
  Let $x,y$ be states in $G$-coalgebras $(X,\gamma)$ and $(Y,\delta)$,
  respectively. If $y$ $(\alpha,\M)$-refines $x$, then
  $\llbracket\phi\rrbracket_\gamma(x)\le\llbracket\phi\rrbracket_\delta(y)$
  for every formula~$\phi$.
\end{lemma}

\begin{expl}\label{E:logics}
  We describe characteristic logics for some of the graded behavioural
  preorders of \autoref{S:gbp}.  In the first two examples, we take
  the truth object $\Omega$ to be $\mathbbm{2}$, the linear order
  $\{\bot<\top\}$, and in the final example on probabilistic trace
  inclusion, we take $\Omega = [0, 1]$.
\end{expl}

\paragraph{Coalgebraic Modal Logic on $\Pos$}\label{para:cml} Recall from \autoref{E:gbp} that
the graded monad with $M_n=G^n$ captures coalgebraic similarity on
$G$-coalgebras.  As noted in \autoref{E:tobject}, we can use all
monotone maps $\mathbbm{2}^k\rightarrow\mathbbm{2}$ as propositional
operators, in particular disjunction, conjunction, truth, and falsity.
Furthermore, $M_1$-algebras are given by monotone maps
$a^{10}\colon GA_0\rightarrow A_1$; in particular, modalities are
interpreted as monotone maps $G\mathbbm{2}\rightarrow\mathbbm{2}$.
The evaluation of such maps according to \autoref{D:semantics}
corresponds precisely to the meaning of modalities in coalgebraic
modal logic~\cite{Pattinson04,Schroder08} in its poset
variant~\cite{KKV12}.

\paragraph{Bisimilarity} We can instantiate the previous example to
$G=C_\omega(\A\times(-))$ (cf.\
\autoref{E:convex}\ref{item:convex}). In this case, we can use
modalities $\Diamond_a$ \emph{and} $\Box_a$, interpreted by monotone
(!) maps $G\mathbbm{2}\to\mathbbm{2}$ given by
$\llbracket \Diamond_a\rrbracket(S)=\top$ iff $(a,\top)\in S$ and
$\llbracket \Box_a\rrbracket(S)=\top$ iff $(a,\bot)\notin S$. We can
thus define negation by taking negation normal forms, ending up with
full Hennessy-Milner logic on $\A$-LTS.

\paragraph{Similarity} We can similarly instantiate to
$G=\pow^{\down}_{\omega}(\A\times(-))$. We can then still use
modalities $\Diamond_a$, interpreted analogously as
over~$C_\omega(\A\times(-))$, but no longer~$\Box_a$, whose
interpretation would now fail to be monotone. We thus obtain the
positive fragment of Hennessy-Milner logic on $\A$-LTS, which is well
known to characterize similarity~\cite{Glabbeek90}.  The adaptation of
the logic to further simulation-type preorders such as ready
simulation is achieved by indexing the modalities with ready sets~$I$
like in earlier work on graded process equivalence~\cite{DMS19}:
$\lozenge_{a,I}\phi$ requires, besides satisfaction
of~$\lozenge_a\phi$, that the ready set is exactly~$I$.

\paragraph{Synchronous Bisimilarity} For synchronous bisimilarity
(\autoref{sec:synchrony}), we obtain almost the same logic as for
bisimilarity, except that as noted in Remark~\ref{E:tobject}, we can
no longer include truth~$\top$ as a (constant) propositional operator;
instead, it needs to be a truth constant. Since the latter appear only
at depth~$0$, this restricts the set of formulae; e.g.\
$\Diamond_a\top\land\Diamond_a\Diamond_b\top$ is no longer a formula.

\paragraph{Probabilistic Trace Inclusion}
Recall that probabilistic transition systems (with actions in $\A$) in
$\Pos$ are coalgebras for $\sdist(\A\times(-))$, and probabilistic
trace inclusion is given by the graded monad with
$M_n= \sdist(\A^n\times(-))$. In particular, we have
$M_0\cong \sdist$. We take the ensuing logic
$\mathcal{L}_{\mathsf{Prob}}$ to consist of modalities from
$\Lambda:=\{ \langle a\rangle\mid a\in\A\}$ and a single truth
constant $\top$. As mentioned above, we take $\Omega=[0,1]$ as the
object of truth values, made into a subconvex algebra by taking
expected values (i.e.~the algebra stucture
$o\colon \sdist[0,1] \to [0,1]$ maps formal subconvex combinations to
their respective arithmetic sums in $[0,1]$).  We define
$\hat\top = 1 \in[0, 1]$; using that $M_01 = \sdist 1 = [0,1]$ this
yields that $\llbracket \top\rrbracket$ is the identity map on
$[0,1]$.  \smnote{There was a mistake here saying that
  $\llbracket \top \rrbracket = 1\in [0,1]$; but this does not match
  Def.~VI.4.}  We take the interpretation of the modality
$\langle a\rangle$ to be the $M_1$-algebra on $[0,1]$ with structure
maps given by $o$ and by
\[\textstyle
  \llbracket\langle a\rangle\rrbracket\colon\sdist(\A\times[0,1])\rightarrow[0,1], \quad
  \mu\mapsto \sum_{p\in[0,1]}\mu(a, p) \cdot p. 
\]
The logic could be extended to include $M_0$-morphisms, i.e.\
morphisms of subconvex algebras, as propositional operators,
preserving \autoref{lem:preservation}. We show in the next section
that the weaker logic $\mathcal{L}_{\mathsf{Prob}}$ is already
expressive.

\section{Expressiveness of Graded Logics}\label{sec:expressiveness}

\noindent We now introduce an expressiveness criterion for graded
logics (\autoref{thm:expr}), providing a preordered variant of
previous results on process \emph{equivalences} on $\Set$
coalgebras~\cite{DMS19}. Explicitly, a graded logic $\mathcal{L}$ for
$(\alpha,\mathbb{M})$-refinement is \emph{expressive} if the converse
of \autoref{lem:preservation} holds: Given states $x, y$ in coalgebras
$\gamma\colon X\to GX$ and $\delta\colon Y\to GY$, respectively, if
$\llbracket\varphi\rrbracket_{\gamma}(x)\leq\llbracket\varphi\rrbracket_{\delta}(y)$
for all formulae~$\varphi$, then~$y$ $(\alpha,\M)$-refines~$x$. In
combination with \autoref{lem:preservation}, we then have that theory
inclusion precisely characterizes refinement.

We say that a set $\mathfrak{A}$ of monotone maps $X\rightarrow Y$ is
\emph{jointly order-reflecting} if for all $x,y\in X$ such that
$f(x)\leq f(y)$ for all $f\in \mathfrak{A}$, we have $x\leq y$. Then,
our expressiveness criterion is phrased as follows.

\begin{defn}
We say that $\mathcal{L}$ is
\begin{enumerate}
\item \emph{depth}-0 \emph{separating} if the family 
$(\llbracket c\rrbracket\colon M_01\rightarrow\Omega)_{c\in\Theta}$ 
is jointly order-reflecting;

\item \emph{depth}-1 \emph{separating} if for every canonical $M_1$-algebra~$A$ and every jointly order-reflecting set $\mathfrak{A}$ of $M_0$-homomorphisms 
$A_0\rightarrow\Omega$ that is closed under the propositional operators in $\mathcal{O}$ 
(i.e.~$\llbracket p\rrbracket\langle f_1,\dots, f_k\rangle\in\mathfrak{A}$ if $f_1,\dots, f_k\in\mathfrak{A})$, 
the set
\[
\Lambda(\mathfrak{A}):=\{\llbracket L\rrbracket(f)\colon A_1\rightarrow\Omega\mid L\in\Lambda, f\in\mathfrak{A}\}
\]
is also jointly order-reflecting.
\end{enumerate}
\end{defn}

\begin{theorem}\label{thm:expr}
  Let $\mathcal{L}$ be a depth-0 separating and depth-1 separating
  graded logic.  Then, for every $n\in\omega$, the set of evaluations
  $\llbracket\varphi\rrbracket\colon M_n1\rightarrow\Omega$ of
  formulae~$\phi$ of uniform depth~$n$ is jointly order-reflecting. In
  particular, $\mathcal{L}$ is expressive.
\end{theorem}

\noindent We will apply our expressiveness criterion to the graded logics introduced 
in \autoref{E:logics}. In most cases, this amounts to showing
depth-1 separation only.

%\paragraph{Coalgebraic $\omega$-simulation}
\paragraph{Coalgebraic Modal Logic on $\Pos$}\smnote{Use the same
  headlines as in \autoref{E:logics}; otherwise confusing for the reader.}
Recall that the graded monad $M_n=G^n$ captures (coalgebraic) similarity. As
noted, $M_0$-algebras are just posets and modalities in the described
graded logic are maps $L: G\mathbbm{2}\rightarrow\mathbbm{2}$. Since we 
assume that~$G$ is enriched, it follows from the enriched Yoneda
Lemma that such maps are equivalent to \emph{monotone predicate
  liftings}, which in the present context we understand as
operations~$\lambda$ that lift monotone predicates $X\to\mathbbm{2}$
to monotone predicates $GX\to\mathbbm{2}$, subject to a naturality
condition. In generalization of the discrete notion of
separation~\cite{Schroder08,Pattinson04}, we say that~$\Lambda$,
understood as a set of monotone predicate liftings in this sense, is
\emph{separating} if whenever $t,s\in GX$ are such that
$\lambda(f)(t)\le\lambda(f)(s)$ for all $\lambda\in\Lambda$ and all
monotone $f:X\to\mathbbm{2}$, then $t\le s$; this is precisely the
notion of separation used by Kapulkin et al.~\cite{KKV12}. E.g.\ the
standard diamond modality on~$\pow_{\omega}^{\down} $, i.e.~the
predicate lifting $\Diamond$ given by $\Diamond(S)=\top$ iff
$\top\in S$, is separating. In this terminology, we have
\begin{theorem}\label{thn:coalg-sim-expr}
  If the set $\Lambda$ of monotone predicate liftings for~$G$ is
  separating and~$G$ is finitary and enriched, and preserves
  embeddings and surjections, then the logic $\mathcal{L}$ arising by
  taking $\top$ as the only truth constant and disjunction and
  conjunction as propositional operators is depth-0 separating and
  depth-1 separating, hence expressive for coalgebraic simulation.
\end{theorem}
\noindent That is, we obtain the existing expressiveness result for
coalgebraic modal logic on $\Pos$~\cite{KKV12} by instantiation of
\autoref{thm:expr}.

\paragraph{Bisimilarity} It is straightforward to instantiate
\autoref{thn:coalg-sim-expr} to the case where
$G=C_\omega(\A\times(-))$ and
$\Lambda=\{\Box_a,\Diamond_a\mid a\in\A\}$, with semantics as in
\autoref{E:logics}; in this case, we obtain the classical
Hennessy-Milner theorem (slightly generalized to $\A$-LTS on $\Pos$).

\paragraph{Similarity} We can restrict the graded logic for simulation
(\autoref{E:logics}) to a logic $\mathcal{L}_{\mathsf{Sim}}$ featuring
only conjunction, diamonds, and truth. Depth-0 separation is trivial.
\begin{proposition}\label{prop:simulation}
  The logic $\mathcal{L}_{\mathsf{Sim}}$ is depth-1 separating. 
\end{proposition}
\noindent By \autoref{thm:expr}, we recover the well-known result that
$\mathcal{L}_{\mathsf{Sim}}$ is still expressive for simulation,
improving on the corresponding instance of
\autoref{thn:coalg-sim-expr}, where the logic needed to include also
disjunction. The same holds for the characteristic logics for finer
simulation-type equivalences such as ready similarity as introduced
above. % , the key being that $\lozenge_{I,a}$ requires the ready set to
% be \emph{exactly}~$I$.

\paragraph{Synchronous Bisimilarity} Recall the graded logic for
synchronous bisimilarity from \autoref{E:logics}. Again, depth-$0$
separation is trivial. It is easy to show that depth-$1$ separation
holds as well, so we obtain that the logic is expressive by
\autoref{thm:expr}.\lsnote{@Lutz: Add proof to appendix}

\paragraph{Probabilistic Trace Inclusion}%
\takeout{% Previous short text taken out
The application of our expressiveness condition the graded logic
$\mathcal{L}_{\mathsf{Prob}}$ for probabilistic trace inclusion is
completely analogous to the proof given for probabilistic trace
equivalence in the full version of~\cite{DMS19}; indeed, the proof of
depth-1 separation becomes slightly easier since the underlying
predicates preserve subconvex instead of only convex
combinations. Like in the set-based case, having subconvex
combinations as propositional operators in the logic is not necessary
to obtain expressiveness.}% 

We will show that the logic $\mathcal{L}_{\mathsf{Prob}}$ is
expressive for probabilistic trace inclusion on probabilistic
transition systems, which we construe as coalgebras for
$\sdist(\A\times -)$.  Recall $\sdist X$ is the poset of finite
subdistributions $\mu\colon X\to [0,1]$, represented as formal sums
$\sum_{i=1}^n p_i\cdot x_i$ where $n\in\omega$, $p_i\in(0,1]$, and
$\sum_{i=1}^np_i\leq 1$; its partial ordering $\preccurlyeq$\smnote{I
  changed this from $\leq$; we should use the symbol from Thm.~V.3.}
puts $\mu\preccurlyeq\nu$ if and only if there exist subdivisions
$s_{\mu}$ and $s_{\nu}$ such that $s_{\mu}$ is obviously below
$s_{\nu}$ (see \autoref{d:sdistorder}).

By \autoref{thm:expr}, it suffices to show that
$\mathcal{L}_{\mathsf{Prob}}$ is depth-0 and depth-1
separating. Depth-0 separation is immediate:
$\mathcal{L}_{\mathsf{Prob}}$ has the truth object $\Omega= [0,1]$ and
a unique truth constant $\top$ with evaluation
$\llbracket\top\rrbracket = \id_{[0,1]}$, which is obviously
order-reflecting.

We will now see that $\mathcal{L}_{\mathsf{Prob}}$ is depth-1
separating, whence expressive, as expected. The following fact
about estimating subdistributions is immediate from the description of
the ordering $\preccurlyeq$ on $\sdist X$ given in the previous
section.

\begin{lemma}\label{l:keyexpressive}
  Let $\mu$ and $\nu$ be subdistributions of the form
  $\mu=\sum_i\sum_j p_{ij}\cdot x_{ij}$ and
  $\nu= \sum_i\sum_j q_{ij}\cdot x_{ij}$ (with ranges of indices left
  implicit) such that
\[ \textstyle
\sum_j p_{ij}\cdot x_{ij}\preccurlyeq \sum_j q_{ij}\cdot x_{ij}
\]
for all $i$. Then $\mu\preccurlyeq\nu$.
\end{lemma}
\begin{proof}
  By the definition of $\preccurlyeq$ (\autoref{d:sdistorder}), the subdistributions
  $\sum_j p_{ij}\cdot x_{ij}$ and $\sum_j q_{ij}\cdot x_{ij}$ have
  subdivisions $d_i$ and $d'_i$, respectively, such that
  $d_i\sqsubseteq d'_i$ for each~$i$. We can combine these into
  subdivision $d'$ and $d$ of~$\mu$ and~$\nu$, respectively, such
  that $d\sqsubseteq d'$, showing $\mu\preccurlyeq\nu$.
\end{proof}
\begin{proposition}
  The logic $\mathcal{L}_{\mathsf{Prob}}$ is depth-1 separating.
\end{proposition}
\begin{proof}
  Given a canonical $M_1$-algebra $A$ with main structure map
  $a^{10}\colon \sdist(\A\times A_0)\to A_1$ and a jointly
  order-reflecting family $\FA$ of subconvex algebra morphisms
  $A_0\to[0,1]$ (closure under propositional combinations will not
  actually be needed), we will show that the family
  \[
    \Lambda(\FA)
    =
    \{\llbracket \langle a\rangle \rrbracket(f)\colon A_1\to [0,1] \mid a\in\A, f\in\FA\}
  \]
  is jointly order-reflecting.

  Given $z, z' \in A_1$ such that
  $\llbracket\langle a\rangle\rrbracket(f)(z)\leq \llbracket\langle
  a\rangle\rrbracket(f)(z')$ holds in $\Omega = [0,1]$ for all
  $a\in\A$ and all $f\in\FA$, we are to prove that $z \leq z'$. First,
  note that the map $a^{10}\colon \sdist(\A\times A_0)\to A_1$ is a
  coequalizer by \autoref{P:coeq} since $A$ is canonical, whence
  $a^{10}$ is surjective. Thus, there exist subdistributions
  $\mu, \nu \in \sdist(\A \times A_0)$ with $a^{10}(\mu) = z$ and
  $a^{10}(\nu) = z'$. Since $a^{10}$ is monotone, it suffices to show that
  $\mu \preccurlyeq \nu$ holds. Using the commutativity
  of~\eqref{Di:modality}, we obtain
  %(note that the sums below and in  the following are finite since $\mu(a,x) \neq 0$ for only finitely many pairs $(a,x)$)
  \begin{align*}
    \llbracket \langle a\rangle\rrbracket (f)(z)
    &=
    \llbracket \langle a\rangle\rrbracket (f)(a^{10}(\mu))
    \\
    &=
    \llbracket\langle a\rangle\rrbracket(M_1f (\mu))
    \\
    &\textstyle=
    \llbracket\langle a\rangle\rrbracket\big(\sum_{(a,x)\in \A \times A_0}\mu(a, x)\cdot (a,f(x))\big)
    \\
    &\textstyle=
    \sum_{x\in A_0}\mu(a, x)\cdot f(x),
  \end{align*}
  and similarly for $\nu$.
  \takeout{% old text
    Thus every element of $A_1$ is represented
    by a subdistribution on $\A\times A_0$ which we cast as a subconvex
    combination of elements of the shape $a(x)$ where
    $(a,x)\in\A\times A_0$. Now, given subdistributions
    $\mu,\nu\in\sdist(\A\times A_0)$ such that
    $\llbracket\langle a\rangle\rrbracket(f)(\mu)\leq \llbracket\langle
    a\rangle\rrbracket(f)(\nu)$ holds in $A_1$ for all $a\in\A$ and all
    $f\in\FA$, we will show that $\mu\preccurlyeq\nu$.
    Unwinding our definitions we see that\smnote{Since I don't see it,
      so wouldn't a hostile reviewer; let's explicitly demonstrate the
      unwinding!}
  }% end takeout
  Hence we see that the following inequality holds in
  $[0,1]$\smnote{I strongly propose that we add where each of those
    inequality holds; otherwise the whole reasoning is harder to follow.} for
  every $f\in\FA$ and every $a\in\A$:
  \begin{equation}\label{Eq:assumption}
    \textstyle
    \sum_{x\in A_0}\mu(a, x)\cdot f(x)\leq\sum_{x\in A_0}\nu(a, x)\cdot f(x).
  \end{equation}
  In particular, we have
  \begin{align*}
    &\textstyle\phantom{=}\ f(\sum_{x\in A_0}\mu(a,x)\cdot x) \\
    &\textstyle= \sum_{x\in A_0}\mu(a,x)\cdot f(x) \tag{$f$ a homomorphism} \\
    &\textstyle\leq \sum_{x\in A_0}\nu(a, x)\cdot f(x) \tag{by \eqref{Eq:assumption}} \\
    &\textstyle= f(\sum_{x\in A_0}\nu(a, x)\cdot x) \tag{$f$ a homomorphism}
  \end{align*}
  Since $\FA$ is jointly order-reflecting, it follows that the 
  inequalities below hold in $A_0$ for all for all $a\in\A$:
  \[\textstyle
    \sum_{x\in A_0}\mu(a, x)\cdot x
    \leq
    \sum_{x\in A_0}\nu(a, x)\cdot x.
  \]
  By applying the monotone operation $a(-)$ to both sides of this
  inequality and using that actions distribute over subconvex
  combinations we obtain that for all $a\in\A$ we have the following
  inequalities in $A_1$:
  \[\textstyle
    \sum_{x\in A_0}\mu(a, x)\cdot a(x) \leq \sum_{x\in A_0}\nu(a, x)\cdot a(x).
  \]
  Since these inequalities were derived in the theory, we can
  represent them in $\sdist(\A\times A_0)=M_1A_0$ as follows (using
  surjectivity of $a^{10}\colon \sdist(\A\times A_0)\to A_1$):
  \smnote{I expanded on this step in the proof because I thought it
    was not written in a clear way. Please do not shorten.}
  \[
    \textstyle
    \sum_{x\in A_0}\mu(a, x)\cdot (a,x)
    \preccurlyeq
    \sum_{x\in A_0}\nu(a, x)\cdot (a,x).
  \]
  Hence, by \autoref{l:keyexpressive}, we conclude with the desired 
  \[%\textstyle
    \mu
    =
    \sum_{a\in A}\sum_{x\in A_0}\mu(a, x)\cdot a(x)
    \preccurlyeq
    \sum_{a\in A}\sum_{x\in A_0}\nu(a, x)\cdot a(x)
    =
    \nu.\qedhere
  \]
%  as desired.
\end{proof}

%%%%%%%%%%%%%%%%%%%%%%%%%%%%%%%%%%%%%%%%%%%%%%%%%%
%CONCLUSION
%%%%%%%%%%%%%%%%%%%%%%%%%%%%%%%%%%%%%%%%%%%%%%%%%%
%
\section{Conclusion}
\noindent We have introduced a generic framework for behavioural
preorders that combines coalgebraic parametrization over the system
type (nondeterministic, probabilistic, etc.) with a parametrization
over the granularity of system semantics. The latter is afforded by
mapping the type functor into a graded monad on the category $\Pos$ of
partially ordered sets. The framework includes support for designing
characteristic logics; specifically, it allows for identifying, in a
straightforward manner, propositional operations and modalities that
automatically guarantee preservation of formula satisfaction under the
behavioural preorder, and it offers a readily checked criterion for
the converse implication, standardly referred to as expressiveness.
Important topics for future research include the development of
generic minimization and learning algorithms modulo a given graded
semantics, as well as a generic axiomatic treatment of graded
logics.

% ---- Bibliography ----
%
% Alternative 1
%
%\clearpage %% SM: remove only if we want to start references on p.12
%\IEEEtriggeratref{22} % column break within reference list
%
% Alternative 2
%
%
\IEEEtriggeratref{25} % column break within reference list
%
% end alternatives
%
\bibliographystyle{IEEEtranS}
\bibliography{gbp-via-graded-monads}

\clearpage
\appendix

%\section{Omitted Proof Details}

\subsection{Details for \autoref{S:LoIC}}

\begin{proof}[Proof of \autoref{L:subterms}]
  By induction on derivations of $\Gamma\vdash_k s\leq t$; we restrict
  our attention to the inductive step corresponding to derivations of
  $\Gamma\vdash_k s\leq t$ in which the last rule applied was
  $(\mathsf{Ax1})$. To this end, we inductively assume that whenever
  $\Gamma'\vdash_{k'} s'\leq t'$ appears in the premise of the last rule
  applied in the derivation of $\Gamma\vdash_k s\leq t$ and 
  $\sigma'(f')\in\sub(s', t')$ where 
  $f'\colon\arity(\sigma')\to\mathsf{T}_{\Sigma, n'}(\Gamma')$, 
  then $\Gamma'\vdash_{n'} f'(i)\leq f'(j)$ is derivable for all 
  $i\leq j$ in $\arity(\sigma')$.

  Now, given a uniform substitution
  $\gamma\colon |\Delta| \to \mathsf{T}_{\Sigma,\ell}(\Gamma)$ 
  and an axiom $\Delta\vdash \overline{s}\leq\overline{t}$ such 
  that $s=\bar\gamma(\overline{s})$ and $t=\bar \gamma(\overline{t})$, 
  we verify that $\Gamma\vdash_n t_i\leq t_j$ for all $i\leq j$ in
  $\arity(\sigma)$. %Without loss of generality, $\sigma(f)\in\sub(\gamma(\bar s))$. 
  We make a case distinction on the following basis:
  \begin{enumerate}
  \item $\sigma(f)=\gamma(x)$ for some $x\in\Delta$ or 
  \item\label{L:subterms:2} $\sigma(f)=\bar \gamma(\sigma(g))$ for some $g\colon|\arity(\sigma)|\to
    \mathsf{T}_{\Sigma, m}(\Delta)$ such that $m+\ell=n$ and 
    $\sigma(g)\in\sub(\bar s, \bar t)$. %That is, $f=\overline{\gamma}\cdot g$.
  \end{enumerate}
  Given $x\in\Delta$ such that $\sigma(f)=\gamma(x)$, note that
  $\Delta\vdash_0 x\leq x$ is derivable by application of
  $(\mathsf{Var})$ since the underlying order of $\Delta$ is
  reflexive. In particular, $\Gamma\vdash_k \gamma(x)\leq \gamma(x)$
  appears in the premise of the last rule applied in the derivation of 
  $\Gamma\vdash_k s\leq t$. Hence also $\Gamma\vdash_n f(i)\leq f(j)$ is 
  derivable for all $i\leq j$ in $\arity(\sigma)$ by induction.

  Otherwise, \ref{L:subterms:2} holds. Then, for the uniform substitution
  $\gamma\colon |\Delta| \to \mathsf{T}_{\Sigma,\ell}(\Gamma)$,
  the axiom $\Delta\vdash \overline{s} \leq \overline{t}$, and 
  $\sigma(g)\in\sub(\overline{s},\overline{t})$, the condition $(**)$
  of $(\mathsf{Ax2})$ applies. Thus, for all $i\leq j$ in $\arity(\sigma)$,
  $\Gamma\vdash_{n}\bar\gamma (g(i))\leq \bar\gamma (g(j))$
  is derivable via $(\mathsf{Ax2})$. As 
  $\sigma(f)=\bar\gamma(\sigma(g)) =\sigma(\bar\gamma\cdot g),$ we
  see that $f=\bar\gamma\cdot g$. Thus $\Gamma\vdash_n f(i)\leq f(j)$ 
  is derivable for all $i\leq j$ in $\arity(\sigma)$, as desired.
\end{proof}

\begin{proof}[Proof of \autoref{P:subterms}]
We distinguish cases on the basis of whether $u\in\sub(s, t)$ is a variable
in context $\Gamma$ or $u=\sigma(g)$ for some operation $\sigma$ 
and some $g\colon|\arity(\sigma)|\to\mathsf{T}_{\Sigma, n}(\Gamma)$. If the 
former holds, then $\Gamma\vdash_0 \down u$ is derivable via an application 
of $(\mathsf{Var})$ since $\Gamma$ is reflexive. If the latter holds, then 
$\Gamma\vdash_n g(i)\leq g(j)$ is derivable for all $i\leq j$ in $\arity(\sigma)$
by~\autoref{L:subterms}. In particular, 
$\Gamma\vdash_{n+d(\sigma)} \down \sigma(g)$ is derivable via 
$(\mathsf{Ar})$, as desired.
\end{proof}

\subsection{Details for \autoref{S:freemodel}}

\begin{lemma}[Subsitution Lemma]\label{L:subst}
  Let $A$ be a $(\Sigma,n)$-algebra, $\Gamma$ a context and
  $\iota\colon \Gamma \to A_m$ monotone. For every monotone map
  $\gamma\colon \Delta \to \termst k (\Gamma)$ and $t$
  in $\termst \ell (\Delta)$ we have
  \[
    \iota^\#_{\ell+k}(\ol\gamma(t)) = (\iota^\#_k \cdot \gamma)^\#_\ell(t)
  \]
  (in particular, both sides are defined). That is, the following
  diagram commutes:
  \begin{equation*}
    \begin{tikzcd}[column sep = 35]
      \termst \ell(\Delta)
      \ar{r}{(\ell^{\#}\cdot\gamma)^{\#}_{\ell}}
      \ar{d}[swap]{\bar \gamma}
      &
      A_{m+k+\ell}
      \\
      \termst {\ell+k}(\Gamma)
      \ar{ru}[swap]{\ell^{\#}_{\ell+k}}  
    \end{tikzcd}
  \end{equation*}
  %commutes. % SM: no trailing single words after math display please.
\end{lemma}
\begin{proof}
  That both sides of the desired equation are defined follows from
  \autoref{C:arities} and \autoref{P:subterms}. We proceed by
  induction on terms. For $t = x \in \Delta$ we have 
    \[
      \iota^\#_\ell(\ol\gamma(x))
      =
      \iota^\#_\ell(\gamma(x)) 
      =
      (\iota^\#_\ell  \cdot \gamma)^\#_\ell(x).
    \]
    For $t = \sigma(f)$ for an operation $\sigma$ and a monotone
    $f\colon \arity(\sigma) \to \terms_{\T,p}(X)$ with $\ell =
    p+d(\sigma)$ we have
    \begin{align*}
      \iota^\#_{\ell+k}(\ol\gamma(t))
      &= \iota^\#_{\ell+k}(\ol\gamma(\sigma(f)))
      & \text{since $t= \sigma(f)$}\\
      &= \iota^\#_{\ell+k} (\sigma(\ol\gamma \cdot f))
      & \text{def.~of  $\ol{\gamma}$} \\
      &= \sigma^A_{k+p}\big(\iota^\#_{k+d(\sigma)}(\ol\gamma \cdot f)\big)
      & \text{def.~of $\iota^\#$} \\
      &= \sigma^A_{k+p}\big((\iota^\#_k \cdot \gamma)^\#_p \cdot f\big)
      & \text{by induction} \\
      &= (\iota^\#_k \cdot \gamma)^\#_\ell\big(\sigma(f)\big)
      & \text{def.~of $\iota^\#$}\tag*{\qedhere}
    \end{align*}    
\end{proof}
\begin{proof}[Proof of \autoref{T:sound}]
  Let $A$ be any $(\T,n)$-model.
  \begin{enumerate}
  \item Soundness follows from the fact, proved in
    item~\ref{T:sound:2} below that for every monotone
    $\iota\colon X \to A_m$ the components
    $\iota^\#_p\colon \mathsf{T}_{\Sigma,p}(X) \to A_{m+p}$ of the
    evaluation map restrict to total monotone maps
    $\terms_{\T,p} \to A_{m+p}$. Indeed, suppose that the inequation
    $X \vdash_p t_1 \leq t_2$ is derivable. By \autoref{P:subterms} we
    know that $t_1, t_2$ lie in $\terms_{\T,p}(X)$ and satisfy
    $t_1 \leq t_2$ there. Then for every monotone
    $\iota\colon X \to A_m$ both $\iota^\#_k(s)$ and $\iota^\#_k(t)$
    are defined and we have $\iota^\#_k(s) \leq \iota^\#_k(t)$. Thus
    $A$ satisfies the given inequation in context.
    
  \item\label{T:sound:2} Given a monotone map $\iota\colon X \to A_m$,
    we prove that the components
    $\iota^\#_p\colon \mathsf{T}_{\Sigma,p}(X) \to A_{m+p}$ of the
    evaluation map restrict to total monotone maps
    $\terms_{\T,p} \to A_{m+p}$. Suppose that $t_1 \leq t_2$ holds in
    $\termst p(X)$; that is $X \vdash_p t_1 \leq t_2$ is derivable. We
    prove by induction on its derivation that $\iota^\#_p(t_1)$ and
    $\iota^\#_p(t_2)$ are both defined and that
    $\iota^\#_p(t_1) \leq \iota^\#_p(t_2)$.

    We proceed by case distinction on the last rule applied in the
    derivation. The cases of the $(\mathsf{Var})$, $(\mathsf{Ar})$,
    $(\mathsf{Trans})$ and $(\mathsf{Mon})$ rules are easy using
    monotonicity of the operations $\sigma^A_m$ and the respective
    induction hypotheses as well as transitivity of the order on $A_p$
    in the last case.

    For the $(\mathsf{Ax1})$ rule let $\Delta \vdash_n s \leq t$ be an
    axiom of $\T$ and $\gamma\colon |\Delta| \to \mathsf{T}_{\Sigma,k}(X)$ a uniform
    substitution so that $p = n+k$, $t_1 = \ol\gamma(s)$ and
    $t_2 = \ol\gamma(t)$. The assumption of the rule states that
    $\Gamma \vdash_k \gamma(x) \leq \gamma(y)$ for all $x \leq y$ in
    $\Delta$, which implies that $\gamma$ restricts to a monotone map
    $\Delta \to \terms_{\T,k}(X)$. The induction hypothesis
    states that for every $x\leq y$ in $\Delta$ we have
    $\iota^\#_k(\gamma(x)) \leq \iota^\#_k (\gamma(y))$. We conclude
    by the following computation in which all terms are equidefined: 
    \begin{align*}
      \iota^\#_p(t_1)
      &= \iota^\#_p(\ol\gamma(s)) &
      \text{since $t_1 = \ol\gamma(s)$}\\
      &= (\iota^\#_k \cdot \gamma)^\#_m(s)
      & \text{by~\autoref{L:subst}} \\
      &\leq (\iota^\#_k\cdot \gamma)^\#_m(t)
      & \text{$A$ satisfies $\Delta \vdash_n s \leq t$} \\
      &=\iota^\#_p(\ol\gamma(s))
      & \text{by~\autoref{L:subst}} \\
      &= \iota^\#_p(t_2) & \text{since $\ol\gamma(t) = t_2$}.
    \end{align*}
    
    For the $(\mathsf{Ax2})$, we argue similarly. Let
    $\gamma\colon |\Delta| \to\mathsf{T}_{\Sigma,k}(X)$ be a uniform
    substitution and $\Delta\vdash_n s\leq t$ be an axiom such that
    there exist an operation symbol $\sigma$ in $\Sigma$ and a map
    $f\colon|\arity(\sigma)|\to \mathsf{T}_{\Sigma, m}(\Delta)$ such
    that $\sigma(f)\in\sub(s, t)$ and, for some $i\leq j$ in
    $\arity(\sigma)$, we have $u=f(i)$ and $v=f(j)$. Then
    $p = m+k$, $t_1 = \ol\gamma(u)$ and $t_2 = \ol\gamma(v)$. Again,
    the assumption of the rule states that $\gamma$ restricts to a
    monotone map $\Delta \to \terms_{\T, k}(X)$. The induction
    hypothesis states that for every $x \leq y$ in $\Delta$ we have
    $\iota^\#_k(\gamma(x)) \leq \iota^\#_k(\gamma(y))$. Since $A$
    satisfies $\Delta \vdash_n s \leq t$ we know that for every
    monotone $h\colon \Delta \to A_k$ we have that $h^\#_n(s)$
    and $h^\#_n(t)$ are defined (and
    $h^\#_n(s) \leq h^\#_n(t)$). Unravelling the definition of
    $h^\#$ we obtain that $h^\#_{m+d(\sigma)}(\sigma(f))$ must
    be defined, which in turn implies that $h^\#_m(u)$ and
    $h^\#_m(v)$ are defined and
    \begin{equation}\label{eq:iota-uv}
      h^\#_m(u) \leq h^\#_m(v)
    \end{equation}
    Again we conclude by a computation in which all terms are
    equidefined: 
    \begin{align*}
      \iota^\#_p(t_1
      &= \iota^\#_p(\ol\gamma(u)) &
      \text{since $t_1 = \ol\gamma(u)$}\\
      &= (\iota^\#_k \cdot \gamma)^\#_m(u)
      & \text{by~\autoref{L:subst}} \\
      &\leq (\iota^\#_k\cdot \gamma)^\#_m(v)
      & \text{by~\eqref{eq:iota-uv} for $h = \iota^\#_k \cdot \gamma$} \\
      &=\iota^\#_p(\ol\gamma(v))
      & \text{by~\autoref{L:subst}} \\
      &= \iota^\#_p(t_2) & \text{since $\ol\gamma(v) = t_2$}.\tag*{\qedhere}
    \end{align*}
  \end{enumerate}
\end{proof}

\begin{nota}
  \begin{enumerate}
  \item We denote the equivalence class $[t]_k$
    by $[t]$ (i.e.~we omit the subscript) whenever confusion is unlikely.
  \item We write $\Gamma \vdash_k s = t$ as a shorthand notation for
    derivable equality (i.e.~the conjunction of $\Gamma \vdash_k s
    \leq t$ and $\Gamma \vdash_k t \leq s$).
  \end{enumerate}
\end{nota}
\begin{proof}[Proof of~\autoref{L:welldefined}]
  \begin{enumerate}
  \item\label{L:welldefined:1} We first show that $\sigma^{FX}_k(f)$ is defined for all
    $f\in\Pos(\ari, (FX)_k)$. For every $i \leq j$ in $\ari$ we have
    $f(i) \leq f(j)$ in $(FX)_k$, which implies that
    $X\vdash_k u_k(f(i))\leq u_k(f(j))$ is derivable. This fact is
    independent of the choice of $u_k$. For if $s \sim u_k(f(i))$ and
    $t \sim u_k(f(j))$, then we have $X \vdash_k s \leq t$ by two
    applications of $(\mathsf{Trans})$. Hence
    $X\vdash_{k+d(\sigma)}\down \sigma(u_k \cdot f)$ is
    derivable via $(\mathsf{Ar})$.

  \item We will now show that $[\sigma(u_k \cdot f)]$ is independent
    of the choice of $u_k$. So suppose that $v_k$ is another splitting
    of $q_k$, then we have $X\vdash_k u_k(f(i)) = v_k(f(i))$ for all
    $i \in \ari$. By two applications of $\mathsf{(Mon)}$ we obtain
    that
    $X\vdash_{k+d(\sigma)} \sigma(u_k\cdot f) = \sigma(v_k\cdot f)$,
    which implies $[\sigma(u_k\cdot f)] = [\sigma(v_k \cdot f)]$ as
    desired.

  \item Finally, we prove that $\sigma^{FX}_k$ is monotone for every
    $k \in \omega$. Given $f, g \in \Pos(\ari, (FX)_k)$ with $f \leq
    g$. By item~\ref{L:welldefined:1} we have
    \[
      X \vdash_k \down \sigma(u_k\cdot f)
      \quad\text{and}\quad
      X \vdash_k \down \sigma(u_k \cdot g),
    \]
    and similarly we have that for all $i \in \ari$, $f(i) \leq g(i)$
    implies that $X\vdash_k u_k\cdot f(i) \leq u_k\cdot g(i)$ is
    derivable. Thus, an application of $(\mathsf{Mon})$ yields
    $X\vdash_{k+d(\sigma)} \sigma(u_k\cdot f) \leq \sigma(u_k\cdot g)$, which
    implies $[\sigma(u_k \cdot f)] \leq [\sigma(u_k \cdot
    g)]$. \qedhere
  \end{enumerate}    
  \takeout{%%SM: This is notationally of because $f = g$ are the same map right from
           %%the start so it proves nothing.
  We will now show that $\sigma^{FX}_k$ is well
  defined w.r.t.~equivalence classes. That is, given
  $f,g\in\Pos(\ari, (FX)_k)$ such that $f(i)=g(i)$ for all $i\in\ari$,
  we verify that $\sigma^{FX}_k(f)=\sigma^{FX}_k(g)$.  To this end,
  observe that $f(i)=g(i)$ implies that $X\vdash_{k} f(i)\leq g(i)$ is
  derivable for all $i\leq j$. Moreover,
  $X\vdash_{k+d(\sigma)}\down \sigma(u_k\cdot f)$ and
  $X\vdash_{k+d(\sigma)}\down \sigma(u_k\cdot g)$ are derivable
  by the argument above. Hence,
  $X\vdash_{k+d(\sigma)}\sigma(u_k\cdot f)\leq\sigma(u_k\cdot g)$ is
  derivable via $(\mathsf{Mon})$. Similarly,
  $X\vdash_{k+d(\sigma)}\sigma(u+k\cdot f)\leq\sigma(u_k\cdot g)$ is
  derivable. Hence
  \[
    \sigma^{FX}_{k}(f)=[\sigma(u_k\cdot f)] =[\sigma(u_k\cdot g)] = \sigma^{FX}_k(g),
  \]
  as desired.}% end takeout
\end{proof}

\takeout{%% SM: attempt that is possibly not needed
\begin{defn}
  Let $n \leq \omega$, $\T$ a graded theory, and $A$ a
  $(\T,n)$-model. For every monotone map $h\colon  \to A_k$ its
  \emph{extension to depth-$k$ defined terms} is the following family
  of maps $\ext h_p\colon \terms_{\T,p}(X) \to A_{k+p}$ defined
  recursively as follows:
  \begin{itemize}
  \item $\ext h_k(x) = h(x)$ for every $x \in X$, and
  \item $\ext h_p(\sigma(f))  = \sigma^A_m(\ext h_{k+n}\cdot f)$ for
    every $\sigma \in \Sigma(\ari,m)$ and every monotone map $f\colon
    P \to \terms_{\T,k+n}(X)$.
  \end{itemize}
\end{defn}

\begin{lemma}
  The maps $\ext h_p\colon \terms_{\T,p}(X) \to A_{k+p}$ are
  monotone.
\end{lemma}
\begin{proof}
  Suppose that $t_1 \leq t_2$ holds in $\terms_{\T,p}(X)$ for some
  $p \in \omega$. Eqivalently $X \vdash_p t_1 \leq t_2$. We prove by
  induction on its derivation that this implies
  $\ext h_p(t_1) \leq \ext h_p(t_2)$. We proceed by case distinction
  on the last rule applied in the derivation of
  $X \vdash_p t_1 \leq t_2$. The cases for the $(\mathsf{Var})$,
  $(\mathsf{Ar})$, $(\mathsf{Trans})$ and $(\mathsf{Mon})$ rules are
  easy using monotonicity of the operations $\sigma^A_m$ and the
  respective induction hypotheses as well as transitivity of the order
  on $A_p$ in the last case.

  For the $(\mathsf{Ax1})$ rule let $\Delta \vdash_n s \leq t$ be an
  axiom of $\T$ and
  $\gamma\colon |\Delta| \to \mathsf{T}_{\Sigma,k}(X)$ a uniform
  substitution so that $p = n+k$, $t_1 = \ol\gamma(s)$ and
  $t_2 = \ol\gamma(t)$. The assumption of the rule states that
  $\Gamma \vdash_k \gamma(x) \leq \gamma(y)$ for all $x \leq y$ in
  $\Delta$, which implies that $\gamma$ restricts to a monotone map
  $\Delta \to \terms_{\T,k}(X)$.  The induction hypothesis states
  that for every $x\leq y$ in $\Delta$ we have
  $\ext h_k(\gamma(x)) \leq \ext h_k (\gamma(y))$.  We will prove
  below that for every subterm $r \in \terms_{\T,m}(\Delta)$ of
  $s$ or $t$ we have
  \begin{equation}\label{eq:hsharp}
    \ext h_{m+k} (\ol\gamma(r)) = \ext{(\ext h_k\cdot \gamma)}_m(r), 
  \end{equation}
  which allows us to conclude
  \begin{align*}
    h^\sharp(t_1)
    &= h^\sharp_p[\ol\gamma(s)] \\
    &= h^\#(s)
    & \text{by~\eqref{eq:hsharp}} \\
    &\leq h^\#(t)
    & \text{$A$ satisfies $\Delta \vdash$}
  \end{align*}
\end{proof}}% end takeout

\takeout{%% SM: useful fact that we do not need
\begin{rem}
  Even though $(\termst m(X))$ do not form a $(\Sigma, n)$-algebra (since
  they do not form a poset), the canonical quotient maps $q_m\colon
  \termst m (X)\epito (FX)_m$ are still ``homomorphic''. That is, for every
  operation $\sigma$ and every monotone $f\colon \arity(\sigma) \to
  \termst m (X)$ we have that
  \[
    \sigma^{FX}_m (q_m \cdot f) = q_{m+d(\sigma)}(\sigma(f)).
  \]
  To see this we use independence of the definition of $\sigma^{FX}_m$
  from the choice of the splitting $u_m$ of $q_m$. In particular,
  given $\sigma$ and $f$, we choose $u_m$ such that
  $u_m(q_m(f(i)) = f(i)$ for all $i \in\arity(\sigma)$; that means we
  have
  \begin{equation}\label{eq:uq}
    u_m \cdot q_m \cdot f = f. 
  \end{equation}
  Then we have
  \begin{align*}
    \sigma^{FX}_m (q_m \cdot f)
    &= [\sigma(u_m \cdot q_m \cdot f)]
    &\text{def.~of $\sigma^{FX}$}\\
    &= [\sigma(f)]
    &\text{by~\eqref{eq:uq}}\\
    &= q_{m+d(\sigma)} (\sigma(f)).
  \end{align*}
\end{rem}}% end takeout useful remark

\begin{proof}[Proof of~\autoref{P:free}]
  We prove the proposition for $n=\omega$; the details corresponding
  to the case for $n\in\omega$ then follow immediately by making appropriate
  restrictions when necessary. Fix a poset $X$. 
  \begin{enumerate}
  \item We first prove that $FX$ is a $(\T,\omega)$-model. Given an
    axiom $\Delta\vdash_n s\leq t$ of $\T$ and a monotone map
    $\iota\colon\Delta\to (FX)_k$, we verify that both
    $\iota^\#_{n+k}(s)$ and $\iota^\#_{n+k}(t)$ are defined and
    $\iota^\#_{n+k}(s)\leq \iota^\#_{n+k}(t)$.

    We define a uniform substitution $\gamma$ as the map
    \[
      \gamma = \big(
      |\Gamma| \xra{\iota} (FX)_k \xra{u_k} \terms_{\T,k}(X)
      \hookrightarrow \mathsf{T}_{\Sigma,k}(X)
      \big).
    \]
    Since $\iota$ is monotone, we have that
    \[
      \iota^\#_0(x) = \iota(x) \leq \iota(y) = \iota^\#_0(y)
      \quad\text{for all $x\leq y$ in  $X$.}
    \]
    This implies that $X\vdash_k \gamma(x) \leq \gamma(y)$ is
    derivable. By an application of $(\mathsf{Ax1})$ we thus have
    $X \vdash_{n+k} \ol\gamma(s) \leq \ol\gamma(t)$. By
    \autoref{P:subterms} we therefore have
    $X \vdash_{m} \down \ol\gamma(r)$, which implies
    $\ol\gamma(r) \in \terms_{\T,p}(X)$, for every depth-$m$ subterm $r$
    of $s$ or $t$. It remains to prove that for every such $r$, both
    $q_p \cdot \ol\gamma(r) = [\ol\gamma(r)]$ and $\iota^\#_p(r)$ are
    defined and 
    \[
      q_{p}\cdot \ol\gamma(r)
      %[\ol\gamma(r)]
      =
      \iota^\#_{p}(r).
      %\qquad\text{for every depth-$m$ term $r \in \sub(s,t)$.}
    \]
    We proceed by induction on terms. For all $x \in \Delta$ we clearly
    have that both are defined and
    \[
      [\ol\gamma(x)] = q_k (\gamma(x)) = q_k (u_k(\iota(x))) =
      \iota(x) = \iota^\#_k(x).  
    \]
    Now suppose that $p = k + m+ n$,  $\sigma \in \Sigma(\ari, m)$ and
    $f\colon P \to \terms_{\T, n}(X)$ such that every $f(i) \in
    \subs(s,t)$ for every $i \in P$. By induction we know that
    $\ol\gamma(f(i)))$ and $\iota^\#_n(f(i)$ are both defined and
    $[\ol\gamma(f(i)))] = \iota^\#_n(f(i))$ holds for every $i \in P$.
    Thus, we compute
    \begin{align*}
      [\ol\gamma(\sigma(f))]
      &= q_p (\sigma(\ol\gamma \cdot f))
      & \text{def.~of $\ol\gamma$} \\
      &= \sigma^{FX}_n(q_n \cdot \ol\gamma \cdot \ol f)
      & \text{def.~of $\sigma^{FX}_n$} \\
      &= \sigma^{FX}_n(\iota^\#_n \cdot f)
      & \text{by induction} \\
      &= \iota^\#_p(\sigma(f))
      & \text{def.~$\iota^\#_p$.}
    \end{align*}
    \takeout{%%% 
      \smerror{The next sentence makes no sense. This must again involve $u_k$. 
        We will need a
        lemma, namely: $u_k \cdot (\iota^\#_k) = (u_m \cdot
        \iota)^\#_k$, of which the right-hand side is not really defined
        because $\terms_{\T,k}(\Gamma)$ is not a $(\Sigma,k)$-algebra
        since it is not a poset.}%
      Thus $X\vdash_{k+m} \iota^\#_{k+m}(s)\leq \iota^\#_{k+m}(t)$ is
      derivable via $(\mathsf{Ax1})$\cfnote{There is an issue here:
        $\iota$ is not a uniform substitution. SM: Yes, this is a problem!} hence also
      $\iota^\#_{k+m}(s)\leq \iota^\#_{k+m}(t)$.  By the Subterm Rule,
      $\iota^\#_{k+m}(s)$ and $\iota^\#_{k+m}(t)$ are defined.}% end takeout
    
  \item We turn to the the universal property of
    $\eta_X\colon X \to (FX)_0$. First observe that $\eta_X$ is
    clearly monotone since for every $x \leq y$ in $X$ we have
    $X \vdash_0 x \leq y$ by the $(\mathsf{Var})$ rule, whence
    $[x] \leq [y]$ holds in $(FX)_0$. Given $h\colon X\to A_0$, we
    define the family of a maps $h^\sharp_k\colon(FX)_k\to A_{k}$
    using that $q_k\colon \termst k (X) \epito FX$ is the universal
    arrow for $X$ of the reflection of $\Pos$ in the category of
    preordered sets (i.e.~the left adjoint to the inclusion of that
    category in $\Pos$). We know every $h^\sharp_k$ restricts to a
    total monotone map $\termst k (X) \to A_k$ (see the proof of
    \autoref{T:sound}). Thus, there exists a unique monotone map
    $h^\sharp_k\colon (FX)_k \to A_k$ such that $h^\sharp_k \cdot q_k
    = h^\#_k$. In other words, we have 
    \[
      h^\sharp_k [t] = h^\#_k(t)
      \qquad
      \text{for every $t$ in $\termst k(X)$.}
    \]
    We immediately obtain that $h^\sharp_0 \cdot \eta_X = h$. Indeed,
    using the definition of $h^\#_0$ we have, for every variable
    $x \in X$ that
    \[
      h^\sharp \cdot \eta_X(x) = h^\sharp_0[x] = h^\#_0(x) = x
    \]

  \item We now prove that $h^\sharp\colon FX \to A$ is a homomorphism of
    $(\Sigma,n)$-algebras. That is, for every operation $\sigma$ in
    $\Sigma$ we show that the square below commutes:
    \[
      \begin{tikzcd}
        \Pos(\arity(\sigma), (FX)_m)
        \ar{r}{\sigma^{FX}_m}
        \ar{d}[swap]{\Pos(\arity(\sigma), h^\sharp_m)}
        &
        (FX)_{m+d(\sigma)}
        \ar{d}{h^\sharp_{m+d(\sigma)}}
        \\
        \Pos(\arity(\sigma), A_m)
        \ar{r}{\sigma^A_m}
        &
        A_{m+d(\sigma)}
      \end{tikzcd}
    \]
    Given a monotone $f\colon \arity(\sigma) \to (FX)_m$ we compute
    (abbreviating $k = m+d(\sigma)$):
    \begin{align*}
      &h^\sharp_k \cdot \sigma^{FX}_m (f)\\
      &= h^\sharp_k[\sigma(u_m \cdot f)]
      &\text{def.~of $\sigma^{FX}$} \\
      &=h^\#_k(\sigma(u_m \cdot f))
      &\text{def.~of $h^\sharp$} \\
      &=\sigma^A_m(h^\#_m \cdot u_m \cdot f)
      & \text{def.~of $h^\#$} \\
      &= \sigma^A_m(h^\sharp_m \cdot q_m \cdot u_m \cdot f)
      &\text{since $h^\sharp_m \cdot q_m = h^\#_m$}\\
      &=\sigma^A_m(h^\sharp_m \cdot f)
      & \text{since $q_m \cdot u_m = \id_{(FX)_m}$} \\
      &= \sigma^A_m\big(\Pos(\arity(\sigma), h^\sharp_m)(f)\big)
      &\text{def.~of hom-functor} \\
      &= \sigma^A_m \cdot \Pos(\arity(\sigma), h^\sharp_m)(f).
    \end{align*}

  \item It remains to prove that $h^\sharp$ is the unique
    $(\Sigma, n)$-algebra morphism with $h^\sharp_0 \cdot \eta_X = h$.
    Given a of $(\Sigma,n)$-algebra morphism $g\colon FX\to A$ 
    such that $g_0\cdot\eta_X= h$, we verify that
    $g_k[t] = h^\#_k(t)$ for all $k\in\omega$ and all
    $[t]\in\terms_{\T, k}(X)$ by induction on terms. For a variable
    $x$ we have that $k=0$ and
    \[
      g_0[x] = g_0 \cdot \eta_X(x) = h (x) = h^\#_0(x). 
    \]
    Now, suppose that $t = \sigma(f)$ for some operation $\sigma$ and
    $f\colon |\arity(\sigma)| \to \mathsf{T}_{\Sigma,m}$ with $m+d(\sigma)\leq n$. By
    \autoref{C:arities} and \autoref{P:subterms} we can assume that
    $f$ restricts to depth-$m$ $\T$-defined terms and is monotone;
    that is we can take it as a monotone function $f\colon
    \arity(\sigma) \to \terms_{\T,m}(X)$. In the following we assume
    that in the definition of $\sigma^{FX}_m$ we chose a splitting
    $u_m$ of $q_m\colon \terms_{\T,m} \epito (FX)_m$ such that
    $u_m(q_m(f(i)) = f(i)$ for all $i \in\arity(\sigma)$; that means
    we have 
    \begin{equation}\label{eq:uq2}
      u_m \cdot q_m \cdot f = f. 
    \end{equation}
    Abbreviating $n = m + d(\sigma)$ we can now compute
    \begin{align*}
      g_n[t] &= g_n[\sigma(f)] & \text{since $t = \sigma(f)$} \\
      &= g_n [\sigma(u_m \cdot q_m \cdot f)]
      & \text{by~\eqref{eq:uq2}}\\
      &= g_n (\sigma^{FX}_m(q_m \cdot f))
      &\text{def.~of $\sigma^{FX}$}\\
      &= \sigma^A_m(g_m \cdot q_m \cdot f) & \text{$g$ a homomorphism} \\
      &= \sigma^A_m(h^\sharp_m \cdot q_m \cdot f) &
      \text{by induction} \\
      &= \sigma^A_m(h^\#_m \cdot f)
      & \text{since $h^\sharp_m \cdot q_m = h^\#_m$} \\
      &= h^\#_n(\sigma(f))
      & \text{def.~of $h^\#$}
    \end{align*}
  \takeout{%% SM: old text
  \item We turn to the the universal property of
    $\eta_X\colon X \to (FX)_0$. First observe that $\eta_X$ is
    clearly monotone since for every $x \leq y$ in $X$ we have
    $X \vdash_0 x \leq y$ by the $(\mathsf{Var})$ rule, whence
    $[x] \leq [y]$ holds in $(FX)_0$. Given $h\colon X\to A_0$, we
    define a family of a maps
    $h^\sharp_k\colon(FX)_k\to A_{k}$ recursively as follows:
    \begin{itemize}
    \item $h^\sharp_0([x]) = h(x)$ for all $x\in X$ and
    \item $h^\sharp_k([\sigma(f)]) =
      \sigma^A_{m}(h^\sharp_m\cdot q_m \cdot f)$
      where $f\colon |\arity(\sigma)|\to \terms_{\T,m}(X)$ and $k = m +
      d(\sigma)$, 
    \end{itemize}
    where $q_m\colon\terms_{\T,m}(X) \epito (FX)_m$ is the canonical
    quotient map, that is $q_m(t) = [t]$.

  \item\label{P:free:3} We prove a \emph{substitution lemma} stating
    that for every monotone
    $\gamma\colon \Delta \to \terms_{\T,k}(X)$ and depth-$m$
    $\T$-defined term $t \in \terms_{\T,m}(\Delta)$ we have
    \begin{equation}\label{eq:hsharp}
      h^\sharp_{m+k} [\ol\gamma(t)] = (h^\sharp_k\cdot q_k \cdot \gamma)^\#_m(t).
    \end{equation}
    By \autoref{P:subterms}, the function on the right-hand side is defined on every
    subterm of $t$. We proceed by induction on
    terms. For $t = x \in \Delta$ we clearly have
    \[
      h^\sharp_k[\gamma(x)]
      =
      h^\sharp_k \cdot q_k \cdot \gamma(x)
      =
      (h^\sharp_k \cdot q_k \cdot \gamma)^\#_k(x).
    \]
    For $t = \sigma(f)$ for $\sigma \in \Sigma(\ari,n)$ and a monotone
    $f\colon P \to \terms_{\T,p}(X)$ with $m = n+p$ we have
    % \begin{align*}
    \[
      \begin{array}{@{\,}r@{\ }l@{}r}
      h^\sharp_{m+k}[\ol\gamma(r)]
      &= h^\sharp_{m+k}[\ol\gamma(\sigma(f))]
      &\text{since $r = \sigma(f)$} \\
      &= h^\sharp_{m+k}[\sigma(\ol\gamma \cdot f)]
      & \text{def.~of $\ol\gamma$} \\
      &= \sigma^A_{p+k}(h^\sharp_{n+k} \cdot q_{n+k} \cdot \ol\gamma \cdot f)
      & \text{def.~of $h^\sharp$} \\
      &=\sigma^A_{p+k}\big(\lambda i.\, h^\sharp_{n+k}[\ol\gamma(f(i))]\big) \\
      &=\sigma^A_{p+k}\big(\lambda i.\, (h^\sharp_k \cdot q_k \cdot  \gamma)^\#_{p}(f(i))\big) 
      &\text{by induction} \\
      &=\sigma^A_{p+k}\big((h^\sharp_k \cdot q_k \cdot   \gamma)^\#_{p}\cdot f\big)
      \\
      &= (h^\sharp_k \cdot q_k \cdot \gamma)^\#_m(\sigma(f))
      & \text{def.~of $(-)^\#$}
    \end{array}
    \]
%    \end{align*}

  \item\label{P:free:4} Now we show that every $h^\sharp_p$ is
    monotone. Suppose that $[t_1] \leq [t_2]$ holds in $(FX)_p$ for
    some $p \in \omega$, eqivalently $X \vdash_p t_1 \leq t_2$. We
    prove by induction on its derivation that this implies
    $h^\sharp_p([t_1]) \leq h^\sharp_p([t_2])$. We proceed by case
    distinction on the last rule applied in the derivation of
    $X \vdash_p t_1 \leq t_2$.  The cases of the $(\mathsf{Var})$,
    $(\mathsf{Ar})$, $(\mathsf{Trans})$ and $(\mathsf{Mon})$ rules are
    easy using monotonicity of the operations $\sigma^A_m$ and the
    respective induction hypotheses as well as transitivity of the
    order on $A_p$ in the last case.

    For the $(\mathsf{Ax1})$ rule let $\Delta \vdash_n s \leq t$ be an
    axiom of $\T$ and
    $\gamma\colon |\Delta| \to \mathsf{T}_{\Sigma,k}(X)$ a uniform
    substitution so that $p = n+k$, $t_1 = \ol\gamma(s)$ and
    $t_2 = \ol\gamma(t)$. The assumption of the rule states that
    $\Gamma \vdash_k \gamma(x) \leq \gamma(y)$ for all $x \leq y$ in
    $\Delta$, which implies that $\gamma$ restricts to a monotone map
    $\Delta \to \terms_{\T,k}(X)$.  The induction hypothesis
    states that for every $x\leq y$ in $\Delta$ we have
    $h^\sharp_k[\gamma(x)] \leq h^\sharp_k [\gamma(y)]$. Using the
    substituion lemma in item~\ref{P:free:3} we can conclude by
    computing
    \begin{align*}
      h^\sharp_p[t_1]
      &= h^\sharp_p[\ol\gamma(s)] &
      \text{since $t_1 = \ol\gamma(s)$}\\
      &= (h^\sharp_k\cdot q_k \cdot \gamma)^\#_m(s)
      & \text{by~\eqref{eq:hsharp}} \\
      &\leq (h^\sharp_k\cdot q_k \cdot \gamma)^\#_m(t)
      & \text{$A$ satisfies $\Delta \vdash_n s \leq t$} \\
      &=h^\sharp_p[\ol\gamma(s)]
      & \text{by~\eqref{eq:hsharp}} \\
      &= h^\sharp_p[t_2] & \text{since $\ol\gamma(t) = t_2$}.
    \end{align*}
    
    For the $(\mathsf{Ax2})$, we argue similarly. Let
    $\gamma\colon |\Delta| \to\mathsf{T}_{\Sigma,k}(X)$ be a uniform
    substitution and $\Delta\vdash_n s\leq t$ be an axiom such that
    there exist an operation symbol $\sigma$ in $\Sigma$ and a map
    $f\colon|\arity(\sigma)|\to \mathsf{T}_{\Sigma, m}(\Delta)$ such
    that $\sigma(f)\in\sub(s, t)$ and, for some $i\leq j$ in
    $\arity(\sigma)$, we have $u=f(i)$ and $v=f(j)$. Then
    $p = m+k$, $t_1 = \ol\gamma(u)$ and $t_2 = \ol\gamma(v)$. Again,
    the assumption of the rule states that $\gamma$ restricts to a
    monotone map $\Delta \to \terms_{\T, k}(X)$. The induction
    hypothesis states that for every $x \leq y$ in $\Delta$ we have
    $h^\sharp_k[\gamma(x)] \leq h^\sharp_k[\gamma(y)]$. Since $A$
    satisfies $\Delta \vdash_n s \leq t$ we know that for every
    monotone $\iota\colon \Delta \to A_k$ we have that $\iota^\#_n(s)$
    and $\iota^\#_n(t)$ are defined (and
    $\iota^\#_n(s) \leq \iota^\#_n(t)$). Unravelling the definition of
    $\iota^\#$ we obtain that $\iota^\#_{m+d(\sigma)}(\sigma(f))$ must
    be defined, which in turn implies that $\iota^\#_m(u)$ and
    $\iota^\#_m(v)$ are defined and
    \begin{equation}\label{eq:iota-uv}
      \iota^\#_m(u) \leq \iota^\#_m(v)
    \end{equation}
    Using the substitution lemma in item~\ref{P:free:3} we can
    conclude by computing
    \begin{align*}
      h^\sharp_p[t_1]
      &= h^\sharp_p[\ol\gamma(u)] &
      \text{since $t_1 = \ol\gamma(u)$}\\
      &= (h^\sharp_k\cdot q_k \cdot \gamma)^\#_m(u)
      & \text{by~\eqref{eq:hsharp}} \\
      &\leq (h^\sharp_k\cdot q_k \cdot \gamma)^\#_m(v)
      & \text{by~\eqref{eq:iota-uv} for $\iota = h^\sharp_k \cdot q_k
        \cdot \gamma$} \\
      &=h^\sharp_p[\ol\gamma(v)]
      & \text{by~\eqref{eq:hsharp}} \\
      &= h^\sharp_p[t_2] & \text{since $\ol\gamma(v) = t_2$}.
    \end{align*}
    
  \item Well-definedness of $h^\sharp_k$ follows from
    monotonicity. Indeed, if $\Gamma \vdash_k t_1 = t_2$, we have
    $h^\sharp_k[t_1] = h^\sharp_k[t_2]$ using item~\ref{P:free:4}
    twice and antisymmetry of the order on $A_k$.
        
  \item We now prove that $h^\sharp\colon FX \to A$ is a homomorphism of
    $(\Sigma,n)$-algebras. That is, for every operation $\sigma$ in
    $\Sigma$ we show that the square below commutes:
    \[
      \begin{tikzcd}
        \Pos(\arity(\sigma), (FX)_m)
        \ar{r}{\sigma^{FX}_m}
        \ar{d}[swap]{\Pos(\arity(\sigma), h^\sharp_m)}
        &
        (FX)_{m+d(\sigma)}
        \ar{d}{h^\sharp_{m+d(\sigma)}}
        \\
        \Pos(\arity(\sigma), A_m)
        \ar{r}{\sigma^A_m}
        &
        A_{m+d(\sigma)}
      \end{tikzcd}
    \]
    Given a monotone $f\colon \arity(\sigma) \to (FX)_m$ we compute 
    (abbreviating $k = m+d(\sigma)$):
    \begin{align*} 
      &h^\sharp_k \cdot \sigma^{FX}_m (f)\\
      &= h^\sharp_k[\sigma(u_m \cdot f)]
      &\text{def.~of $\sigma^{FX}$} \\
      &=\sigma^A_m(h^\sharp_m \cdot q_m \cdot u_m \cdot f)
      &\text{def.~of $h^\sharp$} \\
      &=\sigma^A_m(h^\sharp_m \cdot f)
      & \text{since $q_m \cdot u_m = \id_{(FX)_m}$} \\
      &= \sigma^A_m\big(\Pos(\arity(\sigma), h^\sharp_m)(f)\big)
      &\text{def.~of hom-functor} \\
      &= \sigma^A_m \cdot \Pos(\arity(\sigma), h^\sharp_m)(f).
    \end{align*}
    \takeout{%% SM: somehow I did not need the remaining argument
    Observe that whenever
    $[\sigma(f)]\in(FX)_k$ and
    $f\colon|\arity(\sigma)|\to\mathsf{T}_{\Sigma, m}(X)$ it follows
    from~\autoref{C:arities} and~\autoref{P:subterms} that
    $f(i), f(j)\in\terms_{\T, m}(X)$ and $f(i)\leq f(j)$ in
    $\terms_{\T, m}(X)$ for all $i\leq j$ in $\arity(\sigma)$. 

    Thus,
    $q \cdot f\in\Pos(|\arity(\sigma)|, \terms_{\T, m}(X))$:
    $\sigma^A_m(\iota^\#_m\cdot q \cdot f)$ is defined \smnote{No, it
      isn't because we do not know that $\iota^\sharp$ is monotone.}
    and lies in $A_{k}$ (since $[\sigma(f)]\in\terms_{\T, k}(X)$
    implies $d(\sigma)+m=k)$. Furthermore, for every
    $\sigma\in\Sigma(\Delta, m)$ and $f\in\Pos(\Delta, (F\Gamma)_k)$,
    we have \smnote{This uses monotonicity of $\iota^\sharp_k$, a fact
      which was never proved.}
    \begin{align*}
      \iota^\#_{k+m}\cdot\sigma^{F\Gamma}_k(f)
      &= \iota^\#_{k+m}([\sigma(f)]_{k+m})
      \\
      &= \sigma^A_k(\iota^\#_k\cdot f)
      \\
      &= \sigma^A_k\cdot\Pos(\Gamma, \iota^\#_k)(f).
    \end{align*}
    Hence $\iota^\#\colon FX\to A$ is a homomorphism of
    $\Sigma$-algebras.}% end takeout
    
  \item It follows immediately by definition of
    $h^\sharp$ that $h^\sharp_0\cdot\eta_X= h$.  it remains to
    be seen that $h^\sharp$ is the unique $\Sigma$-homomorphism
    with this property. Given a homomorphism $g\colon FX\to A$ of
    $(\Sigma,n)$-algebras such that $g_0\cdot\eta_X= h$, we verify that
    $g_k[t] = h^\sharp_k[t])$ for all $k\in\omega$ and all
    $[t]\in\terms_{\T, k}(X)$ by induction on terms. For 
    a variable $x$ we have that $k=0$ and
    \[
      h^\sharp_0[x] = h^\sharp_0 \cdot \eta_X(x) =h (x) = g_0 \cdot
      \eta_X(x) = g_0[x].
    \]
    Now, suppose that $t = \sigma(f)$ for some operation $\sigma$ and
    $f\colon |\arity(\sigma)| \to \mathsf{T}_{\Sigma,m}$ with $m+d(\sigma)\leq n$. By
    \autoref{C:arities} and \autoref{P:subterms} we can assume that
    $f$ restricts to depth-$m$ $\T$-defined terms and is monotone;
    that is we can take it as a monotone function $f\colon
    \arity(\sigma) \to \terms_{\T,m}(X)$. In the following we assume
    that in the definition of $\sigma^{FX}_m$ we chose a splitting
    $u_m$ of $q_m\colon \terms_{\T,m} \epito (FX)_m$ such that
    $u_m(q_m(f(i)) = f(i)$ for all $i \in\arity(\sigma)$; that means
    we have 
    \begin{equation}\label{eq:uq2}
      u_m \cdot q_m \cdot f = f. 
    \end{equation}
    Abbreviating $n = m + d(\sigma)$ we can now compute
    \begin{align*}
      g_n[t] &= g_n[\sigma(f)] & \text{since $t = \sigma(f)$} \\
      &= g_n [\sigma(u_m \cdot q_m \cdot f)]
      &
      \text{by~\eqref{eq:uq}}\\
      &= g_n (\sigma^{FX}_m(q_m \cdot f))
      &\text{def.~of $\sigma^{FX}$}\\
      &= \sigma^A_m(g_m \cdot q_m \cdot f) & \text{$g$ a homomorphism} \\
      &= \sigma^A_m(h^\sharp_m \cdot q_m \cdot f) &
      \text{by induction} \\
      &= h^\sharp_n(\sigma^{FX}_m(q_m \cdot f) &
      \text{$h^\sharp$ a homomorphism} \\
      &= h^\sharp_n[\sigma(u_m \cdot q_m \cdot f)]
      &\text{def.~of $\sigma^{FX}$}\\
      &= h^\sharp_n[\sigma(f)] & \text{by~\eqref{eq:uq2}}\\
      &= h^\sharp_n[t] & \text{since $\sigma(f) = t$.}
    \end{align*}
    \takeout{%% Not needed 
    \smnote[inline]{Old text below.}
    Using that $h$ is a homomorphism, 
    it follows that
    \[
      h_{m+d(\sigma)}\cdot\sigma^{F\Gamma}_m(f) = \sigma^A_m(h_m\cdot f)
    \]
    and, by the inductive hypothesis, we have that
    $h_m(f(i))=\iota^\#_m(f(i))$ for all $i\in\arity(\sigma)$. That is,
    $h_m\cdot f=\iota^\#_m\cdot f$. It now follows that
    $\iota^\#_{m+d(\sigma)}([\sigma(f)])
    =h_{m+d(\sigma)}([\sigma(f)])$. Hence 
    $\iota^\#=h$ by induction, as desired.}% end takeout
}% end takeout old text
This completes the proof. \qedhere
\end{enumerate}
\end{proof}
\begin{lemma}\label{L:defnd}
  Let $t \in \rterms k(X)$. For the universal map $\eta_X\colon X \to
  (FX)_0$, if  $(\eta_X)^\#_k(t)$ is
  defined, then $t$ is a $\T$-defined term, that is
  $t \in \termst k (X)$, and moreover $(\eta_X)^\#_k(t) = [t]$.
\end{lemma}
\begin{proof}
  We prove this fact by induction on $t$ abbreviating $\eta_X$ by
  $\iota$.

  If $t$ is a variable $x$, then clearly
  $X \vdash_0 \down x$ by the $(\mathsf{Var})$ rule, and we have
  \[
    \iota^\#_0(x) = \eta_X(x) = [x].
  \]

  For $t = \sigma(f)$ for some operation $\sigma$ and $f\colon
  |\arity(\sigma)| \to \rterms m (X)$ suppose that $\iota^\#_k(t)$ is
  defined. Equivalently, all $\iota^\#_m(f(i))$ are defined and for
  every $i \leq j$ in $\arity(\sigma)$ we have $\iota^\#_m(f(i)) \leq
  \iota^\#_m(f(j))$. By induction, we see that every $f(i)$ is a
  $\T$-defined term in $\termst m (X)$, and for every $i \leq j$ in
  $\arity(\sigma)$ we have
  \[
    [f(i)] = \iota^\#_k(f(i)) \leq \iota^\#_k(f(j)) = [f(j)]
    \qquad \text{(in $(FX)_m$)}.
  \]
  This implies that $f(i) \leq f(j)$ in $\termst m (X)$; that is $X
  \vdash_m f(i) \leq f(j)$. By the
  $(\mathsf{Ar})$ rule we thus conclude
  \[
    X \vdash_{m+d(\sigma)} \down \sigma(f),
  \]
  in other words $t = \sigma(f)$ is a $\T$-defined term in $\termst k(X)$.
\end{proof}
\begin{proof}[Proof of \autoref{T:complete}]
  \takeout{
  \begin{enumerate}
  \item Soundness. Let $A$ be any $(\T,n)$-model. For every
    $\iota\colon X \to A_m$ the components
    $\iota^\#_p\colon \mathsf{T}_{\Sigma,p}(X) \to A_{m+p}$ restrict to
    monotone maps $\kappa_n\colon \terms_{\T,p} \to A_{m+p}$; the
    proof is very similar but somewhat simpler than the proof of
    monotonicity of $h^\sharp$ in \autoref{P:free} using again a
    substitution lemma stating that for every monotone
    $\gamma\colon \terms_{\T,\ell}(X)$ and depth-$o$ $\T$-defined term
    $t \in \terms_{\T,o}$ we have
    \[
      \kappa_{m+k} (\ol\gamma(t)) = (\kappa_\ell \cdot \gamma)^\#_o(t).
    \]

    Now suppose that the inequation $X \vdash_k s \leq t$ is derivable. By
    \autoref{P:subterms} we know that $s, t$ lie in $\terms_{\T,k}(X)$
    and satisfy $s \leq t$ there. Then for every $\iota\colon X \to
    A_m$ both $\iota^\#_k(s)$ and $\iota^\#_k(t)$ are defined and we
    have $\iota^\#_k(s) \leq \iota^\#_k(t)$. Thus $A$ satisfies the
    given inequation in context.
    
  \item Completeness.}% end takeout %%SM: soundness moved to somewhere else
  Suppose that every $(\T,n)$-model satisfies the inequation
  $X \vdash_k s \leq t$. Then, in particular, $FX$ satisfies it. Then for $\iota =
  \eta_X\colon X \to (FX)_0$ we have that $\iota^\#_k(s)$ and
  $\iota^\#_k(t)$ are defined and $\iota^\#_k(s) \leq
  \iota^\#_k(t)$.
  \takeout{%% SM: there is a little issue with this and we have to do
           %% another induction in Remark A.3
    The definedness states precisely that $s, t$ are
    $\T$-defined terms in $\termst k (X)$.\smnote{It's clear that this is true, but we did not demonstrate it
      in detail.}
    We now show that $\iota^\#_k$ restricts to
    $q_k\colon \termst k(X) \epito (FX)_k$. Indeed, from \autoref{P:free}
    we know that $(\eta_X)^\sharp_k = \id_{(FX)_k}$ and from the proof of
    this proposition $\eta^\sharp_k$ is the unique monotone map such
    that $(\eta_X)^\sharp_k \cdot q_k = \iota^\#_k$. Thus,
    \[
      \iota^\#_k = (\eta_X)^\sharp_k \cdot q_k = q_k,
    \]
    or, in other words for every $\T$-defined term $r$ in $\termst k
    (X)$ we have
    $\iota^\#_k(r) = [r]$.}% end takout
  By \autoref{L:defnd} we have that $s,t$ are $\T$-defined terms in
  $\termst k (X)$ and moreover
  \[
    [s] = \iota^\#_k(s) \leq \iota^\#_k(t) = [t]
    \qquad \text{(in $(FX)_k$)}.
  \]
  By definition of the order on $(FX)_k$, we see that $s \leq t$ holds in
  $\terms_{\T,k}(X)$, which states precisely that the given inequation
  in context is derivable: $X \vdash_k s \leq t$.
\end{proof}

\subsection{Details on the Subdistribution Monad}\label{S:app-sdist}

\noindent We proceed to give an explicit description of the monad
$\sdist$ on $\Pos$ induced by the theory in
\autoref{S:theorymonad}. Specifically, we show that~$\sdist$ is a
lifting of the usual finitely supported subdistribution monad on
$\Set$, i.e.~$|\sdist X|$ is the set of all finitely
supported subdistributions on $X$. Moreover, we characterize the
ordering on~$\sdist X$.  A similar characterization has been provided
by Jones and Plotkin \cite[Lemma 9.2]{JP89} in a more restricted
domain-theoretic setting (directed-complete posets). Their argument is
based on the max-flow-min-cut theorem; we give an independent, more
syntactic argument.

Given a poset $X$, we first note that the terms representing elements
of $\sdist X$ can be normalized (in the standard manner known from the
set-based case~\cite{Pumpluen03}) to a single layer of formal
subconvex combinations, i.e.~formal sums
$\sum_{i = 1}^n p_i\cdot x_i$ where $n\in\omega$, $x_i\in X$,
$p_i\in (0, 1]$, and $\sum_{i\in I}p_i\leq 1$.  We generally write
$[n]=\{1,\dots,n\}$ for $n\in\omega$.  A subconvex combination
$\sum_{i=1}^np_i\cdot x_i$ over~$X$ \emph{represents} the
subdistribution on~$X$ that assigns to $x\in X$ the probability
$\sum_{i\in[n], x_i=x}p_i$.

\begin{defn}\label{d:app-sdistorder}
\begin{enumerate}
\item A \emph{subdivision} of a formal subconvex combination
  $\sum_{i=1}^n p_i\cdot x_i$ is a formal subconvex combination
  $\sum_{i=1}^n(\sum_{j=1}^{m_i}\overline{p}_j)\cdot x_i$ such that
  $\sum_{j=1}^{m_i}\overline{p}_j= p_i$ for all $i\leq n$. We then
  also say that
  $\sum_{i=1}^n(\sum_{j=1}^{m_i}\overline{p}_j)\cdot x_i$
  \emph{refines} $\sum_{i=1}^n p_i\cdot x_i$.

\item We say that a formal subconvex combination
  $\sum_{i = 1}^n p_i\cdot x_i$ is \emph{obviously below}
  $\sum_{i=1}^m q_i\cdot y_i$ if there is an injective map
  $f\colon [n]\to [m]$ such that $p_i\leq q_{f(i)}$ and
  $x_i\leq y_{f(i)}$ for all $i\leq n$. In this case, we write
  $\sum_{i=1}^n p_i\cdot x_i\sqsubseteq \sum_{i=1}^m q_j\cdot y_i$,
  and say that~$f$ \emph{witnesses} this relation. If in fact
  $p_i=q_{f(i)}$ and $x_i=y_{f(i)}$ for all~$i$, then we say that~$f$
  is \emph{non-increasing}.

\item Finally, we define the relation $\preccurlyeq$ on $\sdist X$ by 
putting $\sum_{i=1}^n p_i\cdot x_i\preccurlyeq\sum_{j=1}^m q_j\cdot y_j$ 
if and only if there exist subdivisions $d_1$ and $d_2$ of $\sum_{i=1}^n p_i\cdot x_i$ 
and $\sum_{j=1}^k q_j\cdot y_j$, respectively, such that $d_1\sqsubseteq d_2$. 
\end{enumerate}
\end{defn}
\noindent Recall that in addition to the standard equational axioms,
the theory includes the inequational axiom scheme
\begin{equation}\label{eq:subconvex-ineq}\textstyle
  \{x_1,\dots, x_n\}\vdash\sum_{i=1}^n p_i\cdot x_i\leq \sum_{i=1}^n q_i\cdot x_i
\end{equation}
for $(p_i\leq q_i, i=1,\dots, n)$ and moreover that we generally
enforce that operations, in this case subconvex combinations, are
monotone. It is thus essentially immediate that $\preccurlyeq$ is
included in the ordering~$\le$ on $\sdist X$; explicitly:
\begin{lemma}
  If $d_1\preccurlyeq d_2$, then $d_1\le d_2$.
\end{lemma}
\begin{proof}
  If~$d_1$ is a subdivision of~$d_2$, then the equational part of the
  theory clearly implies $d_1=d_2$; thus, it suffices to show the claim
  for the case $d_1\sqsubseteq d_2$. So let
  $d_1=\sum_{i=1}^np_i\cdot x_i$, $d_2=\sum_{i=1}^m q_i\cdot y_i$, and
  let $d_1\sqsubseteq d_2$ be witnessed by an injective map
  $f\colon[n]\to[m]$, i.e.~$p_i\le q_{f(i)}$, $x_i\le y_{f(i)}$ for
  all $i\le n$. We write $\Im(f)\subseteq [m]$ for the image
  of~$f$. Then we have
  \begin{align*}\textstyle
    d_1
    &\textstyle
    = \sum_{i=1}^np_i\cdot x_i+\sum_{j\in[m]\setminus \Im(f)}0\cdot  y_j
    \\
    &\textstyle
    \le \sum_{i=1}^np_i\cdot y_i+\sum_{j\in[m]\setminus \Im(f)}0\cdot  y_j 
    \le d_2,
  \end{align*}
  where the first inequality is by monotonicity and the
  second by~\eqref{eq:subconvex-ineq}.
\end{proof}
\noindent We proceed to show that the reverse inclusion also holds.
Since the ordering on $\sdist X$ is the smallest partial order (more
precisely, the partial order quotient of the smallest pre-order, which
however we will show to be already antisymmetric) that contains all
instances of the inequational axiom scheme~\eqref{eq:subconvex-ineq}
and moreover is preserved by subconvex combinations, it suffices to
show that $\preccurlyeq$ already has these properties. It is clear
that~$\preccurlyeq$ subsumes all instances
of~\eqref{eq:subconvex-ineq}: if $p_i\leq q_i$ for all $i\leq n$, then
$\sum_{i=1}^n p_i\cdot x_i\preccurlyeq \sum_{i=1}^n q_i\cdot x_i$ is
witnessed by the trivial subdivisions of these subconvex combinations
via the identity on~$[n]$. We now prove the monotonicity of subconvex
combinations on $\sdist X$:
\begin{lemma}
  Subconvex combinations preserve both the subdivision relation and
  the obviously-below relation $\sqsubseteq$, and are hence monotone
  w.r.t.\ $\preccurlyeq$.
\end{lemma}
\begin{proof}
  Monotonicity w.r.t.~$\preccurlyeq$ clearly follows from the first
  two claims. 

  \emph{Preservation of the subdivision relation:} Since the
  subdivision relation is clearly transitive, it suffices to prove
  preservation in one argument of a subconvex combination, which by
  commutativity we can assume to be the first. So let
  $d_1=\sum_{i=1}^n(\sum_{j=1}^{m_i}\bar p_j)\cdot x_i$ be a
  subdivision of $d_2=\sum_{i=1}^np_i\cdot x_i$, i.e.\
  $\sum_{j=1}^{m_i}\bar p_j=p_i$ for $i=1,\dots,n$, and take a
  subconvex combination $\sum_{i=1}^{k} q_i\cdot y_i$. Then
  \begin{align*}\textstyle
    q_1\cdot d_1+\sum_{i=2}^kq_i\cdot y_i &\textstyle =
     \sum_{i=1}^n\sum_{j=1}^{m_i}q_1\bar p_j\cdot x_i +\sum_{i=2}^kq_i\cdot y_i\\
    \textstyle q_1\cdot d_2+\sum_{i=2}^kq_l\cdot y_l & \textstyle=
       \sum_{i=1}^n q_1p_i\cdot x_i  +\sum_{i=2}^kq_i\cdot y_i,
  \end{align*}
  and the first of these subconvex combinations is clearly a
  subdivision of the second.

  \emph{Preservation of~$\sqsubseteq$:} Let
  $d_1=\sum_{i=1}^np_i\cdot x_i$, $d_2=\sum_{i=1}^mq_i\cdot y_i$ such
  that $d_1\sqsubseteq d_2$; that is, we have an injection
  $f\colon[n]\to[m]$ such that $p_i\le q_{f(i)}$ and $x_i\le y_{f(i)}$
  for all $i\in[n]$. Again, it suffices to show preservation in the
  first argument of a subconvex
  combination~$\sum_{j=1}^kr_j\cdot z_j$. Indeed, we have
  \begin{align*}\textstyle
    r_1\cdot d_1+\sum_{j=2}^kr_j\cdot z_j
    & \textstyle = \sum_{i=1}^nr_1p_i\cdot x_i+\sum_{j=2}^kr_j\cdot z_j\\
    \textstyle r_1\cdot d_2+\sum_{j=2}^kr_j\cdot z_j
    & \textstyle = \sum_{i=1}^mr_1q_i\cdot y_i+\sum_{j=2}^kr_j\cdot z_j,
  \end{align*}
  and the first of these subconvex combinations is obviously below the
  second, as witnessed by the injection that acts like~$f$ on the
  indices of $\sum_{i=1}^nr_1p_i\cdot x_i$ (note that
  $r_1p_i\le r_1 q_{f(i)}$ for $i\in[n]$ since $r_1\ge 0$) and as
  identity on the indices of $\sum_{j=2}^kr_j\cdot z_j$.
\end{proof}
\noindent It remains to show that $\preccurlyeq$ is a partial
ordering. To this end, it will be convenient to work with two
alternative presentation of partitions of real numbers, introduced
presently.

For $p\in [0,1]$, we denote by $\summands(p)$ the set of \emph{summand
  representations} of partitions of~$p$, i.e., the set of sequences
$p_1,\dots, p_n\in (0,1]$ (with~$n$ referred to as the \emph{length}
of the sequence) such that $\sum_{i=1}^n p_i = p$. We define a
partial order~$\rightarrow$ (\emph{refines}) on $\summands(p)$ by
$(p_1,\dots, p_n)\rightarrow(q_1,\dots, q_k)$ if and only if
$(q_1,\dots, q_k)$ arises from $(p_1,\dots, p_n)$ by summing
consecutive summands, formally if there exist indices
$1\le j_1<\dots<j_n<j_{k+1}=n+1$ such that
$q_i=\sum_{l=j_i}^{j_{i+1}-1}p_l$ for $i=1,\dots,k$. Observe that a
subdivision of $\sum_{i=1}^n p_i\cdot x_i$ may be specified by the
choice of a sequence $s_i\in\summands(p_i)$ for each $i\leq n$.%

Alternatively, we may represent partitions in terms of sequences of
partial sums: The set $\partition(p)$ of \emph{partial-sum
  representations} of partitions of~$p$ consists of all finite subsets
$P\subseteq [0,p)$ such that $0\in s$. When we write
$P=\{q_0,\dots,q_{n-1}\}$, we mean to imply that
$0=q_0<q_1<\dots<q_{n-1}$, and implicitly understand $q_n=p$; we refer
to~$n$ as the \emph{length} of~$s$. We note the equivalence of the two
representations explicitly:

\begin{lemma}\label{L:sdist:bij}
  For every $p\in[0,1]$, we have a length-preserving order isomorphism
  \[(\summands(p), \rightarrow)\cong(\partition(p), \supseteq),\]
  which in the left-to-right direction assigns to a summand
  representation $(p_1,\dots,p_n)\in\summands(p)$ the partial-sum
  representation $\{\sum_{i=1}^k p_i\mid k=0,\dots,n-1\}$, and in the
  right-to-left direction assigns to a partial-sum representation
  $\{q_0,\dots,q_{n-1}\}$ the summand representation $(p_1,\dots,p_n)$
  given by $p_i=q_i-q_{i-1}$.
\end{lemma}
% \begin{proof}
% We begin by constructing a bijection $\mathsf{pres}_k: \summands_k(p)\cong
% \partition_k(p);$ the desired order isomorphism will then be given by the union over
% the $\mathsf{pres}_k$. 

% To each sequence $p_1,\dots, p_k$ such that $\sum_{i=1}^{k} p_i= p$ we 
% assign the set $\{0, q_1,\dots, q_{k-1}\}$ with $q_{i+1}:= q_i+ p_{i+1}$ for 
% all $i\leq k$. This defines an element of $\partition_k(p):$ for each $i\leq k$ 
% we have $q_i \leq q_i+p_{i+1}$ since $0\leq p_{i+1}$ and $q_i\leq p$ since 
% $q_i=\sum_{j=1}^i p_i\leq\sum_{j=1}^k p_i= p$. 

% Conversely, to each set $\{0, q_1,\dots, q_{k-1}\}\in\partition_k(p)$ we assign 
% the sequence $p_1,\dots, p_k\in[0,1]$ with $p_{i}:= q_{i}-q_{i-1}$ for each 
% $1\leq i\leq k-1$ and $p_k:= p-q_{k-1}$. Note that $p_i\in [0,1]$ since 
% $p\geq q_i\geq q_{i-1}$ for all $i\leq k-1$, and 
% \[
% \sum_{i=1}^k p_i = \sum_{i=1}^{k-1} (q_{i}-q_{i-1})+(p-q_{k-1})= \sum_{i=1}^{k-1} (q_i-q_i) + p = p.
% \]
% Hence $(p_1,\dots, p_n)\in\summands_k(p)$. It is easy to see that
% (i) these assignments are mutually inverse and (ii) both
% $\mathsf{pres}$ and $\mathsf{pres}^{-1}$ are monotone with respect to
% the prescribed orderings. Thus
% $(\summands(p), \rightarrow)\cong(\partition(p), \supseteq)$, as
% desired.
% \end{proof}

\noindent It follows that the set of subdivisions of a subconvex
combination $\sum_{i=1}^n p_i\cdot x_i$ is in bijection with the set
of families $(P_i)_{i\leq n}$ such that $P_i\in\partition(p_i)$. We
next establish the key properties that will be needed in order to
prove that $\preccurlyeq$ is a partial order. We first note that given
subdivisions always have a joint refinement:

\begin{lemma}\label{L:jointsubd}
  Let $d_0$ and $d_1$ be subdivisions of $\sum_{i=1}^n p_i\cdot
  x_i$. Then there exists a common subdivision of $d_0$ and $d_1$.
\end{lemma}
\begin{proof}
  One immediately reduces to the case $n=1$; we then omit the
  indices. By \autoref{L:sdist:bij}, it suffices to consider
  partial-sum representations of partitions. But for subdivisions
  represented by $P,Q\in\partition(p)$, we obtain a common subdivision
  from $P\cup Q\in\partition(p)$. 
\end{proof}
\noindent Next, we show that subdivisions can be pulled back and
pushed forward along the obviously-below relation:
\begin{lemma}\label{L:obviouslybelow}
  Suppose that
  $\sum_{i=1}^n p_i\cdot x_i\sqsubseteq\sum_{i=1}^m q_i\cdot y_i$.
  Then:
  \begin{enumerate}
  \item\label{L:obviouslybelow:1} For every subdivision $d$ of
    $\sum_{i=1}^n p_i\cdot x_i$ there exists a subdivision $d^*$ of
    $\sum_{i=1}^n q_i\cdot y_i$ such that $d\sqsubseteq d^*$;

  \item\label{L:obviouslybelow:2} For every subdivision $d$ of
    $\sum_{i=1}^n q_i\cdot y_i$ there exists a subdivision $d_*$ of
    $\sum_{i=1}^n p_i\cdot x_i$ such that $d_*\sqsubseteq d$.
  \end{enumerate}
\end{lemma}
\begin{proof}
  One immediately reduces to the case that $n=m=1$; we then omit the
  indices. Both claims are now seen straightforwardly via
  partial-sums representations, as we show below. Given
  $p\cdot x\sqsubseteq q\cdot y$, we have $p\le q$ and $x\le y$.
  \begin{enumerate}
  \item We are given a subdivision of $p\cdot x$, say in terms of
    $P\in\partition(p)$. But then~$P$ is also in $\partition(q)$, and
    as such induces the desired subdivision $d^*$ of $q\cdot y$, with
    the relation $d\sqsubseteq d^*$ witnessed by the identity map.

      \item We are given a subdivision of~$q\cdot y$, say in terms
    of $Q\in\partition(q)$. We then obtain $d_*$ as desired from
    $Q\cap[0,p)\in\partition(p)$, with the relation $d_*\sqsubseteq d$
    witnessed by the inclusion of the respective index sets.\qedhere
    \end{enumerate}
\end{proof}

\begin{proposition}
  The relation $\preccurlyeq$ is a preorder.
\end{proposition}
\begin{proof}
\begin{enumerate}
\item  We first note that $\preccurlyeq$ is reflexive since every subconvex combination is a 
subdivision of itself. 

\item We now show that $\preccurlyeq$ is transitive. To this end,
  suppose that (a)
  $\sum_{i=1}^{\ell} p_i\cdot x_i\preccurlyeq\sum_{j=1}^m q_j\cdot
  y_j$ and (b)
  $\sum_{j=1}^m q_j\cdot y_j\preccurlyeq\sum_{k=1}^n r_k\cdot z_k$
  hold; we will show that
  $\sum_{i=1}^{\ell} p_i\cdot x_i\preccurlyeq\sum_{k=1}^n r_j\cdot
  z_j$. By (a), we have subdivisions $s_p\sqsubseteq s_{pq}$ of
  $\sum_{i=1}^{\ell} p_i\cdot x_i$ and $\sum_{j=1}^m q_j\cdot y_j$,
  respectively, and, by (b), we are given subdivisions
  $s_{qr}\sqsubseteq s_r$ of $\sum_{j=1}^m q_j\cdot y_j$ and
  $\sum_{k=1}^n r_k\cdot z_j$, respectively. This is all summarized in
  the diagram below

\begin{center}
\begin{tikzcd}[column sep = 0]
s_p  \arrow[r, symbol = \sqsubseteq] \arrow[d]      & s_{pq} \arrow[rd] &                          & s_{qr} \arrow[ld] \arrow[r, symbol = \sqsubseteq] & s_r \arrow[d]             \\
\sum_{i=1}^{\ell}p_i\cdot x_i &                   & \sum_{j=1}^mq_j\cdot y_j &                             & \sum_{k=1}^n r_k\cdot z_k
\end{tikzcd}
\end{center}
where an arrow from $x$ to $y$ denotes that $x$ is a subdivision of
$y$. Let $s_q$ denote the joint subdivision of $s_{pq}$ and $s_{qr}$
given by \autoref{L:jointsubd}. Then, by
\autoref{L:obviouslybelow}\ref{L:obviouslybelow:2}, we obtain a
subdivision $(s_p)_*$ of $s_p$ such that $(s_p)_*\sqsubseteq s_q$ and,
by \autoref{L:obviouslybelow}\ref{L:obviouslybelow:1} we obtain a
subdivision $(s_r)^*$ of $s_r$ such that $s_q\sqsubseteq (s_r)^*$. Let
$f_{pq}$ and $f_{qr}$ be
injections witnessing these relations.  Then the composite
$f_{qr}\cdot f_{pq}$ is a injection witnessing that
$(s_p)_*\sqsubseteq (s_r)^*$.  Since $(s_p)_*$ and $(s_r)^*$ are
moreover subdivisions of $\sum_{i=1}^{\ell}p_i\cdot x_i$ and
$\sum_{k=1}^n r_k\cdot z_k$, respectively, it immediately follows that
$\sum_{i=1}^{\ell}p_i\cdot x_i \preccurlyeq \sum_{k=1}^n r_k\cdot
z_k$, as desired.\qedhere
\end{enumerate}
\end{proof}
\noindent It remains to show that~$\sdist X$ really consists of
subdistributions, i.e.~that no identifications are forced by
antisymmetry. We begin with a technical lemma.
\begin{lemma}\label{lem:obviously-below-self}
  Let $d_1=\sum_{i=1}^n p_i \cdot x_i$ and
  $d_2=\sum_{i=1}^m q_i\cdot y_j$ be subconvex combinations with
  non-zero coefficients, and let~$f\colon[n]\to[m]$ witness
  $d_1\sqsubseteq d_2$. If $d_1$ and $d_2$ represent the same
  subdistribution, then~$f$ is bijective and \emph{non-increasing}:
  \smnote{The meaning of non-increasing is not quite clear because we
    don't mean the obvious order on the (co)domain of $f$ but $X$ and $[0,1]$.}
  \[
    x_i =  y_{f(i)} 
    \qquad\text{and}\qquad
    p_i = q_{f(i)}
    \qquad
    \text{for all $i$}.
  \]
\end{lemma}
\begin{proof}
  We proceed by induction on the number of summands. Since $d_1$ and
  $d_2$ represent the same subdistribution we have
  \[
    \{x_i\mid i\in[n]\}=\{y_i\mid i\in[m]\}.
  \]
  Let $x \in X$ be a maximal element of that set.  Then necessarily
  $y_{f(i)}=x_i$ whenever $x_i=x$. Using once again that $d_1$ and
  $d_2$ represent the same subdistribution, we see that the sum of
  coefficients of $x$ and $d_1$ and $d_2$ are equal:
  \[
    \sum_{\substack{i\in[n] \\ x = x_i}} p_i
    =
    \sum_{\substack{i\in[m]\\ x =        y_i}} q_i.
  \]
  Since for all~$i$ we have $p_i\le q_{f(i)}$, it follows that
  $p_i=q_{f(i)}$ for all~$i \in [n]$ with $x_i = x$. This in turn
  implies that $y_j=x$ only if $j=f(i)$ for some~$i \in [n]$ such that
  $x_i=x$. Thus,~$f$ restricts to a witness of
  \[
    \sum_{\substack{i\in[n]\\ x_i\neq x}} p_i\cdot x_i
    \sqsubseteq
    \sum_{\substack{i\in[m]\\ y_i\neq x}} q_i\cdot y_i,
  \]
  and the two sides in this relation represent the same
  subdistribution, so we are done by induction.
\end{proof}
\begin{proposition}
  The relation $\preccurlyeq$ is antisymmetric on the set of
  subdistributions on a set~$X$.
\end{proposition}
\begin{proof}
  Suppose that
  $\sum_{i=1}^m p_i\cdot x_i\preccurlyeq \sum_{j=1}^n q_j\cdot y_j$
  and
  $\sum_{j=1}^n q_j\cdot y_j\preccurlyeq \sum_{i=1}^m p_i\cdot x_i;$
  we proceed to show that $\sum_{i=1}^m p_i\cdot x_i$ and
  $\sum_{j=1}^n q_j\cdot y_j$ represent the same subdistribution.
  Unwinding our assumptions, we have subdivisions $s_{p0}, s_{p1}$ and
  $s_{q0}, s_{q1}$ of $\sum_{i=1}^m p_i\cdot x_i$ and
  $\sum_{j=1}^n q_j\cdot y_j$, respectively, such that
  $s_{p0}\sqsubseteq s_{q0}$ and $s_{q1}\sqsubseteq s_{p1}$, all as in
  indicated in the diagram below:
\begin{center}
\begin{tikzcd}[column sep = 5]
s_{q1} \arrow[r, symbol = \sqsubseteq] \arrow[rrdd] & s_{p1} \arrow[rd] &                           & s_{p0} \arrow[r, symbol = \sqsubseteq] \arrow[ld] & s_{q0} \arrow[lldd] \\
                              &                   & \sum_{j=1}^m p_i\cdot x_i &                             &                     \\
                              &                   & \sum_{j=1}^n q_j\cdot y_j &                             &                    
\end{tikzcd}
\end{center}
Applying \autoref{L:jointsubd} to $s_{p0}$ and $s_{p1}$, we obtain a common subdivision $s_p$. Then, by 
\autoref{L:obviouslybelow}, there exist subdivisions $(s_p)_*$  and $(s_q)^*$ of $s_{q1}$ and $s_{q0}$, respectively, such that 
\[
(s_p)_*\sqsubseteq s_p \sqsubseteq (s_p)^*.
\]
Let the left hand relation be witnessed by an injection~$f$ and the
right hand relation by~$g$. Then the composite $g\cdot f$ witnesses
the relation $(s_p)_*\sqsubseteq (s_p)^*$, which runs between two
subdivisions of the subconvex combination $\sum_{j=1}^n q_j\cdot
y_j$. Assuming w.l.o.g.\ that all coefficients in the relevant
subconvex combinations are non-zero, we apply
\autoref{lem:obviously-below-self} to $g\cdot f$ and obtain that
$g\cdot f$ is bijective and non-increasing. Since~$f$ and~$g$ are
already injective, the same then follows for~$f$ and~$g$, so
that~$(s_p)_*$ and~$s_p$ differ only by reordering of summands; hence,
$\sum_{i=1}^m p_i\cdot x_i$ and $\sum_{j=1}^n q_j\cdot y_j$ represent
the same subdistribution.
\end{proof}
\noindent In summary, we have shown \autoref{T:sdist} repeated below:
\smnote{Isn't there something missing? Namely that we obtain the monad
  multiplication out of the theory that we would expect and further
  use below.}
\begin{theorem}
  The monad $\sdist$ given by the theory in \autoref{S:probtraceinc}
  assigns to a poset $X$ the set of all finitely generated
  subdistributions on $X$ equipped with the order $\preccurlyeq$.
\end{theorem}

\takeout{
\subsection{Details for \autoref{E:induced}}\label{S:app-monadtheory}

We verify that the descriptions $(M_n)_{n<\omega}$ proposed in
\autoref{E:induced} of the graded monads induced by the graded
theories from~\autoref{E:theory} are correct. Each of these theories
has a corresponding signature consisting entirely of operations with
discretely ordered arities.

\subsubsection{Details for $\JSL(\A)$}

Recall that the signature $\Sigma$ of $\JSL(\A)$ consists of 
$n$-ary depth-1 operations $a_1(-)+\cdots+a_n(-)$ for every word 
$a_1\cdots a_n\in\A^*$; it will be convenient to denote this operation
by $a_1+\cdots+a_n(-)$ for the remainder of this section. Note that every
depth-0 $\Sigma$-term in a given context $X$ is a variable $x\in X$, 
and a depth-$(k+1)$ term has the form $a_1+\cdots+a_n(f)$ where
$f\colon [n]\to\mathsf{T}_{\Sigma}(X)$. 

\begin{defn}
  The \emph{normalization} $\mathcal{N}(t)$ of a $\JSL(\A)$-defined
  term $t$ with variables in $X$ is defined by induction as follows:\smnote{Please stop using enumerate with manual
      \texttt{bullet}; you should use itemize!}
  \begin{itemize}
  \item $\mathcal{N}(x)= x$ for all $x\in X$ and
  \item $\mathcal{N}(a_1+\cdots+a_n(f))=\conv\{(a_i, \mathcal{N}(f(i)))~|~i\in [n]\}$
    for all $a_1\cdots a_n\in\A^*$ and all functions $f\colon[n]\to\mathsf{T}_{\Sigma, k}(X)$.
  \end{itemize}
\end{defn}

Recall that our goal is to show that the $\JSL(\A)$ induces the graded 
monad with $M_n$ being the $n$-fold composition of the functor
$C_{\omega}(\A\times -)$. That is, $M_0 X= X$ and 
$M_{n+1}X= C_{\omega}(\A\times M_nX)$. We have the following:

\begin{lemma}
For all $k\in\omega$ and all $t\in\mathsf{T}_{\Sigma, k}(X),$
we have $\mathcal{N}(t)\in M_k X$.
\end{lemma}
\begin{proof}
By induction on terms. As a base case, we note that $\mathcal{N}(x)=x\in X=M_0 X$
for every variable $x$. Assume that $t= a_1+\cdots+a_n(f)$ for some word 
$a_1\cdots a_n\in\A^*$ and some $f\colon[n]\to\mathsf{T}_{\Sigma, m}(X)$. By induction, 
$\mathcal{N}(f(i))\in M_mX$. Hence 
$\mathcal{N}(t)= \conv\{(a_i, \mathcal{N}(f(i)))~|~i\in[n]\}\in C_{\omega}(\A\times M_mX).$ 
Since $d(a_1+\dots+a_n)+m= 1+m$, this concludes the proof.
\end{proof}

In particular, we may view the normalization map $\mathcal{N}$ as a family of 
functions
\[
\mathcal{N}_k:\mathsf{T}_{\Sigma, k}(X)\to M_kX.
\]

\begin{proposition}\cfnote{TODO}
  $\mathcal{N}(s)=\mathcal{N}(t)$ implies $s\sim t$.
\end{proposition}
\begin{proof}

\end{proof}
}

\subsection{Details for \autoref{sec:logics}}\label{S:app-logics}

\begin{proof}[Details for \autoref{E:logics}]
  We prove that the interval $[0,1]$ equipped with the structure maps
  $o\colon \sdist[0,1] \to [0,1]$ and $\llbracket\langle
  a\rangle\rrbracket\colon\sdist(\A\times[0,1])\rightarrow[0,1]$ that
  we abbreviate as $\alpha$ and whose definition, in terms of formal
  convex combinations, is given by
  \[\textstyle
    \alpha\big(\sum_{i=1}^n p_i\cdot (a_i, r_i)\big) = \sum_{i=1}^n
    p_i r_i.
  \]
  is indeed an $M_1$-algebra. We abbreviate
  $\llbracket\langle a\rangle\rrbracket$ as $\alpha$. It is well-known
  that $o$, taking expected values, satisfies the laws of an
  Eilenberg-Moore algebra for $\sdist = M_0$. It remains to prove that
  $o$ and $\alpha$ are monotone and the two required instances of the
  squares in \eqref{D:gradalg} related $o$ and $\alpha$; one, called
  \emph{homomorphy}, stating that $\alpha$ is a homomorphism of
  $M_0$-algebras (i.e.~subconvex modules) and the other, called
  \emph{coequalization}, stating that the following diagram is a fork
  (i.e.~$\alpha$ merges the two parallel morphisms):
  \[
    \begin{tikzcd}[column sep = 15]
      M_1M_0 [0,1]
      \ar[yshift=3]{rr}{\mu^{1,0}_{[0,1]}}
      \ar[yshift=-3]{rr}[swap]{M_1o}
      &&
      M_1[0,1]
      \arrow{r}{\alpha}
      &
      {}[0,1]
    \end{tikzcd}
  \]
  \begin{enumerate}
  \item Monotonicity of $o$. Let $d_1, d_2 \in \sdist [0,1]$. Since
    subdivisions of a subconvex combination have the same expected
    value, it suffices to shows that $o(d_1) \leq o(d_2)$ whenever
    $d_1 \sqsubseteq d_2$. This is obvious from the definition of $\sqsubseteq$.

  \item Monotonicity of $\alpha$. Let $d_1 = \sum_{i=1}^n p_i\cdot
    (a_i,r_i)$ and $d_2 = \sum_{i=1}^m q_i \cdot (b_i, s_i)$ in $\sdist(\A
    \times [0,1])$. Since we clearly have $\alpha(d') = \alpha(d)$
    whenever $d'$ is a subdivision of a subconvex combination $d$, it suffices to show that
    $\alpha(d_1) \leq \alpha(d_2)$ whenever $d_1 \sqsubset d_2$. So
    suppose we have an injective map $f\colon [n] \to [m]$ with $p_i
    \leq q_{f(i)}$ and $r_i \leq s_{f(i)}$ for every $i \in [n]$. Then
    we clearly have
    \begin{align*}
      \alpha(d_1)
      &\textstyle= \sum_{i=1}^n p_ir_i \\
      &\textstyle\leq \sum_{i=1}^n q_{f(i)}  s_{f(i)} \\
      &\textstyle\leq \sum_{i=1}^m q_i s_i = \alpha(d_2).
    \end{align*}
    
  \item Homomorphy. We are to prove that the following square
    commutes:
    \[
      \begin{tikzcd}
        M_0M_1 [0,1]
        \ar{r}{\mu^{0,1}_{[0,1]}}
        \ar{d}[swap]{M_0\alpha}
        &
        M_1[0,1]
        \ar{d}{\alpha}
        \\
        M_1[0,1]
        \ar{r}{\alpha}
        &
        {}[0,1]
      \end{tikzcd}
    \]
    Let $\sum_{i=1}^n p_i \cdot d_i$ be an element of $M_0M_1[0,1]$ where $d_i = \sum_{j=1}^{m_i}
    q_{ij} \cdot (a_{ij},r_{ij})$ lies in $M_1[0,1] = \sdist(\A \times
    [0,1])$. We have
    \[\textstyle
      \mu^{0,1}_{[0,1]}\big(\sum_{i=1}^n p_i \cdot d_i\big) =
      \sum_{i=1}^n\sum_{j=1}^{m_i} p_iq_{ij} \cdot (a_{ij}, r_{ij}),
    \]
    which is mapped by $\alpha$ to
    \begin{equation}\label{eq:sum}
      %\textstyle
      \sum_{i=1}^n\sum_{j=1}^{m_i} p_iq_{ij}r_{ij}.
    \end{equation}
    We also have
    \begin{align*}
      \textstyle
      M_0\alpha\big(\sum_{i=1}^n p_i \cdot d_i\big)
      &\textstyle= \sum_{i=1}^n p_i \cdot \alpha(d_i)\\
      &\textstyle= \sum_{i=1}^n p_i \cdot \big(\sum_{j=1}^{m_i}
      q_{ij}r_{ij}\big),
    \end{align*}
    which is mapped by $\alpha$ to the sum in~\eqref{eq:sum} again.

  \item Coequalization. Let $\sum_{i=1}^n p_i \cdot (a_i,d_i)$ be an
    element of $M_1M_0[0,1]$ with $d_i = \sum_{j=1}^{m_i} q_{ij}\cdot
    r_{ij}$ in $M_0 [0,1] = \sdist[0,1]$. We have
    \[\textstyle
      \mu^{0,1}_{[0,1]}\big(\sum_{i=1}^n p_i \cdot (a_i,d_i)\big)
      =
      \sum_{i=1}^n\sum_{j=1}^{m_i} p_iq_{ij} \cdot (a_i,r_{ij}),
    \]
    which is mapped by $\alpha$ to the sum in~\eqref{eq:sum}. We also
    have
    \begin{align*}
      \textstyle
      M_1o\big(\sum_{i=1}^n p_i \cdot (a_i,d_i)\big)
      &\textstyle= \sum_{i=1}^n p_i \cdot (a_i, o(d_i)) \\
      &\textstyle= \sum_{i=1}^n p_i \cdot \big(\sum_{j=1}^{m_i} q_{ij}r_{ij}\big),
    \end{align*}
    which is mapped by $\alpha$ to the sum~\eqref{eq:sum} once again.\qedhere
  \end{enumerate}
\end{proof}

\subsection{Details for \autoref{sec:expressiveness}}\label{S:app-expressiveness}

%\subsection*{Proof of \autoref{thm:expr}}
\begin{proof}[Proof of \autoref{thm:expr}]
We proceed by induction on $n$ to show that, for every $n\in\omega$,
the family $\llbracket\varphi\rrbracket\colon M_n1\rightarrow\Omega$
of evaluations of graded formulae $\varphi\in\mathcal{L}_n$ is jointly
order-reflecting. The base case $n=0$ follows immediately from
depth-0 separation. Let $\mathfrak{A}$ denote the family of evaluation
maps $\llbracket\varphi\rrbracket\colon M_n1\rightarrow\Omega$ of
formulae $\varphi\in\mathcal{L}_n$. By the inductive
hypothesis, $\mathfrak{A}$ is jointly order-reflecting. 
 Moreover, for every $k$-ary propositional operator $p$ and every 
 $\varphi_1,\dots, \varphi_k\in\mathcal{L}_n$, we have that 
 $p(\varphi_1,\dots, \varphi_k)\in\mathcal{L}_n$ whence $\mathfrak{A}$ 
 is closed under propositional operators in $\mathcal{O}$. By depth-1 separation,
 it follows that the family 
\[
\Lambda(\mathfrak{A})= \{\llbracket L\rrbracket(\llbracket \varphi\rrbracket)\mid L\in\Lambda, \varphi\in\mathcal{L}_n\}
\]
is jointly order-reflecting. In other words, the family of evaluations of 
depth-($n+1$) formulae is jointly order-reflecting, as desired.

Finally, we will prove that $\mathcal{L}$ is expressive. Let $x,y$ be
states in $G$-coalgebras $\gamma\colon X\rightarrow GX$ and
$\delta\colon Y\rightarrow GY$ such that
$\llbracket\varphi\rrbracket_{\gamma}(x)\leq\llbracket\varphi\rrbracket_{\delta}(y)$
for all $\varphi\in\mathcal{L}$. We must prove that
$M_n!\cdot\gamma^{n}(x)\leq M_n!\cdot\delta^{n}(y)$ for all
$n\in\omega$. By 
assumption, we know that 
\begin{align*}
  \llbracket\varphi\rrbracket(M_n!\cdot\gamma^{(n)}(x))
  &= \llbracket\varphi\rrbracket_{\gamma}(x)\\
  &\leq \llbracket\varphi\rrbracket_{\delta}(y)\\
  &= \llbracket\varphi\rrbracket(M_n!\cdot\delta^{(n)}(y))
\end{align*}
for all $\varphi\in\mathcal{L}$ of uniform depth $n$ and for all $n\in\omega$, and we 
have just shown that the family $\llbracket\varphi\rrbracket\colon M_n1\rightarrow\Omega$ of 
evaluations of such formulae is jointly order-reflecting. Hence $x$ is an $(\alpha,\M)$-refinement of $y$,
as desired.
\end{proof}

%\subsection*{Proof of \autoref{thn:coalg-sim-expr}}
\begin{proof}[Proof of \autoref{thn:coalg-sim-expr}]
Depth-0 separation is trivial, since $M_01=1$. For depth-1 separation,
we note first that for finitary~$G$, it suffices to check the
condition for finite $A_0$. But for finite~$A_0$ and a given jointly
order-reflecting set $\mathfrak A$ of monotone maps
$A_0\to\mathbbm 2$, \emph{every} monotone $h:A_0\to\mathbbm 2$ can be
written as a finite join of finite meets of elements of $\mathfrak A$,
so depth-1 separation is immediate from the separation condition
assumed in the theorem.
\end{proof}

%\subsection*{Proof of \autoref{prop:simulation}}
\begin{proof}[Proof of \autoref{prop:simulation}]
\noindent Let $A$ be a canonical $M_1$-algebra with carriers $A_0, A_1$ and suppose that~$\mathfrak{A}$ 
is a jointly order-reflecting family of $M_0$-morphisms $A_0\rightarrow\mathbbm{2}$ closed under 
conjunction. 

Now, suppose that $x\not\leq y$ in $A_1$. Since $A$ is canonical, the map $a^{10}\colon \pow_{\omega}^{\down} (\A\times A_0)\rightarrow A_1$ 
is surjective, so $x,y$ have the form $x=a^{10}(\down \{(a_1, x_1),\dots, (a_n, x_n)\})$ and $y=a^{10}(\down \{(b_1, y_1),\dots, (b_n, y_n)\})$ for some $(a_i, x_i), (b_j, y_j)\in\A\times A_0$. By monotonicity of $a$, there exists $i\leq n$ such that $x_i\not\leq y_j$ for all $j$ such that $a_i=b_j$. As $\mathfrak{A}$ is jointly order-reflecting, it follows that for each~$j$ there is $f_j\in\mathfrak{A}$ such that $f_j(x_i)=\top$ and $f_j(y_j)=\bot$. Indeed, as $x_i\not\leq y_j$ it follows that there exists $f_j\in\mathfrak{A}$ such that $f_j(x_i)\not\leq f_j(y_j)$ in $\mathbbm{2}$, and this holds if and only if $f_j(x_i)=\top$ and $f_j(y_j)=\bot$. Now take $f\in\mathfrak{A}$ to be the conjunction of all~$f_j$.

Recall that the modal operator $\llbracket\lozenge_{a_i}\rrbracket\colon\pow_{\omega}^{\down} (\A\times\mathbbm{2})\to\mathbbm{2}$ is defined by $\llbracket\lozenge_{a_i}\rrbracket(S)=\top$ iff $(a_i, \top)\in S$, and $\llbracket \lozenge_{a_i}\rrbracket(f)$ is defined by commutation of the diagram
\[
\begin{tikzcd}[column sep = 40]
\pow^{\down}_{\omega}(\A\times A_0) \arrow[r, "\pow_{\omega}^{\down} (\A\times f)"] \arrow[d, "a^{10}"'] & \pow^{\down}_{\omega}(\A\times\mathbbm{2}) \arrow[d, "\llbracket\lozenge_{a_i}\rrbracket"] \\
A_1 \arrow[r, "\llbracket \lozenge_{a_i}\rrbracket(f)"]    & \mathbbm{2}                                         
\end{tikzcd}
\]
(an instance of Diagram~\eqref{Di:modality}). Thus,
$\llbracket\lozenge_{a_i}\rrbracket(f)(x)=\llbracket\lozenge_{a_i}\rrbracket\down (\{(a_1,
f(x_1)),\dots$, $(a_n, f(x_n))\}) = \top$, and similarly
$\llbracket\lozenge_{a_i}\rrbracket(f)(y)=\bot$,
i.e.~$\llbracket\lozenge_{a_i}\rrbracket(f)(x)\not\le\llbracket\lozenge_{a_i}\rrbracket(f)(y)$,
as required.
\end{proof}

\end{document}